\documentclass[11pt]{article}

\usepackage{xifthen}
\usepackage{amsmath, amssymb} 
\usepackage{bbm}
\usepackage{mathtools}
\usepackage{cite}
\usepackage{url}
\usepackage{appendix}
\usepackage{graphicx}
\usepackage{epstopdf}
\usepackage{setspace}
\usepackage{enumerate}
\usepackage{color}
\usepackage{xcolor}
\usepackage{algorithm}
\usepackage[noend]{algpseudocode}
\usepackage{multirow}
\usepackage{paralist}
\usepackage{times}
\usepackage{tabto}
\usepackage{varwidth} 
\usepackage{caption}
\usepackage{xspace}
\usepackage{mdwlist}
\usepackage{tikz}
\usepackage{arrayjobx}
\usepackage{needspace} 
\usepackage{bookmark}

\usepackage{hyperref}
\hypersetup{
colorlinks=true,
pdftitle={Tight Bounds for MIS in Multichannel Radio Networks},
pdfdisplaydoctitle=true,
}

\newboolean{hasLNCSformat}
\setboolean{hasLNCSformat}{false}

\ifthenelse{\boolean{hasLNCSformat}}
{
  \usepackage{cleveref} 
  \newenvironment{proofsketch}{\begin{proof}[Proof Sketch]}{\end{proof}}
  \newcommand{\placeNovelResearchContribution}{\textbf{Novel Research Contribution}\\}
  \newcommand{\separateAbstractFromBodyLNCS}{}
  \newcommand{\optionalproofparamLNCS}{}
  \newcommand{\qedLNCS}{\qed}
}{
  \usepackage{amsthm}
  \usepackage{times}
  \usepackage[compact]{titlesec}
  \usepackage{cleveref} 

  \newtheorem{theorem}{Theorem}[section]
  \newtheorem{lemma}[theorem]{Lemma}
  \newtheorem{corollary}[theorem]{Corollary}
  \newtheorem{property}{Property}
  
  \newtheorem{definition}[theorem]{Definition}
  \newenvironment{proofsketch}{\begin{proof}[ Sketch]}{\end{proof}}
  \newcommand{\keywords}[1]{\textbf{Keywords: }#1}
  \newcommand{\placeNovelResearchContribution}{}
  \newcommand{\separateAbstractFromBodyLNCS}{\pagebreak}
  \newcommand{\optionalproofparamLNCS}{Proof }
  \newcommand{\qedLNCS}{}
}

\usepackage[letterpaper,margin=1in]{geometry}
\usepackage{fullpage}
\newboolean{swapheraldclaim}
\setboolean{swapheraldclaim}{true}
\newboolean{swaphandshake}
\setboolean{swaphandshake}{false}

\newcommand{\placeStudentPaperAwardEligibility}[2]{
  Eligible for the best student paper award; #1
  \ifthenelse{\equal{#2}{singular}}{is a full-time student.}{are full-time students.}
}

\newtheorem{myclaim}[theorem]{Claim}

\newcommand{\sebastian}[1]{ [[[ \textcolor{green}{\bf Seb:} {\em #1} ]]]}

\newcommand{\logstar}[1]{\log^{\ast}{#1}}

\newcommand{\algoMIS}{\textsc{HeraldMIS}\xspace}
\newcommand{\algoDF}{\textsc{DFilter}\xspace}
\newcommand{\algoActiveState}{\textsc{HeraldProtocol}\xspace}
\newcommand{\algoHandshake}{\textsc{Handshake}\xspace}
\newcommand{\algoRBG}{\textsc{RedBlueGame}\xspace}
\newcommand{\algoMISstate}{\textsc{Dominator}\xspace}

\newcommand{\TODO}[1]{ {\color{blue} TODO: #1 }}

\newcommand{\hide}[1]{}
\newcommand{\mydots}{\dots}

\newcommand{\set}[1]{\ensuremath{\left\{#1\right\}}}
\newcommand{\brackets}[1]{\ensuremath{\left(#1\right)}}
\newcommand{\bigbrackets}[1]{\ensuremath{\big(#1\big)}}
\newcommand{\Bigbrackets}[1]{\ensuremath{\Big(#1\Big)}}
\newcommand{\bb}[1]{\bigbrackets{#1}}
\newcommand{\BB}[1]{\Bigbrackets{#1}}

\newcommand{\dcup}{\ensuremath{\mathaccent\cdot\cup}}
\newcommand{\dotcup}{\dcup}
\DeclareFontFamily{U} {MnSymbolC}{}

\DeclareFontShape{U}{MnSymbolC}{m}{n}{
  <-6> MnSymbolC5
  <6-7> MnSymbolC6
  <7-8> MnSymbolC7
  <8-9> MnSymbolC8
  <9-10> MnSymbolC9
  <10-12> MnSymbolC10
  <12-> MnSymbolC12}{}
\DeclareFontShape{U}{MnSymbolC}{b}{n}{
  <-6> MnSymbolC-Bold5
  <6-7> MnSymbolC-Bold6
  <7-8> MnSymbolC-Bold7
  <8-9> MnSymbolC-Bold8
  <9-10> MnSymbolC-Bold9
  <10-12> MnSymbolC-Bold10
  <12-> MnSymbolC-Bold12}{}

\DeclareSymbolFont{MnSyC} {U} {MnSymbolC}{m}{n}

\DeclareMathSymbol{\cupdot}{\mathbin}{MnSyC}{60}
\DeclareFontFamily{U} {MnSymbolF}{}

\DeclareFontShape{U}{MnSymbolF}{m}{n}{
  <-6> MnSymbolF5
  <6-7> MnSymbolF6
  <7-8> MnSymbolF7
  <8-9> MnSymbolF8
  <9-10> MnSymbolF9
  <10-12> MnSymbolF10
  <12-> MnSymbolF12}{}
\DeclareFontShape{U}{MnSymbolF}{b}{n}{
  <-6> MnSymbolF-Bold5
  <6-7> MnSymbolF-Bold6
  <7-8> MnSymbolF-Bold7
  <8-9> MnSymbolF-Bold8
  <9-10> MnSymbolF-Bold9
  <10-12> MnSymbolF-Bold10
  <12-> MnSymbolF-Bold12}{}

\DeclareSymbolFont{SymbolF} {U} {MnSymbolF}{m}{n}

\DeclareMathSymbol{\dbigcupdot}{\mathop}{SymbolF}{35}
\DeclareMathSymbol{\tbigcupdot}{\mathop}{SymbolF}{34}
\def\bigcupdot{\mathchoice{\dbigcupdot}{\tbigcupdot}{\tbigcupdot}{\tbigcupdot}}

\newcommand{\bigdotcup}{\bigcupdot}

\newcommand{\operatorleftalignbalance}{0.4em}
\newcommand{\sumleft}[2]{
  \mathop{\oalign{$\hspace{\operatorleftalignbalance}\sum \,\, #2$\hfill\cr
    $\begin{subarray}{l}\hspace{-\operatorleftalignbalance}#1\end{subarray}$\hfill\cr}}
}
\newcommand{\prodleft}[2]{
  \mathop{\oalign{$\hspace{\operatorleftalignbalance}\prod \,\, #2$\hfill\cr
    $\begin{subarray}{l}\hspace{-\operatorleftalignbalance}#1\end{subarray}$\hfill\cr}}
}

\newcommand{\whp}[1][\empty]{\ensuremath{\text{w.h.p.}\ifthenelse{\equal{#1}{\empty}}{}{(#1)}}\xspace}
\newcommand{\Whp}[1][\empty]{\ensuremath{\text{W.h.p.}\ifthenelse{\equal{#1}{\empty}}{}{(#1)}}\xspace}
\newcommand{\wcp}{\text{w.c.p.}\xspace}
\newcommand{\Wcp}{\text{W.c.p.}\xspace}

\newcommand{\pseudocode}{pseudo-code\xspace}
\newcommand{\Pseudocode}{Pseudo-code\xspace}
\newcommand{\eps}{\varepsilon}
\renewcommand{\Pr}{\mathbf{P}}

\newcommand{\Freq}{\ensuremath{\mathcal{F}}}
\newcommand{\freq}{\ensuremath{\mathcal{F}}}

\newcommand{\eEvent}{\ensuremath{\mathcal{E}}}
\newcommand{\mEvent}{\ensuremath{\mathcal{M}}}

\newcommand{\iEvent}{\ensuremath{\mathcal{I}}}

\newcommand{\gbg}{bounded independence graph\xspace}
\newcommand{\gbgs}{bounded independence graphs\xspace}

\newcommand{\multichannel}{multichannel\xspace}
\newcommand{\singlechannel}{single channel\xspace}

\newcommand{\singlehop}{single-hop\xspace}

\newcommand{\Hchannel}{\ensuremath{\mathcal{H}}\xspace}
\newcommand{\Gchannel}{\ensuremath{\mathcal{G}}\xspace}

\newcommand{\Achannel}{\ensuremath{\mathcal{A}}}

\newcommand{\Schannel}{\ensuremath{\mathcal{S}}}

\newcommand{\Rchannel}{\ensuremath{\mathcal{R}}}
\newcommand{\Dchannel}{\ensuremath{\mathcal{D}}}

\newcommand{\nac}{{n_\Achannel}}
\newcommand{\ndc}{{n_\Dchannel}}
\newcommand{\nrc}{{n_\Rchannel}}
\newcommand{\nsc}{{n_\Schannel}}

\newcommand{\Wstate}{\ensuremath{\mathbb{W}}\xspace}
\newcommand{\Astate}{\ensuremath{\mathbb{A}}\xspace}

\newcommand{\Sstate}{\ensuremath{\mathbb{S}}\xspace}
\newcommand{\Mstate}{\ensuremath{\mathbb{M}}\xspace}
\newcommand{\Lstate}{\ensuremath{\mathbb{L}}\xspace}
\newcommand{\Dstate}{\ensuremath{\mathbb{D}}\xspace}
\newcommand{\Estate}{\ensuremath{\mathbb{E}}\xspace}

\newcommand{\Hstate}{\ensuremath{\mathbb{H}}\xspace}

\newcommand{\bigO}{\ensuremath{O}}
\newcommand{\softO}{\ensuremath{\tilde{O}}}
\newcommand{\tildebigO}{\softO}

\newcommand{\activity}{\gamma}
\newcommand{\Activity}{\Gamma}
\newcommand{\tthreshold}{\tau}
\newcommand{\cthreshold}{\Delta}
\newcommand{\tth}{\tthreshold}

\newcommand{\ttk}{\kappa}
\newcommand{\Vhf}{\ensuremath{V_{\text{hf}}}\xspace}
\newcommand{\ttkdelta}{{\ttk_\Delta}}
\newcommand{\actc}{\bar{c}}
\newcommand{\actm}{\bar{m}}
\newcommand{\increaseactivity}{\sigma_\oplus}
\newcommand{\decreaseactivity}{\sigma_\ominus}
\newcommand{\rbg}{red-blue game\xspace}
\newcommand{\rbgs}{red-blue games\xspace}
\newcommand{\rbp}{red-blue protocol\xspace}
\newcommand{\RBP}{Red-Blue Protocol\xspace}
\newcommand{\lhp}{leader-herald pair\xspace}
\newcommand{\lhps}{leader-herald pairs\xspace}

\newcommand{\hfilter}{herald filter\xspace}
\newcommand{\dfilter}{decay filter\xspace}
\newcommand{\Hfilter}{Herald Filter\xspace}
\newcommand{\Dfilter}{Decay Filter\xspace}
\newcommand{\runtime}{runtime\xspace}
\newcommand{\Runtime}{Runtime\xspace}

\newcommand{\false}{\bf{false}}
\newcommand{\true}{\bf{true}}

\newcommand{\msg}{\textit{msg}}

\newcommand{\state}{\textit{state}}
\newcommand{\roundcount}{\textit{count}}
\newcommand{\phase}{\textit{phase}}

\newcommand{\textmax}{\textit{max}}
\newcommand{\textmin}{\textit{min}}
\newcommand{\textred}{\textit{red}}
\newcommand{\textblue}{\textit{blue}}
\newcommand{\textitcolor}{\textit{color}}
\newcommand{\low}{\textit{low}}

\newcommand{\leader}{\textit{leader}}
\newcommand{\ID}{\textit{ID}\xspace}
\newcommand{\IDs}{\textit{IDs}\xspace}
\newcommand{\lonely}{\textit{lonely}}
\newcommand{\progress}{\textit{progress}}
\newcommand{\redblue}{\textit{red-blue}}
\newcommand{\notif}{\textit{notification}}
\newcommand{\itruntime}{\textit{\runtime}}
\newcommand{\game}{\textit{game}}
\newcommand{\handshake}{\textit{handshake}}
\newcommand{\success}{\textit{succ}}
\newcommand{\fail}{\textit{fail}}
\newcommand{\rendezvous}{\textit{meet}}

\newcommand{\enforce}{\textit{enforce}}

\newcommand{\RT}{\ensuremath{\log^2 n / \Freq + \log n }}
\newcommand{\bigORT}{\ensuremath{\bigO(\log^2 n / \Freq + \log n)}}
\newcommand{\misRTstrong}{\RT}
\newcommand{\loglogn}{\log\log n}

\newcommand{\AAC}{\Activity_\Astate^\circ}

\DeclareMathOperator{\E}{\mathbb{E}}
\DeclareMathOperator{\polylog}{\ensuremath{\mathrm{polylog}}}
\DeclareMathOperator{\polyloglog}{\ensuremath{\mathrm{polyloglog}}}

\newcommand{\reduceSpaceAroundEquation}[1]{\vspace{-2mm}#1\vspace{-0.0mm}}

\newcommand{\paraclose}[1]{\vspace{0.2em}\noindent\textbf{#1}~}

\algnewcommand\algorithmicswitch{\textbf{switch}}
\algnewcommand\algorithmiccase{\textbf{case}}

\algdef{SE}[SWITCH]{Switch}{EndSwitch}[1]{\algorithmicswitch\ #1\ \algorithmicdo}{\algorithmicend\ \algorithmicswitch}%
\algdef{SE}[CASE]{Case}{EndCase}[1]{\algorithmiccase\ #1}{\algorithmicend\ \algorithmiccase}%
\algtext*{EndSwitch}%
\algtext*{EndCase}%

\algnewcommand\algorithmicwithprob{\textbf{with probability}}
\algnewcommand\algorithmicotherwise{\textbf{otherwise}}

\algdef{SE}[WithProb]{WithProb}{EndWithProb}[1]{\algorithmicwithprob\ #1\ \algorithmicdo}{\algorithmicend\ \algorithmicwithprob}%
\algdef{Ce}[Otherwise]{WithProb}{Otherwise}{EndWithProb}
  {\algorithmicotherwise}%
\algtext*{EndWithProb}%
\algtext*{EndOtherwise}%

\begin{document}

\date{}
\title{Tight Bounds for MIS in Multichannel Radio Networks}

\ifthenelse{\boolean{hasLNCSformat}}{
  \author{
    Sebastian Daum
    \and %
    Fabian Kuhn
  }
  \institute{
    University of Freiburg, Germany
    \\
    \email{
      \{sdaum,kuhn\}@cs.uni-freiburg.de
    }
  }
}{
  \author{
    Sebastian Daum\\
    Dept.\ of Comp.\ Science\\
    U.\ of Freiburg, Germany\\
    {\small\texttt{daum.sebastian@gmail.com}}
    \and 
    Fabian Kuhn\\
    Dept.\ of Comp.\ Science\\
    U.\ of Freiburg, Germany\\
    {\small \texttt{kuhn@cs.uni-freiburg.de}}
  }
}

\maketitle
\begin{abstract}
  In \cite{podc2013} an algorithm has been presented that computes a
  maximal independent set (MIS) within $\bigO(\log^2 n/\Freq+\log n
  \polyloglog n)$ rounds in an $n$-node multichannel radio network
  with $\Freq$ communication channels. The paper uses a multichannel
  variant of the standard graph-based radio network model without
  collision detection and it assumes that the network graph is a
  polynomially bounded independence graph (BIG), a natural combinatorial
  generalization of well-known geographic families. The upper bound of
  \cite{podc2013} is known to be optimal up to a $\polyloglog$ factor.

  In this paper, we adapt algorithm and analysis to improve the result of
  \cite{podc2013} in two ways. Mainly, we get rid of the $\polyloglog$
  factor in the \runtime and we thus obtain an asymptotically
  optimal multichannel radio network MIS algorithm. In addition, our
  new analysis allows to generalize the class of graphs from those
  with polynomially bounded local independence to graphs where the
  local independence is bounded by an arbitrary function of the
  neighborhood radius.




\end{abstract}


\setcounter{page}{0}
\thispagestyle{empty}

\separateAbstractFromBodyLNCS

\section{Introduction}
\label{sec:introduction}

In recent years there has been an increased interest in algorithms for
\emph{shared spectrum networks} \cite{sherman:2008}. Nowadays, most
modern wireless communication networks feature a multitude of
communication frequencies \cite{Bluetooth, 802.11, Zigbee}%
\footnote{For example, the IEEE 802.11 WLAN standard provides a
  channel spectrum of up to 200 (partially overlapping) channels and
  Bluetooth specifies 79 usable channels.}---and we can certainly
expect this trend to continue.

In the light of this development, in the present paper, we settle the
question of determining the \emph{optimal} asymptotic time complexity
of computing a maximal independent set (MIS) in the \multichannel
variant of the classic radio network model first introduced in
\cite{chlamtac:1985,bgi}. The task of constructing an MIS is one of
the best studied problems in the area of large-scale wireless
networks. On the one hand this is due to the fact that MIS (together
with coloring problems) is one of the key problems to study the
problem of symmetry breaking in large, decentralized systems. On the
other hand an MIS provides a simple local clustering of the graph,
which can be used as a building block for computing more enhanced
organization structures in these networks such as, e.g., a
communication backbone based on a connected dominating set
\cite{censor:2011, ephremides87, alzoubi:1}. This is specifically
relevant in the context of wireless mobile ad hoc networks or sensor
networks, in which devices cannot rely on already existing
infrastructure to organize themselves---devices need to compute a
meaningful structure by themselves to coordinate their interactions.

\paraclose{Related Work.}
In \cite{ABI86,Luby86} Alon et al.\ and Luby presented a simple and
efficient randomized parallel algorithm to compute and MIS of a
general graph. It is straightforward to a standard distributed message
passing model and as a consequence, the algorithm soon became an
archetype for many distributed MIS algorithms also in other---usually
more limiting---settings.  The model we assume here is an extension to
the radio network model, for which an MIS algorithm with \runtime
$\bigO(\log^2 n)$ has been presented in \cite{moscibroda05} for the
class of unit disk graphs (UDGs). This algorithm has been proven to be
asymptotically optimal \cite{DKN:disc:2012} even for more basic
version of the problem known as the wake-up problem in single-hop
radio networks. While the UDG restriction is well-known and
popular, a more general variant known as growth-bounded graphs or
\gbgs that contains UDGs has also become the focus of quite some
research, e.g., \cite{KMNW05,schmid06,schneider08}. In particular, in
\cite{schneider08} it is shown that an MIS and many related structures
can be compute in (asymptotically optimal) $\bigO(\logstar n)$ rounds
in such graphs.

Much of the early algorithmic research on \multichannel radio networks
has focused on networks with faults assuming a malicious adversary
that can jam up to $t$ of the $\freq$ available channels
\cite{dolev:2011, dolev2008, dolev2007, dolev:2009b, podc09,
  gilbert2009b, strasser2008, strasser2009, podc2012}. In addition,
for fault-free networks, in \cite{newport:2014} a series of lower
bound proofs have been provided, which show that
$\Omega(\log^2 n/\freq + \log n)$ rounds are needed to solve any
problem which requires communication. In \cite{podc2012} a new
technique (called heralding) to deal with congestion in \multichannel
radio networks has been established to solve leader election in
\singlehop networks in time asymptotically matching the lower bound of
\cite{newport:2014}. This technique has been extended in
\cite{disc2012} and \cite{podc2013} to solve the problems of computing
an approximate minimum dominating set and an MIS, respectively. Our
research here is based on this work and in particular on the MIS
algorithm of \cite{podc2013}.

\paraclose{Contributions.}
In radio network models, in almost all cases a restriction to the
underlying graph model is being assumed. One of the most general ones
are so-called $\alpha$-\gbgs, where $\alpha(r)$ is a function that
limits the size of a \emph{maximum} independent set in any
$r$-neighborhood of the given graph. The MIS algorithm from
\cite{podc2013} solves the MIS problem in time
$\bigO(\log^2 n/\freq + \log n (\loglogn)^d)$ in such graphs for which
$\alpha$ is bounded by polynomial of degree $d$. Here we get rid of
the $\polyloglog$ factor and thus show how to close the gap to the
lower bound from \cite{newport:2014}. At the same time, we remove any
restriction on the function $\alpha$. We do so by adjusting the
algorithm from \cite{podc2013}---and though the change in the
algorithm is relatively small, it leads to a significantly more
involved analysis.

\section{Preliminaries}
\label{sec:preliminaries}


This paper bases strongly on \cite{podc2013} and \cite{tr:podc2013}, the former being the proceedings version and the latter the complete version.
However, we try to be as self-contained as possible.


\paraclose{Radio Network Model.}
We model the network as an $n$-node graph $G=(V,E)$. We assume that
$n$ or a polynomial upper bound on $n$ is known by all nodes. Nodes
start out dormant and are awakened/activated by an adversary. While
nodes do not have access to a global clock, communication is assumed
to happen in synchronous time slots (rounds). The network comprises
$\freq$ communication channels. In each round each node can choose to
operate on one channel, either by listening or broadcasting. A node
that broadcasts does not receive any message in that round, and its
signal reaches all neighbors that operate on the same channel. A node
$v$ listeing on some channel can decode an incoming message iff in the
given round, exactly one of its neighbors broadcasts on the same
channel. If two or more neighbors broadcast, their signals collide at
$v$ and $v$ receives nothing, unable to detect this collision. A node
can only operate on one channel in each round and therefore it does
not learn anything about events on other channels.

\paraclose{Notation.}
In our algorithm 
all nodes move between a finite set of states: $\Wstate$ --
\emph{waiting}, $\Dstate$ -- \emph{decay}, $\Astate$ -- \emph{active},
$\Hstate'$ -- \emph{herald candidate}, $\Hstate$ -- \emph{herald},
$\Lstate'$ -- \emph{leader candidate}, $\Lstate$ -- \emph{leader},
$\Mstate$ -- \emph{MIS node}, $\Estate$ --
\emph{eliminated/dominated}. We overload this notation to also indicate the
set of nodes being in a certain state, e.g.,
$\Astate \coloneqq \set{ v \in V: v \text{ is in state } \Astate }$. Since
nodes change their states, in case of ambiguity, we write $\Astate_r$
for the set of active nodes in round~$r$. State changes always happen
between rounds. We define
$\Vhf \coloneqq \Astate \cup \Hstate' \cup \Lstate' \cup \Hstate \cup
\Lstate$ as the nodes in the so-called \hfilter.

We use $N(v)$ to denote the neighbors of $v$ in $G$, while we use $N^k(v)$ to denote the set of nodes in distance at most $k$ from $v$, including $v$ itself.
We also often write $N_\Sstate(u)$ or $N^k_\Sstate(u)$ to abbreviate $N(u) \cap \Sstate$ or $N^k(u) \cap \Sstate$ respectively, for some state $\Sstate$. For $S \subseteq V$ we let $N(S) \coloneqq \bigcup_{v \in S} N(v)$. 
We call a node $v$ \emph{alone} or \emph{lonely}, if $N_{\Vhf \cup \Mstate}(u) =\emptyset$.

We say that an event $A$ happens \emph{with high probability (\whp)},
\emph{with decent probability}, or \emph{with constant probability
  (\wcp)}, if it happens with probability at least $1-n^{-c}$, $1-\log^{-c} n$, or $\Omega(1)$, respectively, where $c$ is a constant that can be chosen arbitrarily large.
By $x \gg y$ we denote that $x > cy$ for a sufficiently large $c>1$.

\paraclose{Bounded Independence.} 
In addition to the communication characteristics of the network, we
require the network graph to be a \gbg (BIG) \cite{KMNW05,schmid06}. A
graph $G$ is called an \emph{$\alpha$-bounded independence graph} with
\emph{independence function}
$\alpha: \mathbb{N} \rightarrow \mathbb{N}$, if for every node $v$. no
independent set $S$ of the subgraph of $G$ induced by $N^d(v)$ exceeds
cardinality $\alpha(d)$.
Note that in particular, $\alpha$ does not depend on $n$ and thus for
every fixed $d$, $\alpha(d)$ is a constant. In \cite{podc2013},
$\alpha$ is required to be a polynomial, whereas in this paper, we put
no restrictions on $\alpha$. It can easily be verified that one can
always upper bound the largest independent set of the subgraph induced
by $N^d(v)$ by $\alpha(2)^d$ and thus any independence function is
always upper bounded by some exponential function. For simplicity we
define a constant $\alpha \coloneqq \alpha(2)$ and we assume that all
nodes know the value of $\alpha$.

\paraclose{Number of Channels.}
We assume that $\Freq=\omega(1)$ as otherwise \singlechannel
algorithms achieve the same asymptotic time bounds. For
$\Freq=\omega(\log n)$ we only actually use $\Theta(\log n)$ channels
since more channels do not lead to an additional asymptotic
advantage. For ease of exposition we assume $\Freq=\Omega(\loglogn)$
and refer to \cite{tr:podc2013} for an explanation of how to adapt the
algorithm for the case $\Freq=o(\loglogn)$.

\paraclose{Maximal Independent Set.} 
We say an algorithm computes an MIS in time $T$, if the following
properties hold \whp for each round $r$ and node $v$ (waking up in
round $r_v$):
\begin{compactenum}[(P1)]
  \item $v$ declares itself as either \emph{dominating} ($\in \Mstate$) or \emph{dominated} ($\in \Estate$) before round $r_v+ T$ and this decision is permanent.
  \item If $v$ is \emph{dominated} in round $r$, then $N(v) \cap \Mstate_r \neq \emptyset$.
  \item If $v$ is \emph{dominating} in round $r$, then $N(v) \cap \Mstate_r = \emptyset$.
\end{compactenum}

\hide{
\paraclose{Loneliness Detection} 
We say that an algorithm solves LD in time $T$, if for each node $v \in V$ that wakes up in round $r_v$ the following
two properties hold \whp: 
\begin{compactenum}[(P1)]
  \item If $v$ has an awake neighbor $w$ in round $r_v$, $v$ declares itself as \emph{non-lonely} before round $r_v+T$.
  \item Each awake neighbor $w$ (awake in round $r_v$) of $v$ declares itself as \emph{non-lonely} before round $r_v+T$.
\end{compactenum}
}


\section{Algorithm Description}


\begin{algorithm}[!ht] 
  \caption{HeraldMIS---core structure}
  \label{algo:2:coreMIS}
  \small
  \begin{tabular}{@{}ll}
    Input: &$\increaseactivity$, $\decreaseactivity$, $\cthreshold_\textmax$, $\pi_\ell$, $\alpha$, $n$, $\ndc=\Theta(\freq)$, $\nac=\Theta(\loglogn)$, $\nrc=\Theta(\alpha(2))$, \\
    & $\tth_\Wstate=\Theta(\log n)$, $\tth_\Dstate=\Theta(\log n/\freq)$, $\tthreshold_\lonely=\Theta(\log^2 n/\freq+\log n)$, $\tthreshold_\redblue=\Theta(\log n)$\\
  States: &\Wstate---waiting, \Dstate---decay,
  \Mstate---MIS node, \Estate---eliminated\\
  &\Astate---active, \Lstate/$\Lstate'$---leader (candidate), \Hstate/$\Hstate'$---herald (candidate)\\
  Channels: &$\Rchannel_1, \dots, \Rchannel_\nrc$---report, $\Dchannel_1,\dots,\Dchannel_\ndc$---decay,
  \\
  &
  $\Achannel_1,\dots,\Achannel_\nac$---herald, $\Hchannel$---handshake, $\Gchannel$---\rbg
  \end{tabular}

  \begin{algorithmic}[1]
    \State $\roundcount \gets 0$; 
    $\state \gets \Wstate$; 
    $\activity \gets \perp$;
    $\lonely \gets \perp$;
    $\activity_\textmin \gets \log^{-24} n$
    \While{$\state \neq \Estate$}
      \State $\roundcount \gets \roundcount +1$
      \State $\lonely \gets \lonely + 1$
      \State $\activity \gets \min \set{\activity \cdot \increaseactivity, 1/2}$
      \State uniformly at random pick $q \in [0,1)$, $j \in \set{1,\dots,\ndc}$ and $k \in \{1, \dots, \nrc\}$
      \Switch{$\state$}
        \Case{\Wstate or \Dstate:}
          run \algoDF
          \Comment{stage 1---\dfilter}
        \EndCase
        \Case{\Astate:}
          run \algoActiveState
          \Comment{stage 2---\hfilter}
        \EndCase
        \Case {$\Hstate'$ or $\Lstate'$:}
          run \algoHandshake
        \EndCase
        \Case {$\Hstate$ or $\Lstate$:}
          run \algoRBG
        \EndCase
        \Case {$\Mstate$:}
          run \algoMISstate
          \Comment{stage 3---MIS node}
        \EndCase
      \EndSwitch
      \If{$\lonely = \tthreshold_\lonely$}
        \State $\state \gets \Mstate$
      \EndIf
    \EndWhile 
    \State\textbf{endWhile}
  \end{algorithmic}
\end{algorithm}


  \begin{theorem}
    \label{thm:mis:maintheorem}
    Algorithm \algoMIS solves MIS within $\bigO(\misRTstrong)$ rounds.
  \end{theorem}

\hide{
  The algorithm works almost the same as the one in \cite{tr:podc2013}. 
  A detailed algorithm description in \pseudocode is provided in Appendix \ref{sec:pseudo}. Each node first executes Algorithm \ref{algo:decay}, also called the \emph{\dfilter}; this part is unchanged in this paper. In case it manages to leave the \dfilter, it executes Algorithm \ref{algo:heraldCore}---the \emph{\hfilter}. Even while Algorithm \ref{algo:heraldCore} is the core algorithm, it ultimately is a subroutine to Algorithm \ref{algo:decay}.
  
  \begin{theorem}
    \label{thm:mis:maintheorem}
    Algorithm \ref{algo:decay}, executed at each node once it wakes up, solves the MIS problem as stated in Section \ref{sec:preliminaries} within $\bigORT$ rounds.
  \end{theorem}
}

We first give a short summary of how the algorithm works, which includes a recap of results from \cite{tr:podc2013}.
The algorithm is divided into three stages, the \dfilter (states \Wstate and \Dstate), the \hfilter (states \Astate, $\Lstate'$, $\Hstate'$, \Lstate and \Hstate), and decided nodes (states \Mstate and \Estate). Nodes move forward within those stages---possibly omitting the \hfilter---but never backwards. The \dfilter is a powerful tool (which we use as a black box) that provides that over the full \runtime of the algorithm the degree of the graph induced  by nodes in the \hfilter is bounded by $\bigO(\log^3 n)$. In short, nodes first only listen for a while (\Wstate), then they start broadcasting on one out of $\Theta(\freq)$ random channels with probability $1/n$ (\Dstate), doubling this probability every $\bigO(\log n/\freq)$ rounds. A node that broadcasts moves to the \hfilter and a node that receives a message restarts with \Wstate. The \dfilter has not changed and for a detailed analysis we refer to \cite{tr:podc2013}, while \pseudocode is given in Algorithm \ref{algo:decay}.

\begin{algorithm}[!ht] 
  \caption{\algoDF, run at process $v$} 
  \label{algo:decay}
  \small
  
  Input: $\freq$, $\ndc = \Theta(\freq)$, $\nrc\geq 3\alpha^2$, $\tth_\Wstate = \Theta(\log n)$, $\tth_\Dstate = \Theta(\log n/\freq)$\\
  States: \Wstate---waiting, \Dstate---decay \\
  Channels: $\Rchannel_1,\dots,\Rchannel_\nrc$---report, $\Dchannel_1, \dots, \Dchannel_\ndc$---decay 

  \begin{algorithmic}[1]
    \Statex 
    \State $\roundcount \gets 0$, $\state \gets \Wstate$
    \While{$\state \neq \Estate$}
      \State $\roundcount \gets \roundcount + 1$
      \State pick $i \in \set{1, \dots, \nrc}$, $k \in \set{1, \dots, \ndc}$ and $q \in [0,1)$ uniformly at random
      \Switch{$\state$}
        \Case{$\Wstate$}
          \State listen on channel $\Rchannel_i$
          \If{$\roundcount = \tth_\Wstate$}
            \State $\roundcount \gets 0$, $\state \gets \Dstate$, $\phase \gets 0$
          \EndIf
        \EndCase
        \Case{$\Dstate$}
          \label{line:mainbody-decay}
          \Switch{$q$}
            \Case{$q \in [0, 2^\phase/n)$}
              \State send $\msg=(\ID, \state)$ on $\Dchannel_k$
              \State exit \dfilter and enter \hfilter
              \label{line:exit-decay}
            \EndCase
            \Case{$q \in [2^\phase/n, 1/2)$}
              \State listen on $\Dchannel_k$
            \EndCase
            \Case{$q \in [1/2,1)$}
              \State listen on $\Rchannel_i$
            \EndCase
          \EndSwitch
          \If{$\roundcount = \tth_\Dstate$}
            \State $\roundcount \gets 0$, $\phase \gets \min\set{\phase + 1,\log n -2}$
            \vspace{2mm}
          \EndIf
        \EndCase
      \EndSwitch
    \EndWhile
    {\bf  Upon receiving a message $\msg = (\msg.\ID, \msg.\state)$}
    \vspace{2mm}
    \If{$\msg.\state = \Dstate$} \Comment{restart \dfilter}
      \State $\roundcount \gets 0$, $\state \gets \Wstate$ 
      \label{line:hear-decaymessage}
    \EndIf
    \If{$\msg.\state = \Mstate$}
      \State $\state \gets \Estate$
    \EndIf
  \end{algorithmic}
\end{algorithm}


Eliminated nodes (\Estate) know that they have a neighbor in the MIS and stop their protocol. MIS nodes (\Mstate) try to inform their neighborhood (eliminating them), but they also actively disrupt protocols in the \hfilter, causing them to fail; for more details confer \Cref{algo:MISstate}. Apart from this, there is no influence between nodes being in different stages.

\begin{algorithm}[!ht] 
\caption{\algoMISstate}
\label{algo:MISstate}
\small
\begin{algorithmic}[1]
    \If{$\enforce$}
        \State send $(\state, \ID)$ on $\Hchannel$
        \State $\enforce\gets\false$
    \Else
        \Switch{$q$}
        \Case{$q \in \big[0, \frac{1}{2}\big)$}
            \State send $(\state, \ID)$ on $\Hchannel$
            \State $\enforce\gets\false$
        \EndCase
        \Case{$q \in \big[\frac{1}{2}, \frac{3}{4}\big)$}
            \State send $(\state, \ID)$ on $\Gchannel$
            \State $\enforce\gets\true$
        \EndCase
        \Case{$q \in \big[\frac{3}{4}, 1\big)$}
            \State send $(\state, \ID)$ on $\Rchannel_k$
            \State $\enforce\gets\true$
        \EndCase
        \EndSwitch
    \EndIf
\end{algorithmic}
\end{algorithm}


The focus of this paper is almost exclusively on the \hfilter. It helps for understanding the complex algorithm to \emph{only} think of the graph that is induced by nodes in the \hfilter and to recall that its maximum degree is polylogarithmic in $n$.

The \hfilter is divided into three blocks, \emph{active state/herald protocol} (\Astate), \emph{handshake protocol} ($\Lstate'$ and $\Hstate'$) and \emph{\rbp} (\Lstate and \Hstate). The first block has the purpose of nodes trying to make contact with surrounding nodes. If this does indeed happen, both nodes engage in the handshake, which is only successful, if none of the two nodes neighbors any MIS node or a node in the third block. If the handshake succeeds, both nodes start a series of coin flipping games, with the sole purpose of ensuring that no two nodes, that became leaders (\Lstate) \emph{simultaneously}, can join the MIS.
The blocks that differ from the algorithm description in 
  \cite{tr:podc2013} 
are the active state and the \rbp, although changes in the latter are made to impact nodes in the active state.


  Ahead of all we want to mention that there are two ways for a node $v$ to join the MIS---either by waiting for a long time without hearing from any nearby node, or by successfully communicating with a node $u$ during the active state, \emph{teaming up} with it (as a leader-herald pair) and together passing through the handshake and the \rbp.
  The farther a pair of nodes advances in these blocks, the closer its leader is to become an MIS node.
We now recap and describe the behavior of a node $v$ in the \hfilter, i.e., $v \in \Vhf$, pointing out when changes to the original algorithm occur.

\paraclose{Loneliness.}
$v$
maintains a counter $\lonely$. Whenever $v$ hears from another node, it resets $\lonely$ to zero. If $\lonely$ ever exceeds $\tthreshold_\lonely = \Theta\bb{\misRTstrong}$, then $v$ \emph{assumes} that it is \emph{alone/lonely} in the \hfilter (i.e., $N_{\Vhf \cup \Mstate}(v) = \emptyset$) and joins the MIS---\whp, this action is safe, i.e., should $v$ not be alone, then the neighbors of $v$ are far from becoming MIS nodes themselves and $v$ has enough time to eliminate them.%
  \footnote{In \cite{tr:podc2013} there existed some component called \emph{loneliness support block}, operating on its own set of channels $\Schannel_1, \mydots, \Schannel_\nsc$; this block and its channels have been removed.}

\paraclose{Activity.}
Also, 
$v$
maintains an \emph{activity} value $\activity(v) \in \big[\activity_{\textmin}, 1/2\big]$, where $\activity_{\textmin}$, the initial value, is in $\Omega(1/\polylog n)$. $\activity$ governs the behavior of $v$ in \Astate, but all nodes in $\Vhf$ maintain this value. Nodes outside $\Vhf$ have zero activity.  $\activity(v)$ increases by a (small) constant factor $\increaseactivity>1$ each round, such that after $\Theta(\loglogn)$ rounds it would reach the maximum value $1/2$. However, whenever $v$ receives from a neighboring \emph{leader} or \emph{herald}, then $v$ reduces 
$\activity(v)$ by a (large) constant factor $\decreaseactivity \gg \increaseactivity$. This is a change to the original algorithm, where $\activity$ could only increase. 
The reason 
is the following. 
Leaders are nodes that likely become MIS nodes, and if they do then they eliminate their neighbors anyway. For safety reasons a leader $l$ needs to wait for $\Theta(\log n)$ rounds before it may join the MIS. During that time, if $l$'s neighbors keep high activity values, progress stagnation can occur in a $\delta'=\bigO(\loglogn)$ neighborhood of $l$, which is why in \cite{tr:podc2013} an $\alpha(\delta')=\bigO(\polyloglog n)$ speed loss had to be accepted. We show here that by reducing activities, this stagnation can be eliminated. At the same time, 'unjustified' reductions only cause 'minor damage' that can easily be mitigated.
This change is reflected in line \ref{line:algo:heraldActive:1} of \Cref{algo:activeState}.




\begin{algorithm}[!ht] 
\caption{\algoActiveState}
\label{algo:activeState}
\small
\begin{algorithmic}[1]
\State pick an $i \in \{1,\dots,n_\Achannel, \perp\}$ randomly with distribution $\Pr(i=\perp)=2^{-n_\Achannel}$ and $\Pr(i = j) = 2^{-j}$ 
\If{$i=\perp$}
    $q=1$
\EndIf
\Switch{$q$}
  \Case{$q \in [0, \pi_\ell\activity)$}
    \State
    listen on $\Achannel_i$
    \If{$\msg \neq \emptyset$}
      \State $\ID_\leader \gets \msg.\ID$; $\state \gets \Hstate'$; $\roundcount \gets 0$;
      $\handshake \gets \success$; $\lonely \gets 0$
    \EndIf
  \EndCase
  \Case{$q \in [\pi_\ell\activity, \activity)$}
    \State send $(\ID)$ on $\Achannel_i$
    \State
    $\state\gets\Lstate'$; $\roundcount \gets 0$;
    $\handshake \gets \success$;
  \EndCase
  \Case{$q \in [\activity,1]$}
    \State
    listen on $\Rchannel_k$
    \If{$\msg.\state = \Mstate$}
      \State $\state \gets \Estate$; $\activity \gets 0$
    \EndIf
    \If{$\msg.\state \in \set{\Lstate,\Hstate}$}
      \State $\activity \gets \max\set{ \activity/\decreaseactivity, \activity_{\textmin} }$; $\lonely \gets 0$
      \label{line:algo:heraldActive:1}
    \EndIf
  \EndCase
\EndSwitch
\end{algorithmic}
\end{algorithm}


\paraclose{Herald Protocol.}
Confer \Cref{algo:activeState} for detailed \pseudocode.
A node $v$ in the active state ($\Astate$) participates in the \emph{herald protocol} with probability $\activity(v) \in [\activity_\textmin,1/2]$, otherwise it tries to learn of nearby leaders, heralds or MIS nodes, by listening to one of \emph{constant} many \emph{report channels} $\Rchannel_1, \mydots, \Rchannel_{n_\Rchannel}$, 
$n_\Rchannel \geq 3\alpha^2$. %
If $v$ participates in the herald protocol, then it chooses a channel $\Achannel_i$ from $\Achannel_1, \mydots, \Achannel_{n_\Achannel}$ with probability  
  $2^{-i}$.
It then listens on $\Achannel_i$ with probability $\pi_\ell \leq 1/10$ or broadcasts its \ID\ 
  otherwise.\footnote{We want to note that $\pi_\ell$ is a constant parameter that we can choose arbitrarily.}
If $v$ listens, but receives nothing, nothing happens and $v$ stays in \Astate. Should $v$ receive the message of another node $u$ on $\Achannel_i$, then next round it engages with $u$ in the \emph{handshake protocol} as a \emph{herald candidate} ($\Hstate'$), in the hope of moving forward to the \emph{\rbp} together with $u$. Should $v$ choose to broadcast, then it deterministically pursues the handshake as a \emph{leader candidate} ($\Lstate'$), hoping that some other node $u$ has heard its message and joins in for the handshake.
%

\begin{figure}[!ht]
\begin{minipage}{0.48\textwidth}
  \begin{algorithm}[H]
    \caption{\algoHandshake}
    \label{algo:handshake}
    \small
    \begin{algorithmic}[1]
      \Switch{$\state$}   
        \Case{$\Hstate'$}
          \Switch{$\roundcount$}
            \Case{$1,2,5,6$}
              \State Send $\ID_\leader$ on \Hchannel
              \vspace{7mm}
            \EndCase
            \Case{$3,4$}
              \State Listen on \Hchannel 
              \If{$\msg = \emptyset$}
                \State $\handshake \gets \fail$
              \Else\
                \State $\rendezvous \gets \msg.[2]$
                \label{line:algo:handshake:h:rendezvous}
              \EndIf
            \EndCase
          \EndSwitch
          \If{$\handshake = \fail$}
            \State $\roundcount \gets 0$, $\state \gets \Astate$
          \EndIf
          \If{$\roundcount = 6$}
            \State 
            \begin{varwidth}[t]{0.75\linewidth}
              $\roundcount \gets 0$; $\state \gets \Hstate$; $\game \gets \success$; $\lonely \gets 0$
            \end{varwidth}
          \EndIf
        \EndCase
        \algstore{heraldbreakC}
    \end{algorithmic}
  \end{algorithm}
  \renewcommand\footnoterule{}
\end{minipage}
\ \ \ \ \ \ 
\noindent
\begin{minipage}{0.48\textwidth}
  \vspace{3mm}
  \begin{algorithm}[H]
    \small
    \begin{algorithmic}[1]
      \algrestore{heraldbreakC}
      \vspace{3.3mm}
      \Case{$\Lstate'$}
        \Switch{$\roundcount$}
          \Case{$1,2,5,6$}
            \State Listen on \Hchannel
            \If{$\msg = \emptyset$}
              \State $\handshake \gets \fail$
            \EndIf
          \EndCase
          \Case{$3,4$}
            \State $\rendezvous \gets k$
            \State Send $(\ID, \rendezvous)$ on \Hchannel
            \label{line:algo:handshake:l:rendezvous}
            \vspace{10mm}
          \EndCase
        \EndSwitch
        \If{$\handshake = \fail$}
          \State $\roundcount \gets 0$, $\state \gets \Astate$
        \EndIf
        \If{$\roundcount=6$}
          \State 
          \begin{varwidth}[t]{0.75\linewidth}
            $\roundcount \gets 0$; $\state \gets \Lstate$; $\game \gets \success$; $\lonely \gets 0$
          \end{varwidth}
        \EndIf
      \EndCase
    \EndSwitch
    \end{algorithmic}
  \end{algorithm}
  \renewcommand\footnoterule{}
\end{minipage}
\end{figure}


\begin{figure}[!ht]
  \begin{minipage}{0.48\textwidth}
    \begin{algorithm}[H]
      \caption{\algoRBG}
      \label{algo:rbg}
      \small
      \begin{algorithmic}[1]
        \Switch{$\state$} 
          \Case{\Hstate} 
          \State $\activity \gets \max\set{\activity \increaseactivity^{-20}, \activity_\textmin}$
          \label{line:algo:heraldRBG:decr:1}
            \Switch{$\roundcount$}
              \Case{$1,3,5,7 \mod 8$} 
                \Comment \tiny block \Hchannel \scriptsize 
                \State Send $(\state,\ID_\leader)$ on \Hchannel
              \EndCase
              \vspace{7.2mm} 
              \Case{$2 \mod 8$} 
                \Comment \tiny help leader with game \scriptsize 
                \State Send $(\state,\ID_\leader)$ on \Gchannel
              \EndCase
              \vspace{14mm}
              \Case{$4 \mod 8$} 
                \Comment \tiny help leader with game \scriptsize 
                \State Send $(\state,\ID_\leader)$ on \Gchannel
              \EndCase
              \vspace{14mm} 
              \Case{$6 \mod 8$} 
                \State Listen on $\Rchannel_\rendezvous$ 
                \Comment \tiny from previous game \scriptsize 
                \If{$\msg \neq (\ID_\leader, \success, *)$}
                  \State 
                  \begin{varwidth}[t]{0.75\linewidth}
                    $\roundcount \gets 0$, $\state \gets \Astate$; $\lonely \gets 0$
                  \end{varwidth}
                \Else
                  \State $\rendezvous \gets \msg.[3]$
                \EndIf
                \If{$\roundcount > \tthreshold_{\redblue}$} 
                  \State $\state \gets \Estate$
                \EndIf
              \EndCase
              \Case{$8 \mod 8$} 
              \label{line:algo:heraldRBG:round8:h1}
                \Comment \tiny notify neighbors \scriptsize
                \State Send $(\state, \ID_\leader)$ on $\Rchannel_\rendezvous$
                \label{line:algo:heraldRBG:round8:h2}
              \EndCase
            \EndSwitch
          \EndCase
        \algstore{heraldbreakB}
      \end{algorithmic}
    \end{algorithm}
    \renewcommand\footnoterule{}
  \end{minipage}
  \ \ \ \ \ \
  \noindent
  \begin{minipage}{0.48\textwidth}
    \vspace{11mm}
    \begin{algorithm}[H]
      \small
      \begin{algorithmic}[1]
        \algrestore{heraldbreakB}
        \vspace{4mm}
          \Case{$\Lstate$} 
          \State $\activity \gets \max\set{\activity \increaseactivity^{-20}, \activity_\textmin}$
          \label{line:algo:heraldRBG:decr:2}
          \Switch{$\roundcount$} 
              \Case{$1,3, 5,7 \mod 8$} 
                \Comment \tiny block \Hchannel \scriptsize 
                \If{$\roundcount \pmod 8 =1$} 
                  \State pick randomly $\textitcolor \!\in\! \{\textred, \textblue\}$
                \EndIf
                \State Send $(\state,\ID)$ on \Hchannel
              \EndCase
              \Case{$2 \mod 8$} 
                \Comment \tiny \rbg \scriptsize 
                \If{$\textitcolor = \textblue$} 
                  \State Listen on \Gchannel; 
                  \If{$\msg = \emptyset$ or $\ID \notin \msg$} 
                    \State $\game \gets \fail$
                  \EndIf
                \Else\ Send $(\ID)$ on \Gchannel
                \EndIf
              \EndCase
              \Case{$4 \mod 8$} 
                \Comment \tiny \rbg \scriptsize 
                \If{$\textitcolor = \textred$} 
                  \State Listen on \Gchannel; 
                  \If{$\msg = \emptyset$ or $\ID \notin \msg$} 
                    \State $\game \gets \fail$
                  \EndIf
                \Else\ Send $(\ID)$ on \Gchannel
                \EndIf
              \EndCase
              \Case{$6 \mod 8$} 
                \Comment \tiny Send $\game$ \& new $\Rchannel_\rendezvous$ \scriptsize 
                \State Send $(\ID_\leader, \game, k)$ on $\Rchannel_\rendezvous$ 
                \State $\rendezvous \gets k$ 
                \If{$\game = \fail$} 
                  \State 
                  \begin{varwidth}[t]{0.75\linewidth}
                    $\roundcount \gets 0$, $\state \gets \Astate$; $\lonely \gets 0$
                  \end{varwidth}
                \EndIf
                \vspace{3.5mm}
                \If{$\roundcount > \tthreshold_{\redblue}$} 
                  \State $\state \gets \Mstate$
                \EndIf
              \EndCase
              \Case{$8 \mod 8$} 
              \label{line:algo:heraldRBG:round8:l1}
                 \State Listen on $\Rchannel_\rendezvous$
              \label{line:algo:heraldRBG:round8:l2}
              \EndCase
            \EndSwitch
          \EndCase
        \EndSwitch
      \end{algorithmic}
    \end{algorithm}
    \vspace{0.4cm} \renewcommand\footnoterule{}
  \end{minipage}
\end{figure}


\paraclose{Handshake and \RBP.}
\Pseudocode for these two protocols can be found in Algorithms \ref{algo:handshake} and \ref{algo:rbg}. 
In short, a node $h \in \Hstate'$ that received a message in the herald protocol sends for two rounds on the \emph{handshake channel} \Hchannel, then listens twice, and sends again for two rounds. A node $l \in \Lstate'$ that was sending before acts reversely, i.e., it listens, sends and listens. Only if a node receives \emph{all} expected messages it moves forward to the \rbp, otherwise it returns to \Astate. The handshake can only possibly be completed if a pair of exactly one broadcaster and one receiver participates. 

The \emph{\rbp} is a repetition of $\tthreshold_\redblue/8=\Theta(\log n)$ \emph{\rbgs} of $8$ rounds each. In odd rounds, both nodes $l$ and $h$ of the \lhp send a blocking signal on \Hchannel, preventing nearby handshakes to succeed. At the beginning of each game, the leader $l$ picks randomly blue or red. If it picked red, then in round $2$ it will send a message on channel \Gchannel and listens on \Gchannel in round $4$, for blue it acts reversely. In round $6$, $l$ sends on a previously decided meeting channel $\Rchannel_k$ the index $k'$ of the meeting channel for the next \rbg.\footnote{The very first meeting channel is fixed by $l$ during the handshake, confer lines \ref{line:algo:handshake:h:rendezvous} and \ref{line:algo:handshake:l:rendezvous} of \Cref{algo:handshake}.} In round $8$ it listens on $\Rchannel_{k'}$. $h$ on the other hand sends a message in both rounds $2$ and $4$. It listens in round $6$ to update the meeting channel and in round $8$ it sends a message on $\Rchannel_{k'}$.

By design of the handshake and the blocking signals of odd rounds in the \rbp, a leader $l$ can neighbor a leader or herald of a different pair \emph{only} if that other node moved to the \rbp simultaneously or with a $2$-round shift. If $l$ does have such a neighbor, at some point it will not hear its herald in round $2$/$4$, when it listens. $l$ then aborts the \rbp, notifies $h$ in round $6$ and returns to \Astate. The messages sent by $l$/$h$ in round $6$/$8$ also have the purpose of letting nearby listening active nodes reduce their activity values.
An \emph{isolated pair} 
 on the other hand cannot be knocked out anymore\footnote{except by an MIS node, but that already implies progress} and after $\tthreshold_\redblue=\Theta(\log n)$ rounds the pair can assume that \whp there is no other conflicting pair nearby and the leader joins the MIS.

The handshake did not change and the \rbp has been extended by $2$ rounds---now heralds also can reach their neighbors, confer lines \ref{line:algo:heraldRBG:round8:h1}, \ref{line:algo:heraldRBG:round8:h2}, \ref{line:algo:heraldRBG:round8:l1}, \ref{line:algo:heraldRBG:round8:l2} of \Cref{algo:rbg}. Rounds $1$-$6$ are untouched. Unlike in \cite{tr:podc2013}, each round a node spends in the \rbg, it decreases its activity significantly---after all it is getting messages from a leader or herald all the time. This is accounted for in lines \ref{line:algo:heraldRBG:decr:1} and \ref{line:algo:heraldRBG:decr:2}.

\paraclose{Summary of Changes.}
Compared to the algorithm in 
  \cite{tr:podc2013},
the following three things have changed. The \emph{loneliness support block} is not executed anymore, except for maintaining the counter $\lonely$.
Also, the threshold $\tthreshold_\lonely$ has been lowered to $\Theta\bb{\misRTstrong}$ to reflect the new \runtime of the algorithm. The main change is that nodes reduce their activity $\activity$ if they hear from a nearby leader or herald. The change in the \rbg is an addition of $2$ rounds:
the seventh round is just a copy of rounds $1$, $3$ and $5$; the eighth round gives the herald of the pair a possibility to notify nearby active nodes in order to reduce their activity values---so far only leaders and MIS nodes were able to reach out to their neighbors.

Note that while the algorithm itself has barely changed, the analysis needed to be extended vastly to reduce the \runtime of the algorithm to optimal values. 

\section{Analysis}
\label{sec:analysis}

\subsection{Approach}
To prove that Algorithm \ref{algo:2:coreMIS} indeed solves MIS in the given time bounds, we take the following approach. 
In \cite{tr:podc2013} it was proven that the graph, induced by nodes that passed through the \dfilter, has maximum degree $\Delta_\textmax=\bigO(\log^3 n)$.
\hide{
A node $u$ in the \hfilter ($u \in \Vhf$) enters the MIS either if it assumes to be alone, or if it manages to create and maintain a leader-herald bond with a neighboring node for $\tau_\redblue=\Theta(\log n)$ rounds. We show that once $u \in \Vhf$, it either enters the MIS due to assumed loneliness; or if $u$ has a neighbor in $\Vhf$, then within radius $\delta \coloneqq \delta_\alpha=\bigO(\loglogn)$ other nodes enter the MIS in short succession. By adjusting parameters, which affect the \runtime of the algorithm only by constant factors, we can choose $\delta$ as an arbitrary value in $\Theta(\loglogn)$. 
For small enough $\delta$, we get that at most $\alpha^\delta=\bigO(\sqrt{\log n})$ nodes in $N^\delta(u)$ can enter the MIS---before
$u$ or one of its neighbors has to enter itself. In expectation the time intervals between two nodes in $N^\delta(u)$ moving to state \Mstate is in $\bigO(\polyloglog n)$ and using Chernoff bounds we can upper bound the amount of time needed before $u$ gets decided by $\tau_\itruntime=\bigO(\log^2 n / \Freq + \log n)$.
}
A node $u$ in the \hfilter ($u \in \Vhf$) enters the MIS either if it assumes to be alone, or if it manages to create and maintain a leader-herald bond with a neighboring node for $\tau_\redblue=\Theta(\log n)$ rounds. Once $u \in \Vhf$, it either enters the MIS due to assumed loneliness; or if $u$ has a neighbor in $\Vhf$, then within radius $\delta \coloneqq \delta_\alpha=\Theta(\loglogn)$ soon a \lhp is created that maintains its bond for $\tau_\redblue$ rounds.\footnote{``Soon'' indeed means in $\bigO(1)$ rounds in expectation, as long as $\freq =\Omega(\loglogn)$.} So far this is the same as in \cite{tr:podc2013}. There, however, a stagnation of up to $\tau_\redblue$ rounds might follow before the next isolated \lhp or MIS node gets created in $N^\delta(u)$. Considering that up to $\alpha(\delta)$ nodes in $N^\delta(u)$ can enter the MIS before $u$ or one of its neighbors enters itself, the \runtime of the \hfilter is $\bigO(\tau_\redblue \alpha(\delta))$, or $\bigO(\log n \polyloglog n)$ if $\alpha$ is polynomial.

In the present paper, by decreasing activity levels of nodes neighboring \lhps, the stagnation that can be caused by leaders on their way to join the MIS does not last for longer than $\bigO(\loglogn)$ rounds in expectation. This allows the creation of isolated \lhps in $N^\delta(u)$ in a pipelined manner, reducing the expected \runtime of the \hfilter to $\bigO(\alpha(\delta)\loglogn)=\bigO(\alpha^\delta\loglogn)$. Unlike in \cite{tr:podc2013}, here we also can choose $\delta$ as an arbitrarily small value in $\Theta(\loglogn)$ without increasing the \runtime by more than constant factors. Choosing $\delta < \loglogn/\log \alpha$ and a Chernoff argument bounds the \runtime of the \hfilter by $\bigO(\log n)$ with high probability.


In more detail, let $u$ be a node that enters the \hfilter in round $t_u$. For the sake of contradiction assume that $u$ is not decided by time $t_u+\tau_\itruntime$.  If $u$ stays lonely, it enters the MIS eventually in $\tau_\lonely \ll \tau_\itruntime$ rounds. Note that for $u$ to move from being non-lonely to lonely, some node in $N^2(u)$ must have entered the MIS shortly before that and eliminated all neighbors that $u$ had in $\Vhf$. This can happen at most $\alpha^2$ times and thus the time $u$ spends lonely is at most $\alpha^2 \tthreshold_\lonely \ll \tau_\itruntime$.
Hence, assume that $u$ is not lonely, i.e., has a neighbor $u'$, and that no node in $N^2(u)$ joins the MIS. We show that then most of the time both $u$ and $u'$ have a high activity value $\activity$.

The following argumentation motivates this. For a node $u$ to decrease $\activity(u)$, it must neighbor a pair. Let us call isolated pairs (in which the leader does not neighbor another leader or herald) \emph{good pairs} and the others \emph{bad pairs}. Conditioning on the event of a pair being created, there is a constant probability that it is a good pair. This can be considered progress, as it guarantees one of two things: 
Within $\bigO(\log n)$ rounds either the leader of the good pair itself enters the MIS or a neighbor of this pair does.
In the opposite case of bad pairs being created, in expectation these remain bad pairs only for a constant number of rounds.
Moreover, \whp, there are no more than $\bigO(\log n)$ rounds \emph{in total} in which bad pairs exist in $N^3(u)$ after $t_u$, also causing at most $\bigO(\log n)$ rounds of $u$ and $u'$ having an activity value below $1/2$. Adjusting parameters we get that for some $\tau_\progress=\bigO(\tau_\lonely)$ and an arbitrarily small constant $\eps$, for $(1-\eps)\tau_\progress$ rounds in $[t_u,t_u+\tau_\progress]$ the activity values of both $u$ and $u'$ are $1/2$. 

Furthermore, all pairs, good and bad, inform their neighbors. By the definition of good pairs, the leaders of these form an independent set. With our choice of $\delta$ thus at most $\bigO(\sqrt{\log n})$ good pairs exist in $N^\delta(u)$. We argue that the activity values of nodes neighboring a pair that participated in the \rbp for $\Omega(\loglogn)$ rounds (which almost surely holds for good pairs), are below $\activity_\low := \Theta(1/\polylog n)$ with some decent probability (i.e., $1-\log^{\Omega(1)}n$). The total number of nodes in $N^\delta(u)$ becoming part of a good pair in $[t_u,t_u+\tau_\progress]$ is $\bigO(\sqrt{\log n})$ and hence the total amount of nodes neighboring good pairs in that time is $\bigO(\sqrt{\log n} \Delta_\textmax)$. A union bound and a Chernoff bound provide that the total amount of rounds in which any node $v$ in $N^\delta(u)$ neighboring a good pair has an $\activity(v) > \activity_\low$ is less than $\eps \tau_\progress$.

Together with the previous claim we get that in $(1-2\eps)\tau_\progress$ rounds in $[t_u,t_u+\tau_\progress]$ both conditions are true: $\activity(u)=\activity(u')=1/2$ and \emph{all} good pairs in $N^\delta(u)$ ``silenced'' their neighbors---i.e., all their neighbors have activity below $\activity_\low$. Let us call a round with this property \emph{promising for $u$}.
Without going into detail, we can show that now within distance 
$\delta$
there exists a node $w$ with the property of being so-called $\eta$-fat, i.e., $w$'s neighborhood is at least roughly as active as that of any of its neighbors'. Fatness implies that \wcp two nodes $l$ and $h$ in $N^1(w)$ become a \emph{good} \lhp. As said before, such a pair reduces the activity values of its neighbors rather quickly, which causes the property of $\eta$-fatness to move away from $w$ to another node in $N^\delta(u)$ and we can repeat the argument. If a bad pair is created, then $\eta$-fatness might shortly fade, but is restored quickly, so we can almost omit this case. Again using Chernoff tail bounds, we show that at some point $u$ itself becomes $\eta$-fat and now the creation of an MIS node in $N^2(u)$ is inevitably.

We summarize again. Once an MIS node or good pair arises in constant distance from $u$, we are done.
In an $\Omega(\RT)$ interval, $u$ is mostly in a promising state. \Wcp\ every $\bigO(1)$ rounds a node in $N^\delta(u)$ becomes part of a good pair or joins the MIS. In expectation, within $\Theta(\loglogn)$ rounds MIS nodes eliminate their neighbors and good pairs \emph{silence} theirs. After any of those events happen, we measure the time until $u$ is in a promising state again. Using Chernoff over $\bigO(\sqrt{\log n})$ such random variables results in needing at most $\bigO(\log n)$ time, thus, by then $u$ must be covered.


\subsection{Guarantees from the \Dfilter}

We informally state the two main accomplishments of the \dfilter, proper lemma statements are below; for proofs we refer to \cite{tr:podc2013}.
For each node $v$ the \dfilter guarantees that within the \runtime of $\Theta(\RT)$ rounds, 
\begin{compactenum}[(1)]
  \item 
  $v$ or one of its neighbors enters the \hfilter, but 
  \item
  no more than $\Delta_\textmax = \bigO(\log^3 n)$ nodes in $N^1(v)$ do.\footnote{For large enough $n$ it holds that $\Delta_\textmax \leq \log^4 n$ and we assume this in our analysis.}
\end{compactenum}

From now on we only look at the graph $G'$ induced by $V' := \Vhf \cup \Mstate$, induced by non-eliminated nodes that made it past the \dfilter. \emph{All} notations are tied to this subgraph, though we omit this in our notations, i.e., $N(u)$ means the neighborhood of $u$ in $G'$. Instead, if we need to consider nodes from the states $\Wstate$, $\Dstate$, then we explicitly say so and show this e.g. by writing $N_G(u)$.


\begin{lemma}
  \label{lemma:dfiltermaxdegree}
  \Whp, for each node $v$ and each round $r$, at most $\bigO(\log n)$ nodes in $N^{1}_{G}(v)$ come out of the \dfilter\ in round $r$ to enter the \hfilter. Each node that enters the \hfilter\ has spent $\Omega(\log n)$ rounds in the \dfilter.
\end{lemma}


\begin{lemma} 
  \label{lemma:dfilterrunningtime}
  \Whp, for each node $u$ that is in the \dfilter\ in round $r$, by round $r' = r + \bigORT$, either $u$ is \emph{dominated}, in which case it has a dominating neighbor, or at least one node in $N^{1}(u)$ gets out of the \dfilter\ and enters the \hfilter.
\end{lemma}

This statement is the same as Lemma 8.4 in \cite{tr:podc2013}, except that there the bound was listed as $r'=r+\bigO(\log^2n/\Freq)+\softO(\log n)$. Yet the proof in \cite{tr:podc2013} does actually already support the bound $r'=r+\bigORT$. In \Cref{algo:decay}, \algoDF, we changed the style of the algorithm compared to the one in \cite{tr:podc2013}, but not the way the algorithm works, hence we omit the proof for \Cref{lemma:dfilterrunningtime} and refer to \cite{tr:podc2013}.\footnote{The underlying algorithm has been first used and analyzed in \cite{DKN:disc:2012}, in a slightly more restrictive graph model and in \cite{podc2013} it was shown that it also works in BIGs.}



\subsection{Definitions for the \Hfilter}

Practically all parameters (including the above mentioned $\cthreshold_\textmax$) depend in one way or another on the bound on independence, i.e., on $\alpha$, but in most cases those dependencies are captured in the hidden constants of those asymptotic bounds. 

For our analysis of a node $u$ that enters the \hfilter, we observe a specific $\delta=\Theta(\loglogn)$ neighborhood $N^\delta(u)$ of $u$. We set
\reduceSpaceAroundEquation{
  \begin{equation}
    \label{eq:defdelta}
    \delta:=\delta_\alpha:=\frac{\loglogn}{2\log \alpha} = \Theta(\loglogn),
  \end{equation}
}
i.e., $\alpha^\delta = (2^{\log \alpha})^{\frac{\loglogn}{2\log \alpha}} = \sqrt{\log n}$. The choice of $\delta$ guarantees that any independent set in a $\delta$-neighborhood is of size at most $\sqrt{\log n}$.

Our main goal is to show quick progress in $N^\delta(u)$. Progress is clearly achieved if an MIS node arises, but due to the way a node can become an MIS node, we also consider the creation of an isolated \lhp progress (more precisely, the \emph{leader} of the pair needs to be isolated from other nodes in \Lstate or \Hstate), as the leader will eventually join the MIS (or be knocked out permanently by a newly created MIS node). 

\begin{definition} 
(\textbf{Good Pair, Bad Pair}) 
\label{def:goodpair}
  Consider a \lhp $(l,h)$ in round $r$. We say $(l,h)$ is a \emph{good pair} in round $r$ if none of the neighbors of $l$ (other than $h$) is 
    (1) in state \Lstate or 
    (2) in state \Hstate or 
    (3) is a herald candidate in round $5$ or $6$ of its respective handshake.  
  Otherwise we say that $(l,h)$ is a \emph{bad pair}.
\end{definition}

Note that the definition of a good/bad pair is independent of possibly neighboring MIS nodes. MIS nodes existing already for $4$ rounds prevent the creation of \lhps in their neighborhood completely. If on the other hand a new MIS node appears next to a \lhp (which is \whp only possible through the loneliness route), then we have progress in a close neighborhood. 
Also, note that only the leader of the pair must be 'isolated'. There are two reasons for this: (1) only leaders join the MIS (2) by protocol design the \emph{herald} of a pair can only
receive messages from MIS nodes or its own leader---not by other leaders (not even in round $6$) nor other heralds. This is due to the fact that any neighboring heralds act completely synchronously and a leader neighboring a non-paired herald is ahead by precisely $2$ 
rounds.\footnote{
Cf.\ \Cref{lemma:adjpairs} and \Cref{algo:rbg}.
}
Note also that bad pairs can become good, but not vice versa. This is because all leaders and heralds prevent the creation of further leaders/heralds in their neighborhood.

\hide{
The maximum degree $\Delta_\textmax$ of graph $G'$ induced by nodes in $\Vhf \cup \Mstate$ is 
in $\bigO(\log^3 n)$ and hence we can 
assume that $\Delta_\textmax \leq \log^4 n$.
}


%

\begin{definition}
(\textbf{Activity Mass}) 
  \label{def:activity}
  For a node $u$ we define $\Activity(u) := \sum_{v \in N^1(u)} \activity(v)$. We call this the \emph{activity sum} or \emph{activity mass} of node $u$.  
  Furthermore we let $\Activity^\circ(u) := \Activity(u)-\activity(u)=\sum_{v \in N(u)}\activity(v)$.
  In some cases we are only interested in the activity mass of active nodes and then we have $\Activity_\Astate(u) := \sum_{v \in N^1_\Astate(u)} \activity(v)$ and $\AAC(u)$ is defined analogously. Also
  \reduceSpaceAroundEquation{
    \begin{eqnarray}
      \label{def:activity:min}
      \activity_\textmin &:=& \log^{-24} n =
      \Theta(1/\polylog n),
      \\
      \activity_\low &:=& \sqrt{\activity_{\textmin}} = \log^{-12} n.
    \end{eqnarray}
  }
\end{definition}

\begin{definition} (\textbf{Fatness}) 
  We call a node $u$ 
  \emph{$\hat{\eta}$-fat} for some 
$\hat{\eta} \in (0,1)$, if 
$\, \Activity(u) \geq \hat{\eta}\cdot\max_{v \in N(u)} \{\Activity(v)\}$.
\end{definition}
In simple words, in terms of activity mass, $u$ is (at least) in the same 'league' as its neighboring nodes.
Using this we choose a specific fatness parameter $\eta < 1$:
\reduceSpaceAroundEquation{
\begin{equation}
  \label{eq:defeta}
  \eta = \eta_\alpha := \alpha^{-8} \leq \alpha^{-2\frac{\log \Delta_\textmax}{\loglogn}}
\end{equation}
}
The choice of $\eta$ assures that a chain of activity sums $(\Activity(v_i))_{i \geq 1}$ of nodes $v_i$ on a path $v_1, v_2, v_3, \dots$ with $\Activity(v_i) \geq \eta^{-1} \Activity(v_{i-1})$ and $\Activity(v_1)\geq 2$ has length at most $\delta$, because 
\reduceSpaceAroundEquation{
\begin{equation}
  \label{eq:etachain}
  (\eta^{-1})^{\delta} = (\alpha^{-8})^{-\frac{\loglogn}{2\log \alpha}} \geq 
  (2^{\log \alpha})^{2\frac{\log \Delta_\textmax}{\loglogn}\frac{\loglogn}{2\log \alpha}} 
  = \Delta_\textmax 
  \stackrel{\activity(u)\leq 1/2}{>} 
  \max_{u \in \Vhf} \Activity(u).
\end{equation}
}

The algorithm needs to know a few more parameters. $\increaseactivity$ and $\decreaseactivity$ govern the changes in a node's activity level. The former is a small constant, greater than, but close to $1$. In most rounds a node $u$ increases $\activity(u)$ by $\increaseactivity$. $\decreaseactivity$ is a much larger factor used for decreasing activity, large enough to undo many previous increments, but still in $\bigO(1)$. 
\reduceSpaceAroundEquation{
\begin{eqnarray}
  \increaseactivity &:=& 2^{6/(1000 \actm )} > 1
  \\
  \decreaseactivity &:=& \increaseactivity^{20\actm} = 2^{12/100} > 1
\end{eqnarray}
}
\noindent $\actm$ is a large enough constant that depends on $n_\Rchannel$, but assuming that $n_\Rchannel \geq 3\alpha^2$, $\actm \geq 2^{16}n_\Rchannel$ suffices.
Since $\activity_\textmin = \log^{-24} n$, 
$167\actm\loglogn = \Theta(\loglogn)$ consecutive increments raise a node's activity value to $1/2$. Analogously, $\Theta(\loglogn)$ decrements decrease it to its minimal value $\activity_\textmin$. 

Also two time thresholds $\tau_\redblue=\Theta(\log n)$ and $\tau_\lonely = \Theta(\RT)$ are needed by the algorithm. $\tthreshold_\redblue$ is the number of rounds a node spends in the \rbp, and it is a multiple of $8$. If a node $u \in \Vhf$ does not receive a single message for $\tthreshold_\lonely$ consecutive rounds, while being in the \hfilter, a $u$ deduces that it is alone or all its neighbors got eliminated, and joins the MIS.
In our analysis we use further time thresholds $\tau_\notif = \Theta(\log n)$, $\tau_\progress = \Theta(\RT)$ and $\tau_\itruntime = \Theta(\RT)$, for which the following inequality chain holds:
\reduceSpaceAroundEquation{
$$
\tau_\itruntime \gg \tau_\lonely \gg \tau_\progress \gg \tau_\redblue \gg \tau_\notif
$$
}
$\tau_\notif$ is the maximum time needed for an MIS node to notify, \whp, all its neighbors. 
If
a node $u$ is not lonely, then, \whp, significant progress is achieved in less than $\tau_\progress$ rounds; more precisely, an MIS node is created in 
$N^{\bigO(1)}(u)$.
\Whp, $\tau_\itruntime$ is the maximum time a node spends in the \hfilter before it gets decided.



\subsection{Candidate Election---Nodes in States \Astate\ (and $\Lstate'$)}
\label{subsec:candidate:election}
At first we establish a few facts about how nodes can transit from state \Astate\ to state $\Lstate'$ or $\Hstate'$, respectively. Note that nodes can switch between states \Astate\ and $\Lstate'$ without communication, but to get towards any of the three states $\Hstate'$, $\Lstate$ and $\Hstate$, communication is mandatory.

The next lemma contains a variety of events. To not clutter the lemma statement, we define them here. 
Let $k$ be a positive integer \emph{constant}, $r$ some round, $u$ some node in the \hfilter in round $r$, $i$ an index from $1, \dots, \nac$, $S$ be a (possibly empty) subset of $N^k(u) \cap \Astate$. Furthermore let $\partial S \subset S$ be \emph{the} subset of $S$ that has connections outside of $S$, but in $N^k(u)$, i.e., $\partial S := S \cap N\big(N^k(u)\setminus S\big)$. At last, let $S=S^n \dotcup S^b \dotcup S^l$ be a partition of $S$. We call the tuple $(k,u,r,i,S)$) a \emph{constellation}. For a constellation the following events are defined:

\begin{tabular}{lll}
  \hspace{-2.2mm}\textbullet & \hspace{-2mm}$\mathcal{S}_{\neg i}  /  \mathcal{S}^n_{\neg i}$: & \emph{no} node in $S  /  S^n$ operates on $\Achannel_i$ in round $r$, \\
  \hspace{-2.2mm}\textbullet & \hspace{-2mm}$\mathcal{S}_{i}  /  \mathcal{S}^b_{i}  /  \mathcal{S}^l_{i}$: & \emph{all} nodes in $S  /  S^b  /  S^l$ operate on $\Achannel_i$ in round $r$, \\
  \hspace{-2.2mm}\textbullet & \hspace{-2mm}$\partial \mathcal{S}_{\neg i}$: & \emph{no} node in $\partial S$ operates on $\Achannel_i$ in round $r$, \\
  \hspace{-2.2mm}\textbullet & \hspace{-2mm}$H_i$: & \emph{no} node in $N^k(u)\setminus S$ receives a message on channel $\Achannel_i$ in round $r$, \\
  \hspace{-2.2mm}\textbullet & \hspace{-2mm}$H_{\neg i}$: & \emph{no} node in $N^k(u)$ receives a message on some channel $\Achannel_j \neq \Achannel_i$ in round $r$ and \\
  \hspace{-2.2mm}\textbullet & \hspace{-2mm}$H$: & \emph{no} node in $N^k(u)$ receives a message on any channel in $\set{\Achannel_1,\dots,\Achannel_\nac}$ in round $r$.
\end{tabular}


\begin{lemma}
  \label{lemma:wcpnoherald}
  Let $(k,u,r,i,S)$ be a constellation. Then,
\iftrue
  \TabPositions{4.1cm}
  \begin{enumerate}[(1)]
    \item $\Pr(H) =$ \tab $ 1 - \bigO(\pi_\ell \alpha^k)$     
    \label{lemma:wcpnoherald:1}
    \item $\Pr(H_{\neg i}|\mathcal{S}_{\neg i}) =$ \tab $ 1 - \bigO(\pi_\ell \alpha^k)$    
    \label{lemma:wcpnoherald:2}
    \item $\Pr(H_{\neg i}|\mathcal{S}_{i}) =$ \tab $ 1 - \bigO(\pi_\ell \alpha^k)$    
    \label{lemma:wcpnoherald:3}
    \item $\Pr(H_{\neg i}|\mathcal{S}^n_{\neg i} \wedge \mathcal{S}^b_{i} \wedge \mathcal{S}^l_{i}) =$ \tab $ 1 - \bigO(\pi_\ell \alpha^k)$    
    \label{lemma:wcpnoherald:4}
    \item $\Pr(H_i|\partial \mathcal{S}^n_{\neg i}) =$ \tab $ 1 - \bigO(\pi_\ell \alpha^k)$    
    \label{lemma:wcpnoherald:5}
  \end{enumerate}
\fi
\end{lemma}

The proofs for these statements are provided in Appendix \ref{lemma:app:wcpnoherald}.

Look at (\ref{lemma:wcpnoherald:1}). The lemma says that the probability for a herald candidate to be created in any single round for any neighborhood of constant radius is at most linear in $\pi_\ell$. Since $\pi_\ell$ is an arbitrarily small constant parameter chosen by us, we can make the probability for this event arbitrarily small. The proof for (\ref{lemma:wcpnoherald:1}) is exactly the same as the proof for Lemma 8.6 in \cite{podc2013}, with $\alpha(k)$ replaced by $\alpha^k$. We provide its proof 
in Appendix \ref{lemma:app:wcpnoherald} 
nevertheless, as (\ref{lemma:wcpnoherald:2}) and (\ref{lemma:wcpnoherald:3}) are new results that do directly depend on (\ref{lemma:wcpnoherald:1}). What (\ref{lemma:wcpnoherald:2}) says, is, that even if we condition on some nodes $S \subset N^k(u)$ \emph{not} to operate on channel $\Achannel_i$, this does not increase the chance (significantly---i.e., by more than a constant factor) for these or other nodes to receive anything on some channel $\Achannel_j \neq \Achannel_i$. Analogously, (\ref{lemma:wcpnoherald:3}) claims that if we condition on some nodes $S \subset N^k(u)$ to \emph{definitely} operate on channel $\Achannel_i$, this does still not affect the probability for any other node to receive any message on some other channel $\Achannel_j \neq \Achannel_i$.
(\ref{lemma:wcpnoherald:4}) is a combination of (\ref{lemma:wcpnoherald:2}) and (\ref{lemma:wcpnoherald:3}), plus we even fix the knowledge of which nodes, that operate on $\Achannel_i$, do broadcast ($S^b$) or listen ($S^l$).
(\ref{lemma:wcpnoherald:5}) has already been proven in \cite{podc2013} as Claim 8.9, except that there $k$ was fixed to $2$. The analogous proof is provided 
in the Appendix 
as well.

\vspace{1mm}
Under certain conditions the creation of herald candidates can be lower bounded. However, for our algorithm to work, we need not only to prove that they are created, but that this creation happens in \emph{solitude}, i.e., in a close neighborhood no other herald candidates are created. Hence the next lemma is a key result in our whole proof. It is almost Lemma 8.8 from \cite{podc2013}, however with some adaptions. 
To state the lemma, we need to introduce the parameters

\begin{eqnarray}
  \activity_\textmin &:=& \log^{-6\ttkdelta} n =\frac{1}{\Omega(\polylog n)}
  \\
  \activity_\low &:=& \sqrt{\activity_{\textmin}} = \log^{-3\ttkdelta} n,
\end{eqnarray}
depending on the constant $\ttkdelta \geq \log_{\log n} \Delta_\textmax$.

\begin{lemma}
  \label{lemma:fatgivesherald}
  Let $t$ be a round in which for a node $u$ in state \Astate\ in the \hfilter the following holds:
  \begin{compactitem}
    \item there is no herald candidate in $N^2(u)$,
    \item all nodes $v \in N^2(u)$ that neighbor a herald or leader, have $\activity(v) \leq \activity_\low$,
    \item all nodes in $N^2(u)$ neighboring MIS nodes are eliminated,
    \item $\Activity(u) \geq 1$,
  \end{compactitem}
  If in addition it holds that either
  \begin{compactenum}[\bfseries (a)]
    \item $\Activity(u) <5\alpha$, $u$ is $\frac{1}{5\alpha}$-fat and $\activity(u)=\frac{1}{2}$, or
    \item $\Activity(u) \geq 5\alpha$ and $u$ is $\eta$-fat.
  \end{compactenum}
  Then by the end of round $t' \in [t, t+7]$, with probability $\Omega(\pi_\ell)$ either a node in $N^2(u)$ joins the MIS or a good pair $(l,h)\in (\Lstate \cap N^1(u)) \times (\Hstate \cap N^1(u))$ is created. 
\end{lemma}

Let us start with an intuition of this Lemma. The basic intention is to show that if $u$ is $\eta$-fat, then \wcp in constant many rounds a \emph{good} \lhp with both endpoints in $N^1(u)$ arises---for this $u$ itself does not have to have a high activity value, i.e., $u$ does not need to be a likely part of the \lhp. 
The lemma lists many requirements. We show later that shortly after a node $v$ moves to the \hfilter, within distance $\delta=\bigO(\loglogn)$ most of the time there exists a node $u$ that satisfies these conditions. We also show that if an isolated pair is created in $N^1(u)$, those requirements are again satisfied $O(\polyloglog n)$ rounds later (in expectation) by another node $u'$ in this $\delta$-neighborhood of $v$. 


The proof follows mostly the lines of the one in \cite{tr:podc2013}, but there a constant number of factors of $1/2$ accumulate, while here they have to be exchanged by factors of $\eta$---but this only causes changes in asymptotically irrelevant constants. In addition, nodes neighboring leaders or heralds need special attention. A full detailed proof of this new version is listed in Appendix \ref{lemma:app:fatgivesherald}. There are few things we want to elaborate. Within the neighborhood of a fat node $u$ with activity mass at least one, \wcp ``good things'' happen (i.e., the creation of MIS nodes or good \lhps) within constant many rounds, \emph{even if} there are herald candidates nearby or if some nodes neighboring \emph{bad} leaders or heralds have high activity values. In other words, the first two statements could be relaxed a bit. Instead we use other results to show that from those relaxed conditions one can get to the tighter conditions listed here \wcp in constant many rounds. Also note that we allow $7$ rounds to pass, even though the handshake only needs $6$ rounds. This is due to the fact that we require a certain property that might not be given in round $t'$, but, \wcp, is given one round later. As a last remark we want to say that $5\alpha$ in this lemma could be replaced by any constant greater than $\frac{20}{7}\alpha$ and $\eta$ could be replaced by any fatness constant smaller than $1$. 


\medskip
We argue in the subsections about the handshake protocol and the \rbp, Subsections \ref{subsec:handshake} and \ref{subsec:redblue},
that every time an isolated pair is created, the algorithm achieves progress, as it guarantees the creation of an MIS node nearby---even if this event is delayed by $\bigO(\log n)$ rounds.

Therefore, Lemma \ref{lemma:fatgivesherald} ``promises'' progress in the proximity of a fat node. However, we do not have such a statement for areas without fat nodes. As we describe a bit more in detail later, an excessive creation of bad pairs in such areas can even cause problems for our argumentation. However, the next result shows that \emph{if} a pair is created at all, then \wcp this pair is good. This allows us to proof later in Lemma \ref{lemma:alwaysactive} that nodes in the \hfilter are practically always very active in the candidate election process---unless they already neighbor an MIS node or a good pair.



\ifthenelse{\boolean{swapheraldclaim}}{
\begin{lemma}
  \label{lemma:heraldbounds}
  Let $r$ be a round in which node $u$ is in state \Astate and $N_\Astate(u) \neq \emptyset$.
  Let $B^u$ be the event that at the end of round $r$, $u$ moves to state
  $\Hstate'$ due to receiving a message from some node $v \in
  N_\Astate(u)$ on some channel $\Achannel_{\bar{\lambda}}$. Further, let $D^u\subseteq B^u$ be the event that
  $B^u$ holds and in addition no other node $v' \in N^3(v)\setminus\set{u}$
  receives any message on channel $\Achannel_{\bar{\lambda}}$ in round $r$. 
  It holds that
  \begin{eqnarray}
    \label{eq:claimequationsupper}
    \Pr(B^u) &=&
    \begin{cases}
      \bigO\brackets{\pi_\ell \frac{\activity(u)}{\AAC(u)}} & \AAC(u) > 2 \\
      \bigO\brackets{\pi_\ell \activity(u) \AAC(u)} & \AAC(u) \leq 2      
    \end{cases},
    \\
    \label{eq:claimequationslower}
    \Pr(D^u) &=&
    \begin{cases}
      \Omega\brackets{\pi_\ell \frac{\activity(u)}{\AAC(u)}} & \AAC(u) > 2 \\
      \Omega\brackets{\pi_\ell \activity(u) \AAC(u)} & \AAC(u) \leq 2      
    \end{cases}.
  \end{eqnarray}
\end{lemma}
}{
\begin{lemma}
  \label{lemma:heraldwcpgoodherald}
  Let $B(r)^{u,v}$ be the event that in round $r$ node $u \in \Astate_r$ receives a message from one of its neighbors $v\in \Astate_r$, neither of them neighboring any leader, herald or herald candidate in the $5$th or $6$th round of its handshake protocol. Let $\hat H(r')^{u',v'}$ be the event that at the beginning of round $r'$ node $u' \in \Hstate_{r'}$ and $v' \in \Lstate_{r'}$ form a good \lhp and that $\Hstate' \cap N^3(u') = \emptyset$, i.e., there are no herald candidates in the $3$-neighborhood of $u'$. Then
  \begin{equation}
    \Pr(\hat H(r+8)^{u,v}|B(r)^{u,v}) = \Omega(1)
  \end{equation}
\end{lemma}
}




\ifthenelse{\boolean{swapheraldclaim}}{}{
The proof for the lemma is not straightforward. We consider the two cases of $\Activity(u)>2$ and $\Activity(u) \leq 2$ and for those we establish claims about the probabilities for $u$ to become a herald in general and a good herald in particular before we provide the proof for Lemma \ref{lemma:heraldwcpgoodherald}.
Recall that $\Activity_\Astate(u)$ is the activity mass of all active nodes in $N^1(u)$ and that $\AAC(u)$ is the mass of those in $N(u)$, i.e., $\AAC(u)=\Activity_\Astate(u)-\activity(u)$.


\begin{myclaim}
  \label{claim:heraldbounds}
  Let $r$ be a round in which node $u$ is in state \Astate and $N_\Astate(u) \neq \emptyset$.
  Let $B^u$ be the event that at the end of round $r$, $u$ moves to state
  $\Hstate'$ due to receiving a message from some node $v \in
  N_\Astate(u)$ on some channel $\Achannel_{\bar{\lambda}}$. Further, let $D^u\subseteq B^u$ be the event that
  $B^u$ holds and in addition no other node $v' \in N^3(v)\setminus\set{u}$
  receives any message on channel $\Achannel_{\bar{\lambda}}$ in round $r$. 
  It holds that
  \begin{eqnarray}
    \label{eq:claimequationsupper}
    \Pr(B^u) &=&
    \begin{cases}
      \bigO\brackets{\pi_\ell \frac{\activity(u)}{\AAC(u)}} & \AAC(u) > 2 \\
      \bigO\brackets{\pi_\ell \activity(u) \AAC(u)} & \AAC(u) \leq 2      
    \end{cases},
    \\
    \label{eq:claimequationslower}
    \Pr(D^u) &=&
    \begin{cases}
      \Omega\brackets{\pi_\ell \frac{\activity(u)}{\AAC(u)}} & \AAC(u) > 2 \\
      \Omega\brackets{\pi_\ell \activity(u) \AAC(u)} & \AAC(u) \leq 2      
    \end{cases}.
  \end{eqnarray}
\end{myclaim}
}

\begin{proof}
  In the calculations below we make use of the following inequalities.
  \begin{eqnarray}
    \label{eq:basic1}
    1-\pi_\ell &\geq& 0.9 \\
    \label{eq:basic2}
    1-x &\leq& e^{-x}, \quad \forall x \\
    \label{eq:basic3}
    1-x &\geq& e^{-2x}, \quad \forall x \in [0,1/2] \\
    \label{eq:basic4}
    \prod_{w \in A\setminus B} f(w) &\leq& F^{|B|} \prod_{w \in A} f(w), \quad \mbox{if } f(w) \geq F^{-1}
  \end{eqnarray}

  If $\AAC(u) > 2$, for simplicity, we assume that $\log \AAC(u)$ is a positive integer. It becomes clear from the proof that for non-integer values 
  an adaption is straightforward, but hard to read.

  We let $B^{u,v}_i$ be the event that an \emph{active} node $u$ receives a message from an \emph{active} node $v$ on channel $\Achannel_i$, i.e., $u$ listens and $v$ broadcasts on $\Achannel_i$, while no other node $w \in (N_\Astate(u))\setminus\set{v}$ broadcasts on $\Achannel_i$. For different $v$ these events are disjoint, and we can define $B^u_i := \bigcup_{v \in N_\Astate(u)} B^{u,v}_i$ and we see that $B^u = \bigcup_{i=1}^{n_\Achannel} B^u_i$.


  More restrictive are the analogously defined events $D^{u,v}_i$, $D^u_i$, in which we also require that no node $x \in N^3(v)\setminus\set{u}$ receives a message on channel $i$. In that case, $D^u=\bigcup_{i=1}^{n_\Achannel} D^u_i$ is the event that nodes $u$ and $v$ engage in the handshake protocol in the round after they met on some channel $\Achannel_{\bar{\lambda}}$, without other nodes nearby having received something on that same channel. In particular, no other herald candidates try to engage with $v$ in the handshake protocol.



  \paraclose{Upper Bounds}

  Let $q_i^w$ be the probability that node $w$ does broadcast on channel $\Achannel_i$, i.e., $q_i^w = (1-\pi_\ell)\activity(w)2^{-i} \leq 1/4$, and accordingly $\bar{q}^w_i := 1-q^w_\lambda \geq 3/4$ is the probability that $w$ does not broadcast on channel $\Achannel_i$. We denote with $p_i^w$ the probability that $w$ listens on channel $\Achannel_i$. 

\ifthenelse{\boolean{hasLNCSformat}}
{
  \begin{equation*}
    \begin{split}
      \label{eq:BuviUB}
      \Pr(B^{u,v}_i)
      & = p_i^u q_i^v \prod_{w \in N_\Astate(u)\setminus\set{v}} \bar{q}_i^w 
      \\
      &= 
      \pi_\ell \activity(u) 2^{-i}(1-\pi_\ell) \activity(v) 2^{-i} \prod_{w \in N_\Astate(u)\setminus\set{v}} \bigbrackets{ 1-(1-\pi_\ell)\activity(w)2^{-i} } 
      \\
      & \stackrel{(\ref{eq:basic2}),(\ref{eq:basic4})}{\leq} 
      \pi_\ell \activity(u)\activity(v)2^{-2i}
      \frac{4}{3} e^{-\frac{1}{2}\AAC(u)2^{-i}}
    \end{split}
  \end{equation*}

  \begin{equation}
    \label{eq:BuiUB}
    \begin{split}
      \Pr(B^u_i) 
      &\leq 
      \pi_\ell \activity(u) \AAC(u) 2^{-2i+1}
      e^{-\AAC(u)2^{-i-1}}
      \\
      &= 8\pi_\ell \frac{\activity(u)}{\AAC(u)} (\AAC(u) 2^{-i-1})^2
      e^{-\AAC(u)2^{-i-1}} =: C^u_i
    \end{split}
  \end{equation}
}{
  \begin{equation*}
    \begin{split}
      \label{eq:BuviUB}
      \Pr(B^{u,v}_i)
      &= p_i^u q_i^v \prod_{w \in N_\Astate(u)\setminus\set{v}} \bar{q}_i^w 
      = 
      \pi_\ell \activity(u) 2^{-i}(1-\pi_\ell) \activity(v) 2^{-i} \prod_{w \in N_\Astate(u)\setminus\set{v}} \bigbrackets{ 1-(1-\pi_\ell)\activity(w)2^{-i} } 
      \\
      & \stackrel{(\ref{eq:basic2}),(\ref{eq:basic4})}{\leq} 
      \pi_\ell \activity(u)\activity(v)2^{-2i}
      \frac{4}{3} e^{-\frac{1}{2}\AAC(u)2^{-i}}
    \end{split}
  \end{equation*}

  \begin{equation}
    \label{eq:BuiUB}
      \Pr(B^u_i) 
      \leq 
      \pi_\ell \activity(u) \AAC(u) 2^{-2i+1}
      e^{-\AAC(u)2^{-i-1}}
      = 8\pi_\ell \frac{\activity(u)}{\AAC(u)} (\AAC(u) 2^{-i-1})^2
      e^{-\AAC(u)2^{-i-1}} =: C^u_i
  \end{equation}
}

  Consider the case $\AAC(u) >2$. 
  For any $\zeta>0$ it holds that $\zeta^2e^{-\zeta}=\bigO(1)$, and by using $\zeta=\AAC(u)2^{-i}$ and we get that for any fixed $i$
  \begin{equation}
    \label{eq:CuiUB}
    C^u_i = \bigO\Big(\pi_\ell \frac{\activity(u)}{\AAC(u)}\Big) \quad \mbox{and} \quad C^u_i=\Theta(C^u_{i+1})
  \end{equation}
  Furthermore, with $\lambda := \log \AAC(u)$:
  \begin{align*}
    \frac{C_{i+1}^u}{C_i^u} &= \frac{1}{4}e^{\frac{1}{4}\AAC(u)2^{-i}} < \frac{1}{2} & \forall i &\geq \lambda\\
    \frac{C_{i}^u}{C_{i+1}^u} &=  4e^{-\frac{1}{4}\AAC(u)2^{-i}} < \frac{1}{2} & \forall i &\leq \lambda -4
  \end{align*}
  Thus, the sum of all $C^u_i$ with $i$ larger than $\lambda$ can be upper bounded by $C^u_\lambda$ using the geometric series, and, if $\lambda > 4$, all $i \leq \lambda-4$ can be upper bounded by $C^u_{\lambda-4}$. However, due to (\ref{eq:CuiUB}), $C^u_{\lambda-4}+ \dots + C^u_\lambda$ is in $\bigO(C^u_\lambda)=\bigO\bb{\pi_\ell \frac{\activity(u)}{\AAC(u)}}$.
  In total we get that
  $
  \Pr(B^u) = \bigO\bb{\pi_\ell \frac{\activity(u)}{\AAC(u)}}.
  $

  Now consider the case $\AAC(u) \leq 2$. 
  In (\ref{eq:BuiUB}) we upper bound $e^{-\AAC(u)2^{-i-1}}$ by $1$ and get
  \begin{equation*}
    \Pr(B^u_i) \leq \pi_\ell \activity(u) \AAC(u) 2^{-2i+1} =: \hat{C}^u_i \quad \mbox{and} \quad \hat{C}^u_{i+1} \leq \frac{1}{4} \hat{C}^u_i.
  \end{equation*}
  Clearly, $\hat{C}^u_1 \leq \sum_{i=1}^{n_\Achannel} \hat{C}^u_i \leq 2 \hat{C}^u_1$ due to the convergence of the geometric series. In other words:
  \begin{equation*}
    \Pr(B^u) \leq 2 C^u_1 = \bigO(\pi_\ell \activity(u)\AAC(u)).
  \end{equation*}
  
  This
  concludes the proof of the first part of the claim, equations (\ref{eq:claimequationsupper}). Note that in both cases $\AAC(u) \leq 2$ and $\AAC(u) > 2$ we had that
  \begin{equation}
    \label{eq:BulambdaUB}
    \Pr(B^u_\lambda)=\Omega(B^u).
  \end{equation}


  \vspace{3mm}
  \paraclose{Lower Bounds}
  \\
  For lower bounds we study 2 specific channels, depending whether $\AAC(u)$ is greater than $2$ or not. More precisely, we study $\Achannel_\lambda$, where $\lambda := \max\set{1, \log \AAC(u)}$ and show that event $D^u$ happens on this channel with the desired probability. 
  The base idea is to show that \wcp no node $x \in N^3_\Astate(u)$ with $\AAC(x) \hide{\lessapprox} \lesssim \AAC(u)$ is listening on channel $\Achannel_\lambda$ at all, while all nodes $x$ with $\AAC(x) \hide{\gtrapprox}\gtrsim \AAC(u)$ experience a collision on $\Achannel_\lambda$.

  \paraclose{Case 1: \boldmath{$\AAC(u) > 2$, $\lambda = \log \AAC(u)$ and $2^{-\lambda} = 1/\AAC(u)$}}.


  We start with some definitions. We define 
  $$X^{u,v}_m := \set{x \in N^3_\Astate(v)\setminus\set{u,v}: 10m \AAC(u) \leq \AAC(x) < 10(m+1) \AAC(u)}.$$ 
  Let $F^{u,v}_{\lambda,m}$ be the event that \emph{no} $x \in X^{u,v}_m$ receives \emph{any} message broadcasted  on $\Achannel_\lambda$ and let $E^x_\lambda$ be the event that $x \in \Astate$ \emph{successfully} receives a message on channel $\Achannel_\lambda$. If $x \in N_\Astate(v)\setminus\set{u}$, then $E^x_\lambda|B^{u,v}_\lambda$ implies that $x$ receives $v$'s message.
  Further, $p_\lambda^{u,v} = \pi_\ell(1-\pi_\ell)\activity(u)\activity(v)2^{-2\lambda}$ is the probability that nodes $u$ and $v$ meet on channel $\Achannel_\lambda$ with $u$ listening and $v$ broadcasting. 

  Apparently,
  \begin{equation}
    \label{eq:3}
    \Pr(B_\lambda^{u,v})
    =
    p_\lambda^{u,v} \prod_{w \in N_\Astate(u) \setminus\set{v}} \bar{q}_\lambda^w.
  \end{equation}
  We want to calculate $\Pr(E^x_\lambda|B^{u,v}_\lambda)$. For that we need to distinguish between the two cases $x \in N_\Astate(v)\setminus\set{u}$ and $x \notin N^1_\Astate(v)$. For the first case, no other neighbor of $x$ than $v$ is allowed to broadcast, while in the second case exactly one neighbor of $x$ needs to broadcast on $\Achannel_\lambda$.

  \begin{multline}
    x \in N_\Astate(v)\setminus\set{u}: 
    \quad
    \Pr(E^x_\lambda|B^{u,v}_\lambda) 
    = 
    \frac{\Pr(B^{u,v}_\lambda \cap E^x_\lambda)}{\Pr(B^{u,v}_\lambda)} 
    \\
    =
    \frac{1}{\Pr(B^{u,v}_\lambda)} \frac{\Pr(B^{u,v}_\lambda)}{\bar{q}^x_\lambda}  \pi_\ell\activity(x) 2^{-\lambda}
    \hspace{-5mm}
    \prod_{y \in  N_\Astate(x)\setminus\set{u,v}} 
    \hspace{-5mm}
    \bar{q}_\lambda^y
    \leq
    \frac{4}{3} \pi_\ell \frac{\activity(x)}{\AAC(u)} \prod_{y \in N_\Astate(x)\setminus\set{u,v}} \bar{q}_\lambda^y
    \label{eq:lb1}
  \end{multline}

\ifthenelse{\boolean{hasLNCSformat}}
{
  \begin{multline}
    x \notin N^1_\Astate(v):
    \quad
    \Pr(E^x_\lambda|B^{u,v}_\lambda) 
    =
    \Pr(E^x_\lambda|u \text{ listens}) 
    \\ 
    =
    \pi_\ell \activity(x) 2^{-\lambda} \cdot \sumleft{y \in N_\Astate(x)\setminus\set{u}}{
      q^y_\lambda
    } \cdot \prodleft{z \in N_\Astate(x)\setminus\set{u,y}}{
      \bar{q}^z_\lambda
    }
    \\
    =
    \pi_\ell \activity(x) 2^{-\lambda} \cdot \sumleft{y \in N_\Astate(x)\setminus\set{u}}{
      (1-\pi_\ell)
      \activity(y)
      2^{-\lambda}
    } \cdot \prodleft{z \in N_\Astate(x)\setminus\set{u,y}}{
      \bar{q}^z_\lambda
    }
    %
    \\
    \leq
    \pi_\ell\frac{\activity(x)}{\AAC(u)} \cdot \frac{\AAC(x)}{\AAC(u)} \prodleft{z \in N_\Astate(x)\setminus\set{u,y}}{\bar{q}^z_\lambda}
    \label{eq:lb2}
  \end{multline}
}{
  \begin{multline}
    x \notin N^1_\Astate(v):
    \quad
    \Pr(E^x_\lambda|B^{u,v}_\lambda) 
    =
    \Pr(E^x_\lambda|u \text{ listens}) 
    =
    \pi_\ell \activity(x) 2^{-\lambda} \cdot \sumleft{y \in N_\Astate(x)\setminus\set{u}}{
      q^y_\lambda
    } \cdot \prodleft{z \in N_\Astate(x)\setminus\set{u,y}}{
      \bar{q}^z_\lambda
    }
    \\
    =
    \pi_\ell \activity(x) 2^{-\lambda} \cdot \sumleft{y \in N_\Astate(x)\setminus\set{u}}{
      (1-\pi_\ell)
      \activity(y)
      2^{-\lambda}
    } \cdot \prodleft{z \in N_\Astate(x)\setminus\set{u,y}}{
      \bar{q}^z_\lambda
    }
    %
    \leq
    \pi_\ell\frac{\activity(x)}{\AAC(u)} \cdot \frac{\AAC(x)}{\AAC(u)} \prodleft{z \in N_\Astate(x)\setminus\set{u,y}}{\bar{q}^z_\lambda}
    \label{eq:lb2}
  \end{multline}
}

  \hide{
    $(*)$ follows from the fact that 
    $$\sum_{y \in N_\Astate(x)\setminus\set{u,v}}q^y_\lambda 
    = 
    \sumleft{y \in N_\Astate(x)\setminus\set{u,v}}{(1-\pi_\ell)\activity(y)2^{-\lambda}}
    \leq 
    \frac{\AAC(x)}{\AAC(u)}.
    $$
  }
  
  \hide{
    For $x \in N_\Astate^3(v)\setminus N_\Astate^1(v)$ we get a similar result:
    \begin{eqnarray*}
      \Pr(E_\lambda^x \cap B_\lambda^{u,v}) 
      &=& 
      \Pr(B_\lambda^{u,v}) \pi_\ell \activity(x) 2^{-\lambda} \cdot \sumleft{y \in N_\Astate(x)\setminus\set{u,v}}{(1-\pi_\ell)\activity(y)2^{-\lambda}} \cdot \underbrace{\prodleft{z \in N_\Astate(x)\setminus\set{u,y}}{(1-(1-\pi_\ell)\activity(z)2^{-\lambda})}}_{\leq 1}
      \\
      &\leq&
      \Pr(B_\lambda^{u,v}) \pi_\ell \activity(x) 2^{-\lambda} (1-\pi_\ell)2^{-\lambda}\AAC(x) \leq \Pr(B_\lambda^{u,v}) \pi_\ell \activity(x) 2^{-2\lambda} \AAC(x)
      \\
      \Pr(E_\lambda^x|B_\lambda^{u,v}) 
      &\leq& 
      \pi_\ell \activity(x) 2^{-2\lambda} \AAC(x)
      \stackrel{\AAC(x) \leq 10 \AAC(u)}{\leq}
      10 \pi_\ell \activity(x) 2^{-\lambda}
    \end{eqnarray*}
  }

  Now, using Lemma \ref{lemma:weightedturan}, observe that $\sum_{x \in N^3(u)} \frac{\activity(x)}{\Activity(x)} \leq \alpha^3$ and therefore:
\ifthenelse{\boolean{hasLNCSformat}}
{
    \begin{multline}
      \label{eq:weightedturanAAC}
      \sum_{x \in X^{u,v}_m} \activity(x) 
      = \sum_{x \in X^{u,v}_m}
      \frac{\activity(x)}{\Activity(x)} \Activity(x) 
      \\
      \leq 
      \sum_{x \in
        X^{u,v}_m} \frac{\activity(x)}{\Activity(x)} \BB{\AAC(x) +
        \frac{1}{2}} \leq 10(m+1)\alpha^3\AAC(u) + \alpha^3
      %
    \end{multline}
}{
  \begin{equation}
      \label{eq:weightedturanAAC}
      \sum_{x \in X^{u,v}_m} \activity(x) 
      = \sum_{x \in X^{u,v}_m}
      \frac{\activity(x)}{\Activity(x)} \Activity(x) 
      \leq 
      \sum_{x \in
        X^{u,v}_m} \frac{\activity(x)}{\Activity(x)} \BB{\AAC(x) +
        \frac{1}{2}} \leq 10(m+1)\alpha^3\AAC(u) + \alpha^3
  \end{equation}
}
  
  We first look at the case $m=0$. 

  \paraclose{Case 1a: \boldmath{$\AAC(u) > 2$, $m=0$.}}
  Then, by (\ref{eq:weightedturanAAC}), $\sum_{x \in X^{u,v}_m} \activity(x) \leq 11\alpha^3\AAC(u)$. 
  Also, (\ref{eq:lb1}) can be upper bounded by $\frac{4}{3}\pi_\ell \frac{\activity(x)}{\AAC(u)}$ (since $\bar{q}^y_\lambda \leq 1$). 
  For (\ref{eq:lb2}) we note that $\AAC(x)\leq 10\AAC(u)$ and get combined that $\Pr(E^x_\lambda|B^{u,v}_\lambda) \leq 10\pi_\ell \frac{\activity(x)}{\AAC(u)}$, regardless whether $x \in N^1_\Astate(v)$ or not.

\ifthenelse{\boolean{hasLNCSformat}}
{
    \begin{multline*}
      \Pr(F^{u,v}_{\lambda,0}|B^{u,v}_\lambda) \geq \prod_{x \in
        X^{u,v}_0} \brackets{1- 10\pi_\ell
        \frac{\activity(x)}{\AAC(u)}}
      \\
      \geq e^{-2\frac{10\pi_\ell}{\AAC(x)}\sum_{x\in
          X^{u,v}_0}\activity(x)}
      \stackrel{(\ref{eq:weightedturanAAC})}{\geq}
      e^{-220\alpha^3\pi_\ell} \stackrel{\pi_\ell \leq \frac{1}{2}}{=}
      \Omega(1).
    \end{multline*}
}{
    \begin{equation*}
      \Pr(F^{u,v}_{\lambda,0}|B^{u,v}_\lambda) \geq \prod_{x \in
        X^{u,v}_0} \brackets{1- 10\pi_\ell
        \frac{\activity(x)}{\AAC(u)}}
      \geq e^{-2\frac{10\pi_\ell}{\AAC(x)}\sum_{x\in
          X^{u,v}_0}\activity(x)}
      \stackrel{(\ref{eq:weightedturanAAC})}{\geq}
      e^{-220\alpha^3\pi_\ell} \stackrel{\pi_\ell \leq \frac{1}{2}}{=}
      \Omega(1).
    \end{equation*}
}

  \paraclose{Case 1b: \boldmath{$\AAC(u) > 2$, $m\geq 1$.}}
  We show that, \wcp, all $x \in \bigcup_{m\geq 1} X^{u,v}_m$, that listen on channel $\Achannel_\lambda$, have a collision.
  For $E^x_\lambda$ to happen, exactly one of $x$'s neighbors has to broadcast. We look again first at $x \in N_\Astate(v)\setminus\set{u}$. We plug in the values for $\bar{q}^y_\lambda$ into (\ref{eq:lb1}) and we use that $\AAC(x)\geq 10m\AAC(u)$ to get
  
\ifthenelse{\boolean{hasLNCSformat}}
{
  \begin{multline*}
      \Pr(E^x_\lambda|B^{u,v}_\lambda) 
      \leq 
      \frac{4}{3}\pi_\ell \frac{\activity(x)}{\AAC(u)} 
      \prod_{y \in N_\Astate(x)\setminus\set{u,v}} \Big( 1-(1-\pi_\ell)\frac{\activity(y)}{\AAC(u)}\Big)
      \\
      \leq 
      \frac{64}{27}\pi_\ell \frac{\activity(x)}{\AAC(u)} e^{-\frac{1}{2}\frac{\AAC(x)}{\AAC(u)}}
      \leq
      3\pi_\ell \frac{\activity(x)}{\AAC(u)} e^{-5m}
  \end{multline*}
}{
    \begin{equation*}
      \Pr(E^x_\lambda|B^{u,v}_\lambda) 
      \leq 
      \frac{4}{3}\pi_\ell \frac{\activity(x)}{\AAC(u)} 
      \prodleft{y \in N_\Astate(x)\setminus\set{u,v}}{\Big( 1-(1-\pi_\ell)\frac{\activity(y)}{\AAC(u)}\Big)}
      \\
      \leq 
      \frac{64}{27}\pi_\ell \frac{\activity(x)}{\AAC(u)} e^{-\frac{1}{2}\frac{\AAC(x)}{\AAC(u)}}
      \leq
      3\pi_\ell \frac{\activity(x)}{\AAC(u)} e^{-5m}
    \end{equation*}
}
  For $x \in N_\Astate^3(v)\setminus N_\Astate^1(v)$ we do the same with (\ref{eq:lb2}). 
  We also use $\AAC(x) \leq 10(m+1)\AAC(u)$ for our calculations:

  \begin{eqnarray*}
    \Pr(E^x_\lambda|B^{u,v}_\lambda) 
    &\leq&
    \pi_\ell\frac{\activity(x)}{\AAC(u)} \cdot \frac{\AAC(x)}{\AAC(u)} \prod_{z \in N_\Astate(x)\setminus\set{u,y}} \Big(1-(1-\pi_\ell)\frac{\activity(z)}{\AAC(u)}\Big)
    \\ &\leq& 
    \frac{16}{9}10(m+1)\pi_\ell\frac{\activity(x)}{\AAC(u)}e^{-\frac{1}{2}\frac{\AAC(x)}{\AAC(u)}}
    \leq
    40m\pi_\ell \frac{\activity(x)}{\AAC(u)} e^{-5m}
  \end{eqnarray*}



  For all $x \in X^{u,v}_m$ we therefore have $\Pr(E_\lambda^x|B_\lambda^{u,v})  \leq 40m\pi_\ell \frac{\activity(x)}{\AAC(u)} e^{-5m}$. Recall that by (\ref{eq:weightedturanAAC}) we can upper bound $\sum_{x \in X^{u,v}_m} \activity(x)$ by $22m\alpha^3\AAC(u)$.

  For $F^{u,v}_{\lambda,m\geq 1} := \bigcap_{m\geq 1}F^{u,v}_{\lambda,m}$ we get
\ifthenelse{\boolean{hasLNCSformat}}
{
  \begin{eqnarray*}
    \Pr(F^{u,v}_{\lambda, m\geq 1}|B^{u,v}_\lambda) 
    &\geq& 
    \prod_{m\geq 1} \prod_{x \in X^{u,v}_m} (1-\Pr(E^x_\lambda|B^{u,v}_\lambda))
    \\
    &\geq&
    \prod_{m\geq 1} \prod_{x \in X^{u,v}_m} \Big(1-40m\pi_\ell \frac{\activity(x)}{\AAC(u)} e^{-5m}\Big)
    \\
    &\geq& 
    \prod_{m\geq 1} e^{-80m\pi_\ell e^{-5m}\sum_{x \in X^{u,v}_m} \frac{\activity(x)}{\AAC(u)}}
    \\
    &\geq& 
    \prod_{m\geq 1} e^{-1760\pi_\ell \alpha^3  m^2 e^{-5m}}
    \geq 
    e^{ -1760\pi_\ell\alpha^3 \sum_{m\geq 1}m^2e^{-5m} } 
    \\
    &\geq& 
    e^{-12\pi_\ell\alpha^3}
    =
    \Omega(1).
  \end{eqnarray*}
}{  
  \begin{eqnarray*}
    \Pr(F^{u,v}_{\lambda, m\geq 1}|B^{u,v}_\lambda) 
    &\geq& 
    \prod_{m\geq 1} \prod_{x \in X^{u,v}_m} (1-\Pr(E^x_\lambda|B^{u,v}_\lambda))
    \geq
    \prod_{m\geq 1} \prod_{x \in X^{u,v}_m} \Big(1-40m\pi_\ell \frac{\activity(x)}{\AAC(u)} e^{-5m}\Big)
    \\
    &\geq& 
    \prod_{m\geq 1} e^{-80m\pi_\ell e^{-5m}\sum_{x \in X^{u,v}_m} \frac{\activity(x)}{\AAC(u)}}
    \geq
    \prod_{m\geq 1} e^{-1760\pi_\ell \alpha^3  m^2 e^{-5m}}
     \\
     &\geq&
    e^{ -1760\pi_\ell\alpha^3 \sum_{m\geq 1}m^2e^{-5m} } 
    \geq
    e^{-12\pi_\ell\alpha^3}
    =
    \Omega(1).
  \end{eqnarray*}
}
  Let us now define $F^{u,v}_\lambda$ as $F^{u,v}_{\lambda,m=0} \cap F^{u,v}_{\lambda,m\geq 1}$, then we get:
\ifthenelse{\boolean{hasLNCSformat}}
{
  \begin{multline*}
      \Pr(D^{u,v}_\lambda) = \Pr(B^{u,v}_\lambda \cap F^{u,v}_\lambda)
      \\
      = \Pr(B^{u,v}_\lambda) \Pr(F^{u,v}_{\lambda,m=0}|B^{u,v}_\lambda) \Pr(F^{u,v}_{\lambda,m\geq 1}|B^{u,v}_\lambda) = \Omega(\Pr(B^{u,v}_\lambda))
  \end{multline*}
}{
    \begin{equation*}
      \Pr(D^{u,v}_\lambda) = \Pr(B^{u,v}_\lambda \cap F^{u,v}_\lambda)
      = \Pr(B^{u,v}_\lambda) \Pr(F^{u,v}_{\lambda,m=0}|B^{u,v}_\lambda) \Pr(F^{u,v}_{\lambda,m\geq 1}|B^{u,v}_\lambda) = \Omega(\Pr(B^{u,v}_\lambda))
    \end{equation*}
}
  With the results from above, this concludes the analysis for the case of $\AAC(u)>2$.

  \paraclose{Case 2: \boldmath{$\AAC(u) \leq 2$, $\lambda = 1$ and $2^{-\lambda} = 1/2$}}.





  For this case we redefine $X^{u,v}_m$.
  \[
  X^{u,v}_m := \set{x \in N^3_\Astate(v)\setminus\set{u,v}: 10m \leq \AAC(x) < 10(m+1)}
  \]
  Analogously to (\ref{eq:lb1}) and (\ref{eq:lb2}) we get (with $2^{-\lambda}=1/2$):

  \begin{eqnarray}
    x \in N_\Astate(v)\setminus\set{u}: 
    &\quad&
    \Pr(E^x_\lambda|B^{u,v}_\lambda) 
    =
    \Pr(E^x_1|B^{u,v}_1) 
    \leq 
    \frac{2}{3}\pi_\ell \activity(x) \prodleft{y \in N_\Astate(x)\setminus\set{u,v}}{\bar{q}_\lambda^y}
    \label{eq:lb3}
    \\
    x \notin N^1_\Astate(v):
    &\quad&
    \Pr(E^x_\lambda|B^{u,v}_\lambda) 
    =
    \Pr(E^x_1|B^{u,v}_1) 
    \leq
    \frac{\pi_\ell}{4}\activity(x) \AAC(x) \prodleft{z \in N_\Astate(x)\setminus\set{u,y}}{\bar{q}^z_\lambda}
    \label{eq:lb4}
  \end{eqnarray}

  Like with (\ref{eq:weightedturanAAC}), Lemma \ref{lemma:weightedturan} gives us that 
  \begin{equation}
    \label{eq:weightedturanAAC2}
    \sum_{x \in X^{u,v}_m} \activity(x) \leq 10(m+1)\alpha^3+\frac{1}{2}\alpha^3 \leq 11(m+1)\alpha^3.
  \end{equation}
  \paraclose{Case 2a: \boldmath{$\AAC(u) \leq 2$, $m=0$.}}
  
  From \ref{eq:weightedturanAAC2} we get $\sum_{x \in X^{u,v}_0} \activity(x) \leq 11\alpha^3$. With $\bar{q}^y_\lambda = \bar{q}^y_1 \leq 1$ and $\AAC(x) \leq 10$ we get for all $x \in X^{u,v}_0$ that $\Pr(E^x_\lambda|B^{u,v}_\lambda) = \Pr(E^x_1|B^{u,v}_1) \leq 3\pi_\ell \activity(x)$.

  Thus,
  \begin{equation*}
    \Pr(F^{u,v}_{\lambda,0}|B^{u,v}_\lambda) 
    \geq \prod_{x \in X^{u,v}_0} (1- 3\pi_\ell \activity(x)) 
    \geq e^{-66\pi_\ell\alpha^3 } = \Omega(1)
  \end{equation*}

  \paraclose{Case 2b: \boldmath{$\AAC(u) \leq 2$, $m\geq 1$.}}

  Note that since $\lambda=1$ we have $q^w_1 = (1-\pi_\ell)\activity(w)2^{-\lambda} \geq \activity(w)/4$ and therefore $\bar{q}^w_1 \leq e^{-\activity(w)/4}$. We use (\ref{eq:lb3}) and (\ref{eq:lb4}) to show that for \emph{any} $x \in X^{u,v}_m$

  \begin{equation*}
    \Pr(E^x_\lambda|B^{u,v}_\lambda) \leq \frac{1}{4} \pi_\ell \activity(x) 20m \Bigbrackets{\frac{4}{3}}^2\prodleft{z \in N_\Astate(x)}{\bar{q}^z_\lambda}
    \leq
    10m \pi_\ell \activity(x) e^{-2.5m}
  \end{equation*}

  As indicated by (\ref{eq:weightedturanAAC2}), $\sum_{x \in X^{u,v}_m} \activity(x) \leq 22m\alpha^3$ for $m\geq 1$.\\
  For $F^{u,v}_{\lambda,m\geq 1} = \bigcap_{m\geq 1}F^{u,v}_{1,m}$ we thus get again
  \begin{eqnarray*}
    \Pr(F^{u,v}_{1, m\geq 1}|B^{u,v}_1) 
    &\geq& 
    \prod_{m\geq 1} \prod_{x \in X^{u,v}_m} (1-\Pr(E^x_1|B^{u,v}_1))
    \geq
    \prod_{m\geq 1} \prod_{x \in X^{u,v}_m} \bigbrackets{ 1-10m\pi_\ell \activity(x) e^{-2.5m} }
    \\
    &\geq& 
    \prod_{m\geq 1} e^{-20m\pi_\ell e^{-2.5m}\sum_{x \in X^{u,v}_m} \activity(x)}
    \geq
    \prod_{m\geq 1} e^{-440m\pi_\ell \alpha^3 e^{-5m}}
    \geq 
    e^{ -440\pi_\ell\alpha^3 \sum_{m\geq 1}me^{-2.5m} } 
    \\
    &\geq& 
    e^{-50\pi_\ell\alpha^3}
    =
    \Omega(1).
  \end{eqnarray*}

  Analogously to the case $\AAC(x)>2$ it holds that $\Pr(D^{u,v}_1) = \Omega(\Pr(B^{u,v}_1))$.




  Let $c$ be the constant such that $\Pr(D^{u,v}_\lambda) \geq c \Pr(B^{u,v}_\lambda)$ for any $v \in N_\Astate(u)$. Then, since $D^u_\lambda = \bigdotcup_{v \in N_\Astate(v)} D^{u,v}_\lambda$ and $B^u_\lambda = \bigdotcup_{v \in N_\Astate(v)} B^{u,v}_\lambda$, it most hold that $\Pr(D^u_\lambda) \geq c \Pr(B^u_\lambda)$.
  Also, $\Pr(D^u) \geq \Pr(D^u_\lambda)$ and since by (\ref{eq:BulambdaUB}) $\Pr(B^u_\lambda)=\Omega(\Pr(B^u))$, we get that $\Pr(D^u)=\Omega(\Pr(B^u))$.

  \hide{
    We now focus on events $B^{u,v}_i$. Similar to previous calculations we get.
    \begin{equation}
      \label{eq:BuviLB}
      \Pr(B^{u,v}_i) 
      \geq \frac{1}{2}\pi_\ell \activity(u)\activity(v)2^{-2i} e^{-2\AAC(u)2^{-i}}
    \end{equation}
    Event $D^{u,v}_\lambda$ is the happening of both events $B^{u,v}_\lambda$ and $F^{u,v}_\lambda$. Let $\AAC(u) > 2$, then
    \begin{equation}
      \Pr(D^{u,v}_\lambda) = \Pr(B^{u,v}_\lambda) \Pr(F^{u,v}_\lambda | B^{u,v}_\lambda)
      = \Omega(\pi_\ell \frac{\activity(u)}{\AAC(u)}).
    \end{equation}
    For $\AAC(u) \leq 2$ we similarly get
    \begin{equation}
      \Pr(D^{u,v}_\lambda) = \Pr(D^{u,v}_1) = \Omega(\pi_\ell \activity(u)\AAC(u)).
    \end{equation}
  }

  This finishes the proof for the second part of the claim, equations (\ref{eq:claimequationslower}).
  \qedLNCS
\end{proof}




\ifthenelse{\boolean{swapheraldclaim}}{
\begin{corollary}
  \label{corollary:heraldwcpgoodherald}
  Let $B(r)^{u,v}$ be the event that in round $r$ node $u \in \Astate_r$ receives a message from one of its neighbors $v\in \Astate_r$, neither of them neighboring any leader, herald or herald candidate in the $5$th or $6$th round of its handshake protocol. Let $\hat H(r')^{u',v'}$ be the event that at the beginning of round $r'$ node $u' \in \Hstate_{r'}$ and $v' \in \Lstate_{r'}$ form a good \lhp and that $\Hstate' \cap N^3(u') = \emptyset$, i.e., there are no herald candidates in the $3$-neighborhood of $u'$. Then
  \begin{equation}
    \Pr(\hat H(r+8)^{u,v}|B(r)^{u,v}) = \Omega(1)
  \end{equation}
\end{corollary}
\begin{proof}
  According to \Cref{lemma:heraldbounds}, 
}{
We are now ready to prove Lemma \ref{lemma:app:heraldwcpgoodherald}.

\begin{proof}
[of Lemma \ref{lemma:app:heraldwcpgoodherald}]
  According to \Cref{claim:heraldbounds}, 
}
since $D^u \subset B^u$, it holds that $\Pr(D^u|B^u) = \Omega(1)$. Thus, if $B(r)^{u,v}$ happens, where $u$ receives the message on channel $\Achannel_i$, then, \wcp, no other node $w' \in N^3(u)\setminus\set{u}$ (which includes $v$'s neighborhood) receives a message on $\Achannel_i$. If no other node $w' \in N^3(u)\setminus\set{u}$ receives a message on $\Achannel_i$, that means that all such nodes $w'$ either not operate on $\Achannel_i$, or, if they do, none or at least $2$ of their neighbors $w'' \in N^4(u)$ send a message on $\Achannel_i$. Using Lemma \ref{lemma:wcpnoherald}.(\ref{lemma:wcpnoherald:4}) by setting $k=4$, $\pi_\ell$ appropriately small and conditioning on the set $S=S^l \dcup S^b \dcup S^n$ of nodes as listening/broadcasting on channel $\Achannel_i$ or not operating on $\Achannel_i$ at all (denoted by event $Z:=\mathcal{S}^n_{\neg i} \wedge \mathcal{S}^b_i \wedge \mathcal{S}^l_i$), we get that $\Pr(H_{\neg i}|D^u)=\Pr(H_{\neg i}|Z) > 1/2$. 

  If in round $r$ both events $H_{\neg i}$ and $D^u$ happen (where $u$ gets its message on channel $\Achannel_i$), then in round $r+1$ nodes $u$ and $v$ meet on channel $\Hchannel$ and can perform the first round of the handshake protocol without any interruption from any other nearby nodes. By applying Lemma \ref{lemma:wcpnoherald}.(\ref{lemma:wcpnoherald:1}) multiple times, in rounds $r+1$ to $r+7$, \wcp, no other new herald is created in $N^3(u)$. Thus, \wcp, by the beginning of round $r+8$, $u$ and $v$ have emerged from the handshake as a good \lhp.
  \qedLNCS
\end{proof}



\ifthenelse{\boolean{swaphandshake}}{
\input{hf-analysis/handshake-redblue-summary-short}
}{
  We next recall the analysis done for the so-called \emph{Handshake} and \emph{Red-Blue Game}, which shows that the creation of good pairs implies progress.
  All lemmas for the handshake protocol and the \rbp have been proven in \cite{tr:podc2013} and their proofs are almost unaffected by the changes made to the algorithm.

\subsection{Handshake Protocol---Nodes in States $\Lstate'$ and $\Hstate'$}
\label{subsec:handshake}

\begin{lemma}\label{lemma:adjpairs}
  In round $r$ consider two \lhps ($l_1$, $h_1$) and ($l_2$, $h_2$) and suppose that the pairs started their most recent handshakes in rounds $r_1$ and $r_2$, $r_1\leq r_2$, respectively. Say that edge $e$ is \emph{crossing} if one of its endpoints is in $\{l_1, h_1\}$ and its other endpoint is in $\{l_2, h_2\}$. Then, either no crossing edge exists or exactly one of the following conditions is satisfied: (1) $r_1 = r_2$ and crossing edges are $\{l_1,l_2\}$ and/or $\{h_1,h_2\}$, (2) $r_2 = r_1 +2$ and the only crossing edge is $\{l_1,h_2\}$.
\end{lemma}

  Lemma \ref{lemma:adjpairs} corresponds to Lemma 8.10 in \cite{tr:podc2013} and in the proof there it was made use of the fact that in the original \rbg in every second round both leader and herald blocked channel $\Hchannel$. But the new \rbg features exactly the same mechanic. Everything else in the proof stays the same as the handshake protocol has not changed at all, hence we omit the proof here.

\subsection{\RBP---Nodes in States $\Lstate$ and $\Hstate$}
\label{subsec:redblue}

\begin{lemma} 
  \label{lemma:goodpair}
  If a pair $(l,h)$ is good in round $r$ and they started their first \rbg\ in round $r'$, then by the end of round $r'+\tthreshold_\redblue = r'+\Theta(\log n)$, \whp, either
  \begin{compactitem}
    \item the related leader $l$ joins the MIS, or
    \item a node $v \in N(l) \cup N(h)$ joins the MIS by increasing its lonely counter above $\tthreshold_\lonely$.
  \end{compactitem}
\end{lemma}

This lemma and its proof are exactly the same as in \cite{tr:podc2013}, hence we omit any proof. Note that the lemma stresses that a good leader does not necessarily become an MIS node. To see how the second case can happen, assume that some node $v$ is in the \hfilter already for a long time, with its lonely counter almost reaching the threshold. If now a neighboring node in $G$ makes it out of the \dfilter, and, after increasing its activity value for a while, communicates successfully with another node. This can happen before $v$ notices the existence of $u$, and shortly after $u$ becomes a leader, $v$ could join the MIS via the lonely counter.

Leaders of bad pairs do not have a justified claim for being MIS nodes, so we do not want them to hinder progress. But they do that by decreasing activity values of their neighbors. The next lemma shows that bad pairs do not last for very long, thus not causing a long stagnation.

\begin{lemma} 
  \label{lemma:extinction} 
  Consider a node $v$ and suppose that in an arbitrary round $r$, there is a leader or herald of a bad pair in $N^3(v)$. Then, \wcp, in round $r+16$, 
  no node in $N^3(v)$ is in state $\Hstate'$ and all leaders and heralds are part of a good pair.
\end{lemma}

  The proof for Lemma \ref{lemma:extinction} 
remains the same, too, 
except that in the original version round $r+12$ was stated. This change is due to the fact that the length of a single \rbg did increase from $6$ to $8$. It does not affect the proving method and we again omit any proof.




}

\subsection{Joining the MIS---Nodes in States $\Mstate$ and $\Estate$}
\label{subsec:joining:the:MIS}


\begin{property}[P]
  The set $\Mstate$ is an independent set at all times.
\end{property}

This intuitive assumption is needed for some of the upcoming statements; it is clearly
true at the beginning of the algorithm, when $\Mstate=\emptyset$. We show in Lemma \ref{lemma:safety}, that if (P) is violated, then \whp a contradiction occurs. The next lemma makes sure that nodes 
in $N(\Mstate)$
soon learn of their coverage. 

\begin{lemma}
  \label{lemma:domination}
  Assume (P) holds. Let $v$ be a node that enters state $\Mstate$ at time $t$. Let $w$ be a node in $N_G(v)$ that is awake at time $t'\geq t$ and, if $w \in \Lstate \cup \Hstate$, that it is at most in round $\frac{9}{10}\tau_\redblue$ of its corresponding \rbg. Then by time $t'+\tthreshold_\notif=t'+\bigO(\log n)$, \whp, $w$ is in state $\Estate$.\footnote{Note that this lemma also considers nodes $w$ from the \dfilter.}
\end{lemma}


\begin{proof}
  We proof the statement by showing that within $\bigO(\log n)$ rounds a non-sleeping node $w$ receives a message from $v \in \Mstate$ on some channel $\Rchannel_i$. Since nodes in the \dfilter listen on the report channels in any given round with probability at least $1/2$, and therefore at least half as often as nodes in the \hfilter, we restrict our analysis to nodes in the \hfilter and proof the statement for a time bound of $\tthreshold_\notif/2$. State changes from the \dfilter to the \hfilter do not affect the analysis because of the very same reason. Also note that $w$ is \emph{unable} to move to state \Lstate nor \Hstate after time $t+4$, since $v$ disrupts any handshake in its neighborhood by sending at least once on \Hchannel in every set of two consecutive rounds.

  {\bf Case $\mathbf{w \in \Astate}$.} Consider some round $t'' \geq t'$. If $w$ is in state $\Astate$, \wcp, it also is in that state in round $t''+1$. Further, \wcp, MIS node $v$ has its variable $\enforce$ set to false in round $t''+1$, and thus, \wcp, broadcasts on some channel $\Rchannel_i$ in that round. Assume $w$ does not neighbor any bad herald or bad leader, then it neighbors at most $\alpha$ good leaders, $\alpha^2$ good heralds, and $\alpha$ MIS nodes. To see this, note that while good heralds are allowed to be adjacent to each other, 
  each has a neighboring good leader, and thus the number of adjacent good heralds in the direct neighborhood of $w$ is upper bounded by the number of good leaders in $N^2(w)$, and thus by $\alpha^2$. We can therefore upper bound the number of adjacent good leaders, good heralds, or MIS nodes by $3\alpha^2$. The probability that $v$ chooses channel $\Rchannel_i$ while no good herald, good leader or other MIS node neighboring $w$ chooses to operate on the same channel $\Rchannel_i$, is constant, since $n_\Rchannel$ is greater than $3\alpha^2$, but still a constant. With probability at least $\frac{1}{2n_\Rchannel}=\Omega(1)$, $w$ listens on $\Rchannel_i$ in round $t''+1$, and therefore, $w$ learns of $v$'s state with constant probability.

  Now let there be bad pairs in $w$'s neighborhood in round $t''$ and $w$ be in $\Astate$. Then, \wcp, by Lemma \ref{lemma:extinction}, $16$ rounds later $w$ is still/again in state \Astate\ while all bad pairs in $N^2(w)$ are knocked out and no new bad pair has been created due to Lemma \ref{lemma:wcpnoherald}.(\ref{lemma:wcpnoherald:1}). As before, $w$ learns of $v$'s state \wcp after $\bigO(1)$ rounds. 

  Let now $w$ be in different states. If $w$ is in $\Hstate'$ (but at most in the fourth round of the handshake protocol) or in $\Lstate'$, its handshake will fail due to $v$'s routine of disrupting channel $\Hchannel$ at least once every $2$ rounds, reverting $w$ back to state $\Astate$. 
  For the cases of $w$ being either in $\Lstate$, $\Hstate$ or in the last two rounds of the handshake protocol as a herald candidate ($\Hstate'$) (which we denote by $\Hstate'_{5,6}$ in this proof), note that if $w$ ever moves to state $\Astate$ from these states, it is unable to return.
  If $w$ is in $\Lstate$, in each \rbg, there is a constant probability of $\Theta(1)$ that $v$ disrupts the game, bringing $w$ back to state $\Astate$. Instead, if $w$ is in $\Hstate$, during each \rbg, there is a $\Omega(1/n_\Rchannel)$ probability that $v$ operates on the same channel $\Rchannel_\rendezvous$ as $w$ does in round $8$ of its respective \rbg, disrupting the ongoing \rbg\ and sending $w$ back to state $\Astate$. Due to our condition of $w$ not being too far in its \rbp, we can choose $\tau_\redblue$ large enough to make sure that $w$, whether it is in state $\Lstate$ or $\Hstate$, hears from $v$ \whp before it can join the MIS itself. If $w$ is in the last two rounds of the handshake protocol, then this case can be reduced to $w$ being a herald.

  Thus, if $w$ is in $\Lstate \cup \Hstate \cup \Hstate'_{5,6}$, then it leaves this set within $\bigO(\log n)$ rounds and never returns to it. Thus let us assume that $w$ is not in this set (anymore).

  Clearly, if $w$ is not in \Astate in any round $t'' \geq t'$, \wcp, it returns to \Astate in $\bigO(1)$ rounds. And as argued above, if bad pairs are in $w$'s neighborhood, they get eliminated \wcp in $\bigO(1)$ rounds as well. Choosing $\tthreshold_\notif=\bigO(\log n)$ sufficiently large and applying a Chernoff bound proves the statement.
\qedLNCS
\end{proof}


\subsection{Progress and \Runtime}
\label{subsec:progress}

  In Lemma 8.13 of \cite{tr:podc2013} we have shown that once a good \lhp is created, its leader (or another node in distance $2$) joins the MIS within $\bigO(\log n)$ rounds. Also, we used the fact that within close proximity of \emph{fat} nodes (which exist in any $\delta$-neighborhood of any node in the \hfilter) such solitary pairs are created with a constant probability. In the algorithm in \cite{tr:podc2013} it might happen that after a good pair is created within radius $\delta$ of some node $u$, the only fat node within $N^\delta(u)$ is close to that pair. A good pair blocks the creation of other pairs around it, so progress might be stalled until the leader of the good pair joins the MIS, causing it to eliminate its neighbors (and therefore their activity) and finally, forcing the local condition of fatness to move to a different area of the graph.

  Here we changed the algorithm to take care of this potential stagnancy issue. We want the attribute of fatness to move away from a good pair long before the leader joins the MIS. More precisely, a node not neighboring good pairs should become fat within $o(\log n)$ rounds. For this we require good pairs to reduce the activity levels of their neighborhoods. However, a \lhp does not know whether it forms a good pair or a bad one before the $\tau_\redblue = \Theta(\log n)$ \rbgs are over. The idea to deal with this difficulty is the following. Good pairs manage in expectation within $\bigO(\loglogn)$ rounds to reduce their neighborhood's activity far enough such that most of the time those nodes can be considered inactive. Bad pairs, however, last for only a constant number of rounds in expectation, and are created rarely enough\footnote{controlled by reducing the parameter $\pi_\ell$} for affected nodes to recover their lost activity quickly. In other words, the longer a node is a leader, the more likely it is that this node is good.
  
  Careful analysis allows us to transform these observations into high probability results.






In the following $\activity(u,t)$ denotes the activity level of node $u$ in round $t$. 
Also, let $\eps$ be a constant smaller than $1$---about $0.1$ is sufficiently small for the upcoming proofs.
\begin{lemma}
  \label{lemma:alwaysactive}
  Let $t$ be a time at which a node $u \notin N^1(\Mstate)$ is in the \hfilter. Then, \whp, one of following holds:
  \begin{compactenum}[(a)]
    \item Within $\tthreshold_\progress=\bigO(\misRTstrong)$ rounds, 
    $u \in N^1(\Mstate)$, or
    \item $|\set{t' \in [t+1, t+\tthreshold_\progress]: \activity(u,t')=1/2}| \geq (1-\eps) \tthreshold_\progress$.
  \end{compactenum}
\end{lemma}

\begin{proof}
  Initially assume that the lemma allows in addition to conditions (a) and (b) the following:
  \begin{compactenum}[(c)]
    {\it \item or within $\tthreshold_\progress=\bigO(\misRTstrong)$ rounds, 
      there is a good leader in $N^2(u)$.}
  \end{compactenum}
  We prove the statement by contradiction, thus assume that neither (a), (b) nor (c) holds.

  Let $t < T_1 < T_2 < \dots < T_m \leq t+\tthreshold_\progress-\tau_\redblue$ be the rounds in which a respective series of leaders or heralds $v_1, v_2, \dots, v_m \in N(u)$ neighbors $u$ for the first time, i.e., $v_i$ successfully finished its handshake in round $T_i-1$ and managed to reach/was reached by a single node $v'_i \in N^2(u)$ in round $T_i-7$ on some channel $\Achannel_{k_i}$. For $T_i < T_j$ it can hold that $v_i=v_j$ --- in this case node $v_i$ moved from $\Lstate \cup \Hstate$ to state \Astate in the time interval $[T_i, T_j]$. Also, at time $T_i$ more than one leader/herald could neighbor $u$ for the first time, in which case we let $v_i$ be any of these. For the pair $(v_i,v'_i)$ denote with $l_i$ the leader and with $h_i$ the corresponding herald.

  Assume that all corresponding pairs $(l_i, h_i)$ are bad pairs as otherwise the lemma would be trivially fulfilled. 

  Fix $i$. 
  $l_i$ can become a bad leader in round $T_i$, if in round $T_i-7$ a node $l' \in N_\Astate(l_i)$ reaches another node $h' \in \Astate$ in that round and manages to get through its handshake protocol as well. For $(l_i,h_i)$ to get through their handshake protocol, they cannot neighbor a leader or a herald or a herald candidate in round $5$ or $6$ at that time. By \Cref{corollary:heraldwcpgoodherald}, the probability for $(l_i, h_i)$ to turn out a bad pair is in $1-\Omega(1)$. 
  Another way for $l_i$ to become a bad leader is if in round $T_i-5$ a node $h' \in N(l_i)$ successfully receives a message. By Lemma \ref{lemma:wcpnoherald}.(\ref{lemma:wcpnoherald:1}) this happens only with probability $\bigO(\pi_\ell)$.
  Hence, in both cases, the probability for $l_i$ to be a bad leader conditioned on the event that $l_i$ becomes a leader in round $T_i$, is at most $1-\Omega(1)$.
  Since we assume that all leaders $l_1,l_2,\dots, l_m$ are bad leaders, a Chernoff bound then gives us that, \whp, $m = \bigO(\log n)$. Assume this is the case. Let $X_i$ be the random variable that measures the number of \rbgs leader $l_i$ survives before it either becomes part of a good pair or it gets knocked out. If $l_i$ becomes part of a good pair 
  then the extended lemma statement would be fulfilled, hence we assume otherwise. Note that, \whp, $l_i$ (like any other node) cannot finish its \rbp as a bad leader and join the MIS.




  So, for all $i$, let us assume that $l_i$ gets knocked out of state \Lstate while being part of a bad pair. According to Lemma \ref{lemma:extinction} this happens in each \rbg with constant probability, i.e., $\E[X_i]=\Theta(1)$. Let $Y$ be the number of rounds in which at least one node in $N(u)$ plays the \rbg. Then, since $m=\bigO(\log n)$, applying Chernoff once more we have for $Y \leq X:= \sum_{i=1}^m X_i$ that $\E[Y]\leq \E[X]=\bigO(\log n)$. If $u$ always picks a channel $R_\rendezvous$ for the corresponding rounds on which a neighboring bad pair communicates, then $u$ reduces its activity by the factor $\decreaseactivity$ in each such round. 
  This totals to a reduction of at most $\decreaseactivity^Y$, spread over $\tau_\progress$ rounds.
  But since $Y=\bigO(\log n)$ it also takes at most $Y \cdot \log_{\increaseactivity} \decreaseactivity =Y \cdot 20\actm =  \Theta(Y) = \bigO(\log n)$ rounds to recover from those activity reductions. Choosing parameter $\tthreshold_\progress$ large enough, $|\set{t' \in [t+1, t+\tthreshold_\progress]: \activity(u,t')=1/2}| \geq (1-\eps) \tthreshold_\progress$ holds with high probability.
  
  What remains is to remove condition (c). But this follows from the simple fact that if a good leader is created in $N^2(u)$, then within $\tthreshold_\redblue$ rounds it either joins the MIS or it gets knocked out by an MIS node due to Lemma \ref{lemma:goodpair}. The latter can happen at most $\alpha^3$ times. Thus we can omit condition (c) by extending $\tthreshold_\progress$ additively by $\tthreshold_\redblue$ and then by a factor of $\alpha^3$.
\end{proof}

Next we upper bound the number of rounds in which \emph{any} neighbors of good pairs within distance $\delta$ from $u$ manage to exceed the activity threshold $\activity_\low$.

\begin{definition}
  For a node $u$ and a round $r$ let $I(u,r)$ be the event that
  \begin{compactitem}
    \item \emph{all} nodes $x \in N^\delta(u)$, which neighbor an MIS
    node, are in state \Estate, and
    \item \emph{all} nodes $x \in N^\delta(u)$, which neighbor a good herald $h$ or good leader $l$, $h,l \in (\Hstate \cup \Lstate)\setminus N(\Mstate)$, have $\activity(x) \leq \activity_\low = \sqrt{\activity_{\min}} =\log^{-12} n$ and are neither bad leaders nor bad heralds.
  \end{compactitem}
\end{definition}


\begin{lemma}
  \label{lemma:repression}
  Assume that (P) holds. Further, let $\bar{r}$ be a round in which node $u$ is in the \hfilter and set $J:=[\bar{r}+1, \bar{r}+\tthreshold_\progress]$. Then, \whp, one of the following holds:
  \begin{compactitem}
    \item Within $\tthreshold_\progress=\bigO(\log n)$ rounds, there is an MIS node in $N^1(u)$, or
    \item $|\set{r \in J: I(u,r) \text{holds}}| \geq (1-\eps) \tthreshold_\progress$.
  \end{compactitem}
\end{lemma}

\begin{proof}
  Initially, $\tthreshold_\progress$ is chosen large enough to comply with Lemma \ref{lemma:alwaysactive}.
  
  The proof is divided into four parts. In the first part we show that there are only $\polylog n$ many nodes for which $I(u,r)$ can be violated---i.e., nodes that neighbor good leaders, good heralds or MIS nodes. In the second part we show that each such node hears from one of these neighbors every $\bigO(1)$ rounds with constant probability. Since neighbors of MIS nodes are immediately eliminated upon hearing a message from them, the case of a node neighboring a good pair is the more difficult one. In the third part we argue that a node neighboring good pairs reduces its activity value to $\activity_{\min}$ at least once within $\bigO(\loglogn)$ rounds with considerable probability: $1-\log^{-c} n$ for some constant $c>1$. 
  In the last part we combine those results to show that within $\tthreshold_\progress$ rounds, \whp, $I(u,r)$ is only violated for a small constant fraction of those rounds. 
  
  \paragraph{First part.}
  To count the number of nodes that can violate $I(u,r)$, we have to count the number of good heralds, good leaders and good MIS nodes in $N^{\delta+1}(u)$ in round $r$; we denote that latter set by $W_r \subseteq (\Lstate \cup \Hstate \cup \Mstate) \cap N^{\delta+1}(u)$ and their neighboring nodes in the \hfilter by $N_r:=N(W_r)$. Due to (P) and Lemma \ref{lemma:goodpair}, \whp\ at all times all MIS nodes form an independent set, as do all leaders of good pairs, so there are no more than $2\alpha^{\delta+1}$ of these in $W_r$. The number of heralds of good pairs in $N^{\delta+1}(u)$ is at most the number of good leaders in $N^{\delta+2}(u)$, so in total $|W_r|$ amounts to at most $3\alpha^{\delta+2} \leq 3\alpha^2\sqrt{\log n} = o(\log n)$. Due to the guarantee we get from the \dfilter, each such node has at most 
$\bigO(\log^3 n)$
neighbors in the \hfilter, so $N_r < \log^4 n$. Note that even over $\tthreshold_\progress$ rounds the total amount of those nodes cannot exceed $o(\log^6 n)$.
  
  \paragraph{Second part.}
  An MIS node has to reach all its neighbors only once each, in order for them to fulfill the requirement for $I(u,r)$. Leaders and heralds on the other hand need to inform their neighbors multiple times and continuously over the course of $\tthreshold_\progress$ rounds. 
  In most cases a node $v \in N_r$ has a constant probability to hear of at least one of its neighbors in $W_r$ within a $\bigO(1)$-round time interval. More precisely, we claim that within $\bigO(\loglogn)$ rounds, in expectation, $v$ reduces its activity by a factor polylogarithmic in $n$. We prove this by analyzing the $3$ types of neighbors $v$ can have in $W_r$ and the states $v$ can be in.

  \begin{itemize}
    \item[$v \in \Astate$:] At first let $v$ neighbor an MIS node $s$. $s$ might be forced to broadcast on $\Hchannel$ in round $r$, but then with probability $1/4$ it sends on one of the report channels $\Rchannel_k$ in round $r+1$. Also with probability at least $1/2$, $v$ does not act on one of the active channels $\Achannel_i$ in round $r$ and therefore either gets knocked out or is also in state \Astate in round $r+1$, where it listens on some report channel with probability at least $1/2$. Thus, with probability at least $1/16$, in round $r' \in [r,r+1]$, $s$ sends on a report channel and $v$ listens on a report channel. If that is the case, since $|W_{r'} \cap N^1(v)| \leq n_\Rchannel$, with probability at least $\frac{1}{en_\Rchannel}$, $v$ and $s$ act on the same channel $\Rchannel_k$ while no other neighbor of $v$ in $W_{r'}$ operates on $\Rchannel_k$. Once that happens, $v$ moves to state $\Estate$.\\
    Thus let us now assume that $v$ does not neighbor an MIS node, but at least one good leader $l$. With similar reasoning, there is a round $r' \in [r,r+7]$ in which $l$ broadcasts on some channel $\Rchannel_k$ as a apart of its \rbg, unless $l$ transitioned to state \Mstate or $\Estate$---the first case we already covered, the second does not happen with probability at least $1/2$. With probability at least $2^{-7}$, $v$ is still or again in state $\Astate$ in that round $r'$. With similar argumentation as above, $v$ then gets $l$'s message with probability at least $\frac{1}{2en_\Rchannel}$. In total, in each $8$-round interval, with probability at least $\frac{1}{2^8en_\Rchannel}$, $v$ receives a message from $l$.\\
    If $v$ does not neighbor a good leader nor an MIS node, but at least one good herald $h$, then the same logic applies for an $8$-round interval, unless $h$ gets knocked out---which only happens if its leader joins the MIS or another MIS node is created nearby. 
    But, with probability at least $1/2$, $h$ nor its leader gets knocked out by a neighboring MIS node. If its leader joined the MIS, then just a new MIS node has been created in $N^2(v)$, which can happen at most $\alpha^2$ times and therefore delay $v$ to hear from any of its neighbors by at most $\bigO(\alpha^2)$ rounds during the whole execution of the algorithm. For simplicity we ignore those $\bigO(1)$ rounds. 
    As above, with probability at least $\frac{1}{2^8en_\Rchannel}$, $v$ hears from $h$ in any $8$-round interval.
  \end{itemize}
  In total we get that for every $2^{11}e n_\Rchannel$ rounds that $v$ spends in $\Astate \cap N_r$, in expectation it hears at least once from one of its neighbors in $W_r$.
  At this point we fix $\actm$ to be eight times as large, i.e., $2^{14}e n_\Rchannel$. Note that $\actm$ is a constant depending on $\alpha$ only, via $n_\Rchannel$.
  \begin{itemize}
    \item[$v \in \Hstate \cup \Lstate$:] As long as $v$ is in one of these states, $v$ hears from the node it partnered up with every $8$ rounds as long as the pair remains. In the algorithm it is accounted for that by decreasing $\activity(v)$ every round by a factor of $\increaseactivity^{20}$.
    \item[$v \in \Lstate' \cup \Hstate'$:] 
    If $v$ was already in $\Lstate'$ or $\Hstate'$ when it joined $N_r$, then it might manage to finish its handshake, in which case we refer to the previous case. Otherwise, if its handshake gets disrupted, then $v$ returns to state \Astate at latest $6$ rounds later, which we also already covered. This also implies, that a node switching between states $\Astate$ and $\Hstate'$ or $\Lstate'$ spends at least one eight of these rounds in $\Astate$.
  \end{itemize}
  We get that $v$ hears from one of its neighbors in $W_r$ at least once every $\actm$ rounds, irrelevant of its own state.
  Recalling that $\decreaseactivity=\increaseactivity^{20 \actm}$, $v$ reduces its activity each $8$ rounds by a factor of at least $\increaseactivity^{152}$ in expectation.

  \paragraph{Third part.}
  Let $r_i$ be the round in which the $i$th MIS node or good \lhp in $N^\delta(u)$ is created and we talk of event $\eEvent_{r_i}$; if more than one is created in a single round, we ignore additional ones. 
  In round $r_i$ all nodes $v \in N_{r_i}$ have $\activity(v)\leq 1/2$.
  Let us assume that no other good pair or MIS node is created for $\Omega(\loglogn)$ rounds, as otherwise we assume all rounds between $r_i$ and $r_{i+1}$ as \emph{violated}---more details on that in part four. We want to show that all nodes $v \in N_{r_i}$ decrease their activity quickly. For as long as $v$ is part of a pair, this reduction is guaranteed by design of the algorithm, so we also only consider rounds in which $v$ is not part of a pair. For some arbitrary constant $c$, in the next $c \loglogn$ rounds $v$ can increase its activity by at most a factor of $\increaseactivity^{c \loglogn} = \Theta(\polylog n)$. Let $D_v$ be the random variable that counts the number of times $v$ receives a message of one of its neighboring good leaders or heralds in those $c \loglogn$ rounds. In expectation, $D_v$ increases at least every $\actm$ rounds by one. With Chernoff we get
  \begin{equation}
    \label{eq:lemmarepression:thirdpart}
    \Pr\left(D_v \leq \Bigbrackets{1-\frac{1}{2}} \frac{c \loglogn}{\actm}\right) \leq e^{-\frac{c\loglogn}{4\actm}}
  \end{equation}
  Recall that $|N_{r_i}| \leq \log^6 n$. We choose $c \coloneqq 100 \actm\geq (20+6)\cdot 4\cdot \ln 2 \cdot\actm$. If $D_v > \frac{c\loglogn}{2\actm}$ as in (\ref{eq:lemmarepression:thirdpart}), then $\activity(v)$ decreases enough in those $c\loglogn$ rounds to ``touch'' $\activity_\textmin$ at least once.
  Now we can make a union bound over all nodes in $N_{r_i}$ and we get that with probability 
  $1-\frac{1}{\log^{20} n}$ 
  \emph{all} nodes in $N_{r_i}$ touch $\activity_\textmin$ at least once in those $c \loglogn$ rounds.

  \paragraph{Fourth part.}
  Let us now count the amount of rounds $r$ in which $I(u,r)$ does not hold. 
  We keep the definition of rounds $r_1 < r_2 < \dots < r_k$ and of events $\eEvent_{r_i}$.
  Due to the definition of good pairs, there can be no more good heralds in $N^{\delta+1}(u)$ than there are good leaders in $N^{\delta+2}$. Once being part of a good pair, a node can only stop being good by being knocked out by an MIS node, which then prevents the creation of new leaders, heralds and MIS in its neighborhood. MIS nodes in distance $\delta+3$ can still influence good leaders and heralds in $N^{\delta+1}(u)$. However, no more than $\alpha$ good leaders (and their corresponding good heralds) and no more than $\alpha^2$ good heralds (and their corresponding good leaders) can be knocked back to state \Astate by an MIS node. 
  Therefore, no more than $2\alpha^{\delta+5} = \bigO(\sqrt{\log n})$ such events $\eEvent_{r}$ can happen, i.e., $k = \bigO(\sqrt{\log n})$. 


  We split the interval $[\bar{r},\bar{r}+\tthreshold_\progress]$ into $\ell$ smaller intervals $(J_i)_{1 \leq i \leq \ell}$ of length $c\loglogn$ each, i.e., $\ell=\frac{\tthreshold_\progress}{c\loglogn}$. Then we color each such interval $J_i$ \emph{red} if it contains one of the events $\eEvent_{r}$ and we color $J_i$ \emph{orange} if its preceding interval $J_{i-1}$ contains such an event, but not $J_i$ itself. All other intervals are colored \emph{blue}. From the third part we know the following. Independently of the activity values of all nodes in $N_{r}=N(W_{r})\setminus W_{r}$ in some round $r$, if no event $\eEvent_{r'}$ happens in round $r' < r+c\loglogn$, then with probability $1-\log^{-20} n$ all those nodes touch $\activity_\textmin$ within the next $c\loglogn$ rounds after round $r$. I.e., in every orange and blue interval, this is likely to happen at least once. If this happens in some orange or blue interval $J_i$, let us call this event $\mEvent_i$. Any node that touches $\activity_\textmin$ in interval $J_i$ cannot recover its activity by a factor higher than
  $$
  \increaseactivity^{c\loglogn} 
  = 2^{\frac{24 c \loglogn}{2000 \actm}} 
  = \log^{\frac{24\cdot100\actm}{2000 \actm}} n
  =
  \log^{1.2} n,
  $$
\hide{
  $$
  \increaseactivity^{c\loglogn} 
  = 2^{\frac{6\ttkdelta c \loglogn}{2000 \actm}} 
  = \log^{\frac{6c\ttkdelta}{2000 \actm}} n
  \stackrel{c \geq (80+4\kappa')\actm\ln 2}{\leq}
  \log^{\frac{3\ttkdelta(320 \actm+80\actm)}{2000 \actm}}n
  \stackrel{\actc \geq 20}{\leq}
  \log^{0.6\ttkdelta} n,
  $$
}
  until the end of interval $J_{i+1}$.
  In other words, if $J_i$ is blue, then $J_{i-1}$ cannot be red and if $\mEvent_{i-1}$ happened, then $I(u,r)$ holds throughout the whole interval $J_i$. Let $\iEvent$ be the index set for all blue intervals. Since $k=\bigO(\sqrt{\log n})$, the total number of rounds in red and orange intervals are both in $o(\log n)$. Thus the number $\ell' \coloneqq |\iEvent|$ of blue intervals is in $(1-o(1))\ell$, i.e., $\ell' \geq (1-\eps/2)\ell$. We define random variables $X_i$ that evaluate to $1$ if $\mEvent_{i-1}$ does \emph{not} hold and to $0$ otherwise. Let $X \coloneqq \sum_{i \in \iEvent} X_i$, $p \coloneqq \log^{-20} n$, $\eps' \coloneqq \eps/2$ and $\tau_\progress \geq \frac{c}{\eps'(1-\eps')}\log n$, then
  \begin{eqnarray*}
    \Pr(X \geq \eps' \ell') 
    &\leq& 
    \binom{\ell'}{\eps' \ell'} p^{\eps' \ell'}
    \stackrel{\binom{n}{k} \leq \bb{\frac{en}{k}}^k}{\leq}
    \left(\frac{e \ell'}{\eps' \ell'}\right)^{\eps' \ell'} \log^{-20\eps' \ell'} n
    \\
    &\leq&
    2^{((\log \frac{e}{\eps'}) -20 \loglogn)) \eps' (1-\eps')\frac{\tthreshold_\progress}{c\loglogn}}
    \leq
    n^{-19}.
  \end{eqnarray*}
  Hence the number of intervals of length $c\loglogn$, in which no violation of $I(u,r)$ occurs, is \whp at least $(1-\eps)\ell$, which concludes the proof.
\end{proof}





\begin{lemma}
  \label{lemma:progress}
  Assume that (P) holds. Let $t_u$ be a round in which a node $u \notin N^1(\Mstate)$ has a neighbor $u' \notin N^1(\Mstate)$ in the \hfilter.
  Then, \whp, within $\tthreshold_\progress = \bigORT$ rounds a node in $N^1(\set{u,u'})$ joins the MIS.
\end{lemma}


\begin{proof}
  The constant $\tthreshold_\progress$ is chosen by the end of this proof, and it is in $\bigORT$. Initially we only set it to be large enough to comply with Lemma \ref{lemma:alwaysactive} and Lemma \ref{lemma:repression}.
  Assume that the statement is not true, i.e., no node in $N^1(\set{u,u'})$ joins the MIS in the given time bounds. 

  
  Both nodes $u$ and $u'$ stay competitive, i.e., they stay in the states $\Astate$, $\Hstate'$, $\Lstate'$, $\Hstate$ or $\Lstate$, as states \Mstate and \Estate would clearly imply the creation of an MIS node in $N^1(\set{u,u'})$.
%
  We now apply Lemma \ref{lemma:alwaysactive} to both nodes to get
  \begin{align}\label{eq:uuprimeactive}
    \left|\set{t' \in [t_u+1, t_u+\tthreshold_\progress]: \activity(u,t')=\activity(u',t')=1/2}\right| \quad\geq\quad (1-\eps) \tthreshold_\progress .
  \end{align}
  

  We also apply Lemma \ref{lemma:repression} to extend (\ref{eq:uuprimeactive}) to
  \begin{equation}\label{eq:numbergoodrounds}
    \left|\set{t' \in [t_u+1, t_u+\tthreshold_\progress]: \activity(u,t')=\activity(u',t')=1/2 \wedge I(u,t') \text{ holds}}\right| \quad \geq\quad (1-2\eps) \tthreshold_\progress .
  \end{equation}

  We call rounds $r$ in which $\activity(u)=\activity(u')=1/2$ and $I(u,r)$ holds, a \emph{promising round}. In such a round clearly $\Activity(u) \geq 1$. 

  We first look at the case in which a round is promising, 
  $\Activity(u) < 5\alpha$ and $u$ is $\frac{1}{5\alpha}$-fat. By Lemma \ref{lemma:fatgivesherald} with probability $\Omega(\pi_\ell)$ within $7$ rounds either a good pair or an MIS node is created nearby, an event that would contradict our assumption. Hence, as long as $\Theta(\pi_\ell^{-1}\log n)$ such rounds appear, \whp such a contradicting event occurs. We choose $\tthreshold_\progress$ sufficiently large such that $\eps \tthreshold_\progress$ such rounds would cause a contradiction. 

  W.l.o.g. we thus assume that within $\tthreshold_\progress$ rounds, a $(1-3\eps)$ fraction of these are promising, but it does not hold that both $\Activity(u) < 5\alpha$ and $u$ is $\frac{1}{5\alpha}$-fat. For simplicity we exclude these cases from our definition of a promising round.

  \newcommand{\prom}{\textit{prom}}


  In all promising rounds, in distance $\delta$ there must be a node $w$, for which $\Activity(w) \geq 5\alpha$ and $w$ is $\eta$-fat. This is obvious if $\Activity(u) \geq 5\alpha$, as then either $u$ is already that node, or there exists a chain of activity sums, which increase in each step by a factor of at least $\eta^{-1}$. By definition of $\eta$, such a chain has length at most $\delta$. If $\Activity(u) < 5\alpha$, then $u$ cannot be $\frac{1}{5\alpha}$-fat by our renewed definition of a promising round. But since $\Activity(u) \geq 1$, $u$ then must have a neighbor $v$ for which $\Activity(v) \geq 5\alpha$ and we can apply the same logic as before to find that $\eta$-fat node $w$.


  For such a node $w$ the requirements for Lemma \ref{lemma:fatgivesherald} hold.

  Now let $t_{1,\prom}$ be the first promising round, $t_{2,\prom}$ the second and so on.

  Let $T_{1,\prom}$ be the random variable that counts the number of promising rounds $t_{i,\prom}$ until the first time a good pair or an MIS node is created in $N^\delta(u)$ or \emph{has been} created in a non-promising round between $t_{i-1,\prom}$ and $t_{i,\prom}$ (we denote such an event as $A$). I.e., 
  \[
  T_{1,\prom} := \min\set{i: \text{ A happens in }(t_{i-1,\prom},t_{i,\prom}]}.
  \]
  $T_{j,\prom}$ for $j>1$ is defined accordingly. By Lemma \ref{lemma:fatgivesherald} event $A$ happens in expectation at least every $(c_1\pi_\ell)^{-1}$ promising rounds for some constant $c_1$ depending only on $\eta$. 

  Let $M_u$ be the random variable that counts the number of events $A$ happening in $N^\delta(u)$ until one happens in $N^1(\set{u,u'})$, i.e., until either a good herald or good leader is in $N^1(\set{u,u'})$. The number of good leaders and the number of MIS nodes that can co-exist in $N^\delta(u)$ is at most $2\alpha^\delta \leq 2\sqrt{\log n}$. If a node stops being a good leader, then it must have been knocked out by an MIS node, which still limits the total number of these events to $3\sqrt{\log n}$.

  Note also that usually after an event $A$ there is an $\Omega(\loglogn)$ pause until the next promising round happens, because it takes at least that many rounds for good pairs to reduce a neighbor's activity beneath the threshold $\activity_\low$, if that neighbors activity level is in $\Omega(1)$ by the time $A$ happens. We account for that by ``paying'' $c_3\loglogn$ rounds for each such event, for some constant $c_3$.

  Now, using a Chernoff bound we get that $\sum_{i=1}^{M_u} T_{i,\prom} \leq 6\sqrt{\log n} \bb{\frac{1}{c_1\pi_\ell}+c_3\loglogn} + c_2 \log n$, \whp, for some constant $c_2$.
  With suitably large chosen $\tthreshold_\progress$, this is less than $(1-3\eps) \tthreshold_\progress$, causing the last contradiction to our assumption and thus finishing the proof.
\end{proof}



\begin{lemma}
  \label{lemma:safety}
  \Whp, property (P) is not violated throughout the \runtime of the algorithm.
\end{lemma}

\begin{proof}
  Assume in round $\bar{t}$ it happens for the first time that $2$ neighboring nodes $u_1$ and $u_2$ are in $\Mstate$. W.l.o.g., we let $u_1$ enter the MIS before $u_2$ or at the same time and denote those times correspondingly $t_1 \leq t_2 = \bar{t}$. Additionally, for all pairs of nodes that violate (P) in round $t_2$ we choose the pair for which $t_1$, the 'age' of the older MIS node, is minimized, and ties among candidates for those positions are broken by \IDs, first for $u_1$, then for $u_2$.

  There are two ways for a node in $\Astate$ to join the MIS: either by increasing its $\lonely$ counter to $\tthreshold_\lonely$ or by passing through stages $\Lstate'$ and $\Lstate$.

  \begin{enumerate}[(A)]
    \item 
    Let us first assume that both nodes $u_1$ and $u_2$ enter state \Mstate\ by finishing their respective \rbgs. If $t_1 < t_2$, then $u_1$ started its first \rbg before $u_2$ did, but by Lemma \ref{lemma:adjpairs}, they cannot be in state $\Lstate$ at the same time. Thus, $u_1$ joins the MIS before $u_2$ becomes a leader. But as a leader, $u_1$ disrupts handshake channel \Hchannel\ every second round, and then as an MIS node, it does so at least every second round, preventing $u_2$ from ever becoming a leader. Therefore, let $t_1=t_2$. Then $u_1$ and $u_2$ were neighboring leaders for $\tthreshold_\redblue/8 = \Omega(\log n)$ \rbgs, and in each such game having a probability of at least $1/2$ to conflict with each others' \rbg. By Chernoff, \whp, this is not possible.
    \item
    Let us next assume that both nodes enter state \Mstate through their $\lonely$ counters. First assume that $t_2-t_1 \geq \tthreshold_\lonely/2 \geq \tau_\notif$. But then in $[t_1, t_2-1]$ (P) holds true and with Lemma \ref{lemma:domination}, $u_2$ gets eliminated by $u_1$ before it can become an MIS node. Thus, let $t_2-t_1 < \tthreshold_\lonely/2$. But then during rounds $[t_1-\tthreshold_\lonely/2,t_1]$ both nodes are in the active states $\Astate$, $\Hstate'$, $\Lstate'$, $\Hstate$ and $\Lstate$ and they do not neighbor an MIS node. The latter stems from the following. If $u_1$ would neighbor an MIS node $u_3$, then the first violation of (P) would happen in round $t_1$, a contradiction if $t_1 < t_2$. For $t_1=t_2$, nodes $u_1$ and $u_2$ contradict our choice of the pair being investigated, as $u_3$ would have had a higher age. Similar argumentation keeps $u_2$ from neighboring an MIS node. Hence, assume that none of both nodes neighbor an MIS node. 
    Lemma \ref{lemma:progress} tells us that by round $t_1-\tthreshold_\lonely+\tthreshold_\progress \ll t_1$, an MIS node $v$ would arise in $N^1(\set{u_1,u_2})$, with $v \neq u_1,u_2$. The remaining rounds, by Lemma \ref{lemma:domination} suffice for $v$ to eliminate $u_1$ or $u_2$ \whp, again contradicting our initial assumption.
    \item
    Now let $u_1$ join the MIS via loneliness and $u_2$ by being a leader. If $t_2-t_1 \geq \tthreshold_\redblue/2= \Omega(\log n)$, then in each $8$-round \rbg after round $t_1+1$, $u_1$ has a constant chance to disrupt $u_2$'s game on channel $\Gchannel$. Choosing $\tthreshold_\redblue$ large enough guarantees us that, \whp, this cannot happen. If $t_2-t_1 < \tthreshold_\redblue/2$, then in each $8$-round \rbg in $(t_2-\tthreshold_\redblue, t_1)$, there is a constant probability that $u_2$ transmits on a report channel $\Rchannel_i$ on which no other neighbor of $u_1$ transmits. The argumentation is similar to the one in Lemma \ref{lemma:repression}, except that $u_1$ does not neighbor an MIS node yet (see argumentation above w.r.t. to our choice of $u_1$ and $u_2$): 
    There are at most $2\alpha^2$ good leaders and heralds neighboring $u_1$, and \emph{all} bad leaders/heralds get knocked out with constant probability in each \rbg, see Lemma \ref{lemma:extinction}. At the same time $u_1$ listens on the same report channel $\Rchannel_i$ with probability at least $\frac{1}{2n_\Rchannel}=\Omega(1)$. Therefore, with $\tthreshold_\redblue$ large enough, \whp $u_1$ hears from $u_2$ before round $t_1$ and thus resets its $\lonely$ counter.
    \item
    Last switch $u_1$'s and $u_2$'s roles from the previous case. In this case, $u_2$ is neighboring $u_1$ throughout its leadership state, and with analogous argumentation, regardless whether $u_1 \in \Lstate$ or $u_1 \in \Mstate$, $u_2$ hears from $u_1$ with high probability.
  \qedhere
  \end{enumerate}
\end{proof}



Now we have everything at hand to prove Theorem \ref{thm:mis:maintheorem}.

\begin{proof}[\optionalproofparamLNCS of Theorem \ref{thm:mis:maintheorem}]
  As stated earlier, Lemma \ref{lemma:dfilterrunningtime} provides that the \runtime of the \dfilter is within $\bigO\bb{\misRTstrong}$, i.e., for a node $u$ executing Algorithm \ref{algo:decay} (the \dfilter) by that time one node $v \in N^1_G(u)$ enters the \hfilter. From Lemma \ref{lemma:dfiltermaxdegree} we get that over the course of $\bigO(\log^2 n)$ rounds the maximum degree of the graph $G'$ induced by all nodes in the \hfilter is at most $\bigO(\log^3 n)$.

  Let thus $u$ be a node that enters the \hfilter. If it stays lonely for $\tau_\lonely = \Theta\bb{\misRTstrong}$ rounds, then $u$ joins the MIS and we are done. Hence assume that $u$ does hear from a neighboring node $u'$ in the \hfilter before $\tau_\lonely$ rounds have passed. We can now apply Lemma \ref{lemma:progress} to get an MIS node $v$ created within $\tau_\progress=\bigO\bb{\misRTstrong}$ rounds. It either neighbors $u$, in which case within $\tau_\notif=\bigO(\log n)$ rounds $u$ is decided \whp, or it neighbors $u'$, which is also then eliminated in $\tau_\notif$ rounds. That way $u$ can become lonely again. However, since an MIS node has been created in $N^2(u)$, this can happen at most $\alpha^2$ times. Thus, at most $\tau_\itruntime = 2\alpha^2 \tau_\lonely$ rounds after $u$ entered the \hfilter, $u$ is decided.
\end{proof}

\hide{
\begin{proofsketch}[\optionalproofparamLNCS of Theorem \ref{thm:mis:maintheorem}]
  As stated earlier, the \runtime of the \dfilter is within $\bigO\bb{\misRTstrong}$, i.e., for every node $u$ in the \dfilter, by that time one node $v \in N^1_G(u)$ enters the \hfilter. We also know that over the course of $\bigO(\log^2 n)$ rounds the maximum degree of the graph $G'$ induced by all nodes in the \hfilter is at most $\bigO(\log^3 n)$.

  Let thus $u$ be a node that enters the \hfilter. If it stays lonely for $\tau_\lonely = \Theta\bb{\misRTstrong}$ rounds, then $u$ joins the MIS and we are done. Hence assume that $u$ does hear from a neighboring node $u'$ in the \hfilter before $\tau_\lonely$ rounds have passed. We can now apply Lemma \ref{lemma:progress} to get an MIS node $v$ created within $\tau_\progress=\bigO\bb{\misRTstrong}$ rounds. It either neighbors $u$, in which case within $\tau_\notif=\bigO(\log n)$ rounds $u$ is decided \whp, or it neighbors $u'$, which is also then eliminated in $\tau_\notif$ rounds. That way $u$ can become lonely again. However, since an MIS node has been created in $N^2(u)$, this can happen at most $\alpha^2$ times. Thus, at most $\tau_\itruntime = 2\alpha^2 \tau_\lonely$ rounds after $u$ entered the \hfilter, $u$ is decided.
\end{proofsketch}
}


\pagebreak
\bibliography{references}
\bibliographystyle{abbrv}
\pagebreak
\appendix
\bookmarksetupnext{level=-1}
\addappheadtotoc
\Needspace{15\baselineskip}
\section{Properties of Graphs with Bounded Independence}
\label{sec:boundedindep}

We need a few statements about bounded independence graphs. The proofs are provided in \cite{tr:podc2013}.

\begin{lemma}\label{lemma:weightedturan}
  Let $G=(V,E)$ be a graph and assume that every node $u\in V$ has a positive
  weight $w_u>0$. Define $W:=\sum_{v\in V} w_v$ and for
  each $u\in V$, $W_u:=\sum_{v\in N_G^1(u)} w_v$. It then holds that
  \begin{eqnarray}
  \sum_{v\in V} \frac{w_v}{W_v} \leq \alpha(G) \quad \text{and} \label{eq:weightedturan:a}\\
  \sum_{v\in V} w_v\cdot W_v \geq \frac{W^2}{\alpha(G)}, \label{eq:weightedturan:b} 
  \end{eqnarray}
  where $\alpha(G)$ is the independence number of $G$.
\end{lemma}

\begin{lemma}\label{lemma:heavyneighborhoods}
  Let $G=(V,E)$ be a graph and assume that every node $u\in V$ has a
  positive 
  weight $w_u>0$. Define $W:=\sum_{v\in V} w_v$ and
  for each $u\in V$, $W_u:=\sum_{v\in N_G^+(u)} w_v$. Let
  $V_{\mathit{heavy}}\subseteq V$ be the set of nodes $v$ for which
  $W_v\geq \frac{W}{2\alpha(G)}$. The total weight of nodes in
  $V_{\mathit{heavy}}$ is at least
  \[
  \sum_{v\in V_{\mathit{heavy}}} w_v > \frac{W}{2\alpha(G)}.
  \]
\end{lemma}

\section{Candidate Election---Statements and Proofs for $\Astate$ (and $\Lstate'$)}
\label{appendix:candidate:election}

\subsection{Lemma \ref{lemma:wcpnoherald}, ``\Whp nothing happens''}

Let $k$ be a positive integer \emph{constant}, $r$ some round, $u$ some node in the \hfilter in round $r$, $i$ an index from $1, \dots, \nac$, $S$ be a (possibly empty) subset of $N^k(u) \cap \Astate$. Furthermore let $\partial S \subset S$ be \emph{the} subset of $S$ that has connections outside of $S$, but in $N^k(u)$, i.e., $\partial S := S \cap N\big(N^k(u)\setminus S\big)$. At last, let $S=S^n \dotcup S^b \dotcup S^l$ be a partition of $S$. We call the tuple $(k,u,r,i,S)$) a \emph{constellation}. For a constellation the following events are defined for round $r$:

\newcommand{\prebullet}{\hspace{-5mm}}
\newcommand{\postbullet}{\hspace{-0mm}}
\begin{tabular}{lll}
  \prebullet\textbullet & \postbullet$\mathcal{S}_{\neg i}  /  \mathcal{S}^n_{\neg i}$: & \emph{no} node in $S  /  S^n$ operates on $\Achannel_i$, \\
  \prebullet\textbullet & \postbullet$\mathcal{S}_{i}  /  \mathcal{S}^b_{i}  /  \mathcal{S}^l_{i}$: & \emph{all} nodes in $S  /  S^b  /  S^l$ operate on $\Achannel_i$, \\
  \prebullet\textbullet & \postbullet$\partial \mathcal{S}_{\neg i}$: & \emph{no} node in $\partial S$ operates on $\Achannel_i$, \\
  \prebullet\textbullet & \postbullet$H_i$: & \emph{no} node in $N^k(u)\setminus S$ receives a message on channel $\Achannel_i$, \\
  \prebullet\textbullet & \postbullet$H_{\neg i}$: & \emph{no} node in $N^k(u)$ receives a message on some channel $\Achannel_j \neq \Achannel_i$ and \\
  \prebullet\textbullet & \postbullet$H$: & \emph{no} node in $N^k(u)$ receives a message on any channel $\Achannel_1,\dots,\Achannel_\nac$.
\end{tabular}

\begin{lemma}
  \label{lemma:app:wcpnoherald}
  Let $(k,u,r,i,S)$ be a constellation. Then,
  \TabPositions{4.1cm}
  \begin{enumerate}[(1)]
    \item $\Pr(H) =$ \tab $ 1 - \bigO(\pi_\ell \alpha^k)$     
    \label{lemma:app:wcpnoherald:1}
    \item $\Pr(H_{\neg i}|\mathcal{S}_{\neg i}) =$ \tab $ 1 - \bigO(\pi_\ell \alpha^k)$    
    \label{lemma:app:wcpnoherald:2}
    \item $\Pr(H_{\neg i}|\mathcal{S}_{i}) =$ \tab $ 1 - \bigO(\pi_\ell \alpha^k)$    
    \label{lemma:app:wcpnoherald:3}
    \item $\Pr(H_{\neg i}|\mathcal{S}^n_{\neg i} \wedge \mathcal{S}^b_{i} \wedge \mathcal{S}^l_{i}) =$ \tab $ 1 - \bigO(\pi_\ell \alpha^k)$    
    \label{lemma:app:wcpnoherald:4}
    \item $\Pr(H_i|\partial \mathcal{S}^n_{\neg i}) =$ \tab $ 1 - \bigO(\pi_\ell \alpha^k)$    
    \label{lemma:app:wcpnoherald:5}
  \end{enumerate}
\end{lemma}

Look at (\ref{lemma:app:wcpnoherald:1}). The lemma says that the probability for a herald candidate to be created in any single round for any neighborhood of constant radius is at most linear in $\pi_\ell$. Since $\pi_\ell$ is an arbitrarily small constant parameter chosen by us, we can make the probability for this event arbitrarily small. The proof for (\ref{lemma:app:wcpnoherald:1}) is exactly the same as the proof for Lemma 8.6 in \cite{tr:podc2013}, with $\alpha(k)$ replaced by $\alpha^k$. We provide its proof nevertheless, as (\ref{lemma:app:wcpnoherald:2}) and (\ref{lemma:app:wcpnoherald:3}) are new results that do directly depend on (\ref{lemma:app:wcpnoherald:1}). What (\ref{lemma:app:wcpnoherald:2}) says, is, that even if we condition on some nodes $S \subset N^k(u)$ \emph{not} to operate on channel $\Achannel_i$, this does not increase the chance (significantly---i.e., by more than a constant factor) for these or other nodes to receive anything on some channel $\Achannel_j \neq \Achannel_i$. Analogously, (\ref{lemma:app:wcpnoherald:3}) claims that if we condition on some nodes $S \subset N^k(u)$ to \emph{definitely} operate on channel $\Achannel_i$, this does still not affect the probability for any other node to receive any message on some other channel $\Achannel_j \neq \Achannel_i$.
(\ref{lemma:app:wcpnoherald:4}) is a combination of (\ref{lemma:app:wcpnoherald:2}) and (\ref{lemma:app:wcpnoherald:3}), plus we even fix the knowledge of which nodes, that operate on $\Achannel_i$, do broadcast ($S^b$) or listen ($S^l$).
(\ref{lemma:app:wcpnoherald:5}) has already been proven in \cite{tr:podc2013} as Claim 8.9, except that there $k$ was fixed to $2$. The analogous proof is provided as well for sake of completeness.

\iftrue
\begin{proof}
  We first prove (\ref{lemma:app:wcpnoherald:1}). For the whole proof we only use the graph $G_\Astate$ induced by nodes in state \Astate\ in round $r$. We also solely focus on nodes $v$ that do have at least one active neighbor in $G_\Astate$, as isolated nodes cannot become herald candidates. We will use the notation $N_\Astate(v)$ and $N_\Astate^d(v)$ to refer to $N_{G_\Astate}(v)$ and $N_{G_\Astate}^d(v)$, respectively.

  To become a herald candidate, a node $v$ in state $\Astate$ must receive a message from one of its neighbors on one of the channels $\Achannel_1,\dots,\Achannel_{n_\Achannel}$. This is only possible if in round $r$, $v$ chooses to listen on a channel $\Achannel_j$ and exactly one of $v$'s neighbors in $G_\Astate$ broadcasts on channel $\Achannel_j$.

Consider an arbitrary channel $\Achannel_j$ from the herald election channels $\Achannel_1, \dots, \Achannel_{n_\Achannel}$. Let $p_v(j)=2^{-j} \cdot \activity(v)$ be the probability that an active node $v$ chooses to broadcast or listen on channel $\Achannel_j$. In addition, we define $P_v(j) := 2^{-j} \Activity(v) = \sum_{w\in N_\Astate^1(v)} p_w(j)$. Let $B_j^{v,w}$ be the event that $v$ listens on channel $\Achannel_j$, while exactly one of its neighbors $w\in N_\Astate(v)$ transmits on channel $\Achannel_j$ and all other neighbors $w'\in N_\Astate(v)$ are either not on channel $\Achannel_j$ or they choose to listen as well.
  
  \begin{equation}\label{eq:wcpnoherald:1:1}
    \begin{split}
    \Pr(B_j^{v,w}) 
    &= 
    \pi_\ell p_v(j) \cdot (1-\pi_\ell) p_w(j) \cdot \prod_{\mathclap{w' \in N_\Astate(v)\setminus\{w\}}}\big(1-p_{w'}(j)(1-\pi_\ell)\big) \\
    &\leq
    \pi_\ell p_v(j) p_w(j) \cdot \prod_{\mathclap{w'\in \{v,w\}}}
    \frac{1}{1-(1-\pi_\ell)p_{w'}(j)} \cdot
    \prod_{\mathclap{w' \in N_\Astate^1(v)}}\big(1-p_{w'}(j)(1-\pi_\ell)\big)\\
    &\leq
    \pi_\ell p_v(j) p_w(j) \cdot 4\cdot e^{-\frac{1}{2}P_v(j)}\\
    &=
    \pi_\ell 2^{-2j}\activity(v)\activity(w) \cdot 4 e^{-2^{1-j}\Activity(v)}
    \end{split}
  \end{equation}
  In the last inequality, we use that $p_{w'}(j)\leq \frac{1}{2}$ and that $\pi_\ell\leq \frac{1}{2}$. Define $B_j^v$ to be the event that $v$ listens on $\Achannel_j$ and exactly one of its neighbors transmits on that channel. Since $B_j^v=\bigcup_{w\in N_\Astate(v)} B_j^{v,w}$ and the events $B_j^{v,w}$ are disjoint for different $w$, we have
  
  \begin{eqnarray*}
    \Pr(B_j^v)= \sum_{\mathclap{w\in N_\Astate(v)}} \Pr(B_j^{v,w}) \leq \pi_\ell p_v(j) P_v(j) \cdot 4e^{-\frac{1}{2}P_v(j)} =: C_j^v.
  \end{eqnarray*}
  
  For any $x>0$ and constant $c$, $cx^2e^{-x} = \bigO(1)$, which by using $x=P_v(j)$ implies that $C_j^v = \bigO \big(\pi_\ell \frac{p_v(j)}{P_v(j)}\big)=\bigO\big(\pi_\ell \frac{\activity(v)}{\Activity(v)}\big)$ for any fixed $j$.  Next we show that $\sum_{j=1}^{n_\Achannel} C_j^v =\bigO\big(\pi_\ell \frac{\activity(v)}{\Activity(v)}\big)$, too.
\ifthenelse{\boolean{hasLNCSformat}}{
  \begin{equation}
    \begin{split}
      \frac{C_{j+1}^v}{C_j^v} 
      &= \frac{p_v(j+1)}{p_v(j)} \frac{P_v(j+1)}{P_v(j)} e^{-\frac{1}{2}P_v(j+1)+\frac{1}{2}P_v(j)}
      \\
      &= \frac{1}{4}e^{\frac{1}{4}\Activity(v)2^{-j}} < \frac{1}{2} 
      \qquad \forall j 
      \geq \log \Activity(v)
    \end{split}
  \end{equation}
  \begin{equation}
    \begin{split}
      \frac{C_{j}^v}{C_{j+1}^v} 
      &= \frac{p_v(j)}{p_v(j+1)} \frac{P_v(j)}{P_v(j+1)} e^{-\frac{1}{2}P_v(j)+\frac{1}{2}P_v(j+1)} 
      \\
      &= 4e^{-\frac{1}{4}\Activity(v)2^{-j}} < \frac{1}{2} 
      \qquad \forall j 
      \leq \log \Activity(v) -4
    \end{split}
  \end{equation}
}{
  \begin{equation}
      \frac{C_{j+1}^v}{C_j^v} = \frac{p_v(j+1)}{p_v(j)} \frac{P_v(j+1)}{P_v(j)} e^{-\frac{1}{2}P_v(j+1)+\frac{1}{2}P_v(j)}
      = \frac{1}{4}e^{\frac{1}{4}\Activity(v)2^{-j}} < \frac{1}{2} 
      \qquad \forall j 
      \geq \log \Activity(v)
  \end{equation}
  \begin{equation}
      \frac{C_{j}^v}{C_{j+1}^v} 
      = \frac{p_v(j)}{p_v(j+1)} \frac{P_v(j)}{P_v(j+1)} e^{-\frac{1}{2}P_v(j)+\frac{1}{2}P_v(j+1)} 
      = 4e^{-\frac{1}{4}\Activity(v)2^{-j}} < \frac{1}{2} 
      \qquad \forall j 
      \leq \log \Activity(v) -4
  \end{equation}
}

  We can therefore deduce the upper bounds
  $$\sum_{j \geq \log \Activity(v)}C_j^v \leq 2C_{\lceil\log\Activity(v)\rceil}^v \text{ and } \sum_{j \leq \log \Activity(v)} C_j^v \leq 2C_{\max\{1,\lfloor\log\Activity(v)-4\rfloor\}}^v,$$

  \noindent proving the claim that $\sum_{j=1}^{n_\Achannel} C_j^v =\bigO\big(\pi_\ell \frac{\activity(v)}{\Activity(v)}\big)$.

  Using Lemma \ref{lemma:weightedturan}, choosing $G':=G_\Astate[N_\Astate^k(u)]$, $w(v):=\activity(v)$ and $W(v):=\Activity(v)$, we get that $\sum_{v \in G'} \frac{\activity(v)}{\Activity(v)} \leq \alpha(G') \leq \alpha^k$. (Note that the independence number of a graph is larger than or  equal to the independence number of any induced subgraph.)

  Let $B^v$ be the event that $v$ moves from state \Astate\ to $\Hstate'$ and $B=\bigcup_{v\in N^k(u)}=\bigcup_{v\in N_\Astate^k(u)}$. Then,
  \begin{equation*}
    \Pr(B) 
    \leq \sum_{\mathclap{v\in N_\Astate^k(u)}} \Pr(B^v) 
    \leq \sum_{v\in N_\Astate^k(u)} \sum_{j=1}^{n_\Achannel} C_j^v
    = \sum_{\mathclap{v\in N_\Astate^k(u)}} \bigO\left(\pi_\ell \frac{\activity(v)}{\Activity(v)}\right) 
    = \bigO\left(\pi_\ell \alpha^k\right).
  \end{equation*}
  Choosing a sufficiently small $\pi_\ell$ concludes the proof of (\ref{lemma:app:wcpnoherald:1}).


  For (\ref{lemma:app:wcpnoherald:2}) let a proper set $S$ be given. We apply the same calculations as before but we condition now on $S_i$, i.e., that all nodes in $S$ do \emph{not} operate on channel $\Achannel_i$. We let $p_v(j)$ still be the probability that node $v$ operates on channel $\Achannel_j$, but its value might now be different. If $v \notin S$, then this probability is exactly $p_v(j)$ as above, i.e., equals $2^{-j}\activity(v)$. Otherwise, this probability is larger than that by a factor of $\frac{1}{1-2^{-i}}$ due to conditioning on $v$ not operating on $\Achannel_i$. This term maximizes at $i=1$ and then evaluates to $2$, i.e., $p_v(j) \leq 2\cdot 2^{-j}\activity(v)$. The same thought needs to be applied to all nodes and therefore also $P_v(j)$ can change: It now ranges between $2^{-j}\Activity(v)$ and $2^{1-j}\Activity(v)$. Since we are looking for an upper bound in Equation \ref{eq:wcpnoherald:1:1}, we use that $e^{-\frac{1}{2}P_v(j)} \leq e^{-2^{1-j}\Activity(v)}$. In total we get that the right hand side of Equation \ref{eq:wcpnoherald:1:1} increases by a factor at most $4$.

The remainder of the proof is completely analogous to the one for (\ref{lemma:app:wcpnoherald:1}), except that in any summations over all channels we omit channel $\Achannel_i$ (resp. index $i$), which just benefits the cause.


  For (\ref{lemma:app:wcpnoherald:3}) we assume again that a proper set $S$ is given and we condition on $\tilde{S}_i$, i.e., all nodes in $S$ \emph{do} operate on channel $\Achannel_i$. But we are only interested in heralds being created on channels $\Achannel_j \neq \Achannel_i$. We apply the following simple adaption to our initial setup.
  We only focus on the graph $G'_\Astate$ induced by nodes in state \Astate\ that will \emph{not} operate on channel $\Achannel_i$ in round $r$, i.e., we exclude all nodes from $S$, but also nodes outside of $S$ whose local random bits indicate that they choose to operate on $\Achannel_i$ this round. In this case all calculations stay \emph{exactly} the same, except that channel $\Achannel_i$ (resp. index $i$) needs to be removed from all summations, unions (as a matter of fact, the probability of any analyzed event referring to channel $\Achannel_i$ equals zero in $G'_\Astate$) and observations in general. When analyzing values $C^v_j$ the index $i$ needs to be skipped. Finally, the subgraph $G'$ used in the last step must be based on $G'_\Astate$.


  For (\ref{lemma:app:wcpnoherald:4}) let initially $S^n=\emptyset$. Then the statement is the same as in (\ref{lemma:app:wcpnoherald:3}), except that we not only know that nodes in $S$ operate on channel $\Achannel_i$, we also know whether they broadcast (set $S^b$) or listen (set $S^l$). However, this additional knowledge does not affect the proof of (\ref{lemma:app:wcpnoherald:3})---what happens among nodes on other channels than $\Achannel_i$ is completely detached from events on $\Achannel_i$. Let now $S^n \neq \emptyset$, then (\ref{lemma:app:wcpnoherald:4}) is a simple combination of (\ref{lemma:app:wcpnoherald:2}) and this expanded version of (\ref{lemma:app:wcpnoherald:3}).

  
  The proof of (\ref{lemma:app:wcpnoherald:5}) is exactly the same as the proof for Claim 8.9 in \cite{podc2013}, which is equivalent to setting $k=2$.

  In the following, we use the notation $\mathcal{N}:=N_\Astate^k(u,r)\setminus S$. For a node $x\in \mathcal{N}$, let $B_x$ be the event that node $x$ receives a message on channel $\Achannel_i$. Event $B_x$ occurs iff $x$ listens on channel $\Achannel_i$ and exactly one of its neighbors broadcasts on channel $\Achannel_i$. The probability for a node $x\in\mathcal{N}$ to pick channel $\Achannel_i$ is $\activity(x)\cdot 2^{-i}$. We therefore have
\ifthenelse{\boolean{hasLNCSformat}}
{
    \begin{multline*}
      \Pr(B_x|\partial \mathcal{S}_{\neg i})
      \\ 
      = \frac{\pi_\ell\activity(x,r)}{2^i}\sum_{z \in N_\Astate(x)\setminus S}\frac{(1-\pi_\ell)\activity(z,r)}{2^i}\cdot\!\!\!\!
      \prod_{y\in N_\Astate(x)\setminus (S \cup \{z\})} \!\!\left(
        1- \frac{(1-\pi_\ell)\activity(y,r)}{2^i}
      \right) 
      \\
      \leq
      \frac{\pi_\ell\activity(x,r)}{2^i}\frac{\Activity(x,r)}{2^i}\cdot
      e^{-\Activity(x,r)2^{-i}} =
      \bigO\left(\frac{\activity(x,r)}{\Activity(x,r)}\cdot \pi_\ell\right).
    \end{multline*}
}{
  \begin{eqnarray*}
      \Pr(B_x|\partial \mathcal{S}_{\neg i})
      &=&
      \frac{\pi_\ell\activity(x,r)}{2^i}\sum_{z \in N_\Astate(x)\setminus S}\frac{(1-\pi_\ell)\activity(z,r)}{2^i}\cdot\!\!\!\!
      \prod_{y\in N_\Astate(x)\setminus (S \cup \{z\})} \!\!\left(
        1- \frac{(1-\pi_\ell)\activity(y,r)}{2^i}
      \right) 
      \\
      &\leq&
      \frac{\pi_\ell\activity(x,r)}{2^i}\frac{\Activity(x,r)}{2^i}\cdot
      e^{-\Activity(x,r)2^{-i}} =
      \bigO\left(\frac{\activity(x,r)}{\Activity(x,r)}\cdot \pi_\ell\right).
    \end{eqnarray*}
}
  Let $X$ be the number of nodes $x\in \mathcal{N}$ that receive a message on channel $\Achannel_i$ in round $r$. For the expectation of $X$, we then get
  \[
  \E[X|\partial \mathcal{S}_{\neg i}] = \bigO(\pi_\ell)\cdot 
  \sum_{x\in \mathcal{N}}\frac{\activity(x,r)}{\Activity(x,r)} = \bigO(\pi_\ell).
  \]
  The second equation follows from Lemma \ref{lemma:heavyneighborhoods} because the graph induced by $\mathcal{N}$ has independence at most $\alpha^k$. Applying the Markov inequality, we get $\Pr(X\geq1|\partial \mathcal{S}_{\neg i})\leq \E[X|\partial \mathcal{S}_{\neg i}]=\bigO(\pi_\ell)$, which concludes the proof of (\ref{lemma:app:wcpnoherald:5}).
\qedLNCS
\end{proof}
\fi


\subsection{Lemma \ref{lemma:fatgivesherald}} 

\begin{lemma}
  \label{lemma:app:fatgivesherald}
  Let $t$ be a round in which for a node $u$ in state \Astate\ in the \hfilter the following holds:
  \begin{compactitem}
    \item there is no herald candidate in $N^2(u)$,
    \item all nodes $v \in N^2(u)$ that neighbor a herald or leader, have $\activity(v) \leq \activity_\low$,
    \item all nodes in $N^2(u)$ neighboring MIS nodes are eliminated,
    \item $\Activity(u) \geq 1$,
  \end{compactitem}
  If in addition it holds that either 
  \begin{compactenum}[\bfseries (a)]
    \item $\Activity(u) <5\alpha$, $u$ is $\frac{1}{5\alpha}$-fat and $\activity(u)=\frac{1}{2}$, or
    \item $\Activity(u) \geq 5\alpha$ and $u$ is $\eta$-fat.
  \end{compactenum}
  Then by the end of round $t' \in [t, t+7]$, with probability $\Omega(\pi_\ell)$ either a node in $N^2(u)$ joins the MIS or a good pair $(l,h)\in (\Lstate \cap N^1(u)) \times (\Hstate \cap N^1(u))$ is created. 
\end{lemma}

This proof is an adaption to the following lemma, from \cite{tr:podc2013}:

\begin{lemma}
\label{lemma:orig:fatgivesherald}
  Let $t$ be a round in which for a node $u$ in state $\Astate$ in the
  herald filter it holds that there is no herald, leader, or herald
  candidate in $N^2(u)$. Furthermore, all neighbors of MIS nodes in
  $N^2(u)$ are in state $\Estate$, $\Activity(u) \geq 1$, and either
  \setlength{\pltopsep}{0.6ex}
  \begin{compactenum}[\bfseries (a)]
    \item $\Activity(u) < 3\alpha(1)$, $u$ is $\frac{1}{3\alpha(1)}$-fat,
      and $\activity(u) = \frac{1}{2}$, or
    \item $u$ is $\frac{1}{2}$-fat and $\Activity(u) \geq 3\alpha(1)$.
  \end{compactenum}
  Then by 
round $t' \in [t, t+7]$, with
  probability $\Omega(\pi_\ell)$ either a node in $N^2(u)$ joins the
  MIS or a pair $(l,h)\in \Lstate\times\Hstate$ is
  created in $N^1(u)$ such that
  $(\!N(\{l,h\})\setminus\{l,h\}) \cap
  (\Hstate' \!\cup \Hstate \cup \Lstate) \!=\! \emptyset$. 
\end{lemma}

Note that we do not exclude the existence of leaders and heralds anymore. Instead we require such nodes to have 'silenced' their neighborhood, i.e., to have drove back such nodes' activity values. Good pairs are supposed to do so with a good probability, and while bad pairs can manage to do that, too (with a low probability), it does not affect the lemma's correctness if they do. Last, the probability for a new isolated leader-herald pair goes down a bit, but is still in $\Omega(\pi_\ell)$, yet the hidden constant now depends heavily on $\eta$. One of the main reason to change $3\alpha$ from the original lemma to $5\alpha$ was to ease argumentation in one part and to differentiate it from $n_\Rchannel \geq 3\alpha^2$---the value in this proof is completely unrelated to $n_\Rchannel$.

The proof follows mostly the lines of the one in \cite{tr:podc2013}, but in that paper a constant number of factors of $1/2$ accumulate, while here they have to be exchanged by factors of $\eta$---but this only causes changes in asymptotically irrelevant constants. In addition, nodes neighboring leaders or heralds need special attention. We give a full detailed proof here, despite its similarity to the proof in \cite{tr:podc2013}.

\iftrue
\begin{proof}
  We make use of the notation $X(u,t)$ to indicate the value of local variable $X$ at node $u$ in round $t$.
  For the remainder of the proof, assume that
  in rounds $t,\dots,t+7$, no node in $N^2(u)$ joins the MIS, as
  otherwise, the claim of the lemma is trivially satisfied.  In order
  to prove the lemma, we first show that either in round $t$ or in
  round $t+1$, \wcp, a herald candidate is created in
  $N^1(u)$. Formally, we define the event $H_u$ as follows. In round
  $t'$, event $H_u$ occurs iff there are two neighboring nodes $v,w\in
  N_\Astate^1(u,t') \setminus N(\Lstate \cup \Hstate)$ such that
  \begin{itemize}
  \item $v$ and $w$ both operate on a channel $\lambda\in
    \set{\Achannel_1,\dots,\Achannel_{n_\Achannel}}$
  \item no other neighbor of $v$ and $w$ chooses channel $\lambda$, and
  \item no other node in $N_\Astate^2(u,t')$ receives a message on
    channel $\lambda$.
  \end{itemize}
  Clearly, if event $H_u$ holds either in round $t$ or $t+1$, the
  nodes $v,w$ have a probability of $2\pi_\ell(1-\pi_\ell)$ of
  becoming a herald-leader candidate pair and no other herald
  candidate is created on channel $\lambda$ in that round. Combined
  with appropriate applications of Lemma \ref{lemma:app:wcpnoherald}.(\ref{lemma:app:wcpnoherald:1}), this
  suffices to prove the claim of the lemma.

  For a node $v\in N_\Astate^1(u,t')$, let
  $\Activity^u_\Astate(v,t'):=\sum_{w\in N^1(v,t')\cap
    N_\Astate^1(u,t')}\activity(w,t')$ be the total activity value of
  all active nodes in round $t'$ in the $1$-neighborhood of $v$
  restricted to the $1$-neighborhood of $u$. To estimate the
  probability that $H_u$ occurs in a round $t'\in \set{t,t+1}$, we
  first show that in one of the two rounds $t'\in \set{t,t+1}$, with
  probability at least $\frac{1}{4}$ it holds that $u$ is in state $\Astate$
  and $\Activity_\Astate(u,t'):=\sum_{v\in
    N_\Astate^1(u,t')}\activity(v,t')\geq\frac{3}{5}\cdot\Activity(u,t)$; i.e., in one of both rounds, we have a high activity mass provided by nodes in state $\Astate$ as opposed to those in state $\Lstate'$.
  Assume that the claim is not true for $t'=t$. 
  As the lemma statement
  is based on the assumption that $u$ is in state $\Astate$ in round
  $t$, this implies that
  $\Activity_\Astate(u,t) < \frac{3}{5}\Activity(u,t)$.
  Also
  by the assumptions of the lemma, in round $t$, no nodes in $N(u)$
  are in states $\Hstate'$, and those in $\Hstate$ or $\Lstate$ have an activity of less than $\activity_\low$. 
  Since $|N^2(u) \cap N(\Lstate \cup \Hstate)| \leq \Delta_\textmax^2 \leq \log^{2\ttkdelta} n \ll \activity_\low^{-1}$, we can assume that the total activity mass of nodes in $N^2(u) \cap N(\Lstate \cup \Hstate)$ is less than $1/100$. As nodes $w$ in
  states $\Mstate$ and $\Estate$ have $\activity(w)=0$ and thus do not
  contribute to $\Activity(u)$, we therefore have
  \[
  \Activity_{\Lstate'}(u,t):=\hspace{-3mm}\sum_{v\in N_{\Lstate'}^1(u,t)}\hspace{-3mm}
  \activity(v,t)\geq \Activity(u,t)-\Activity_\Astate(u,t)-\frac{1}{100} \stackrel{\Activity(u,t)\geq 1}{\geq} \frac{99}{100}\Activity(u,t)-\Activity_\Astate(u,t).
  \]
  Because
  by assumption, there are no nodes in state $\Hstate'$ in round $t$,
  all nodes that are in state $\Lstate'$ in round $t$ switch back to
  state $\Astate$ for the next round. As by assumption, no nodes
  switch to states $\Mstate$ or $\Estate$, and a node $v$ that is in
  state $\Astate$ in round $t$ can only move out of $\Astate$ if it
  decides to operate on one of the channels
  $\Achannel_1,\dots,\Achannel_{n_\Achannel}$. This happens with
  probability at most $\activity(v,t)\leq \frac{1}{2}$. Therefore, with
  probability at least $\frac{1}{2}$, at least half of the total activity
  value of the nodes in $N_\Astate^1(u,t)$ remains in state $\Astate$
  for round $t+1$. And (independently) with probability at least
  $\frac{1}{2}$, also node $u$ remains in state $\Astate$ for round
  $t+1$. Thus, with probability at least $\frac{1}{4}$, $u$ is in state
  $\Astate$ in round $t+1$ and at least half of the total activity
  contributing to $\Activity_\Astate(u,t)$ also contributes to
  $\Activity_\Astate(u,t+1)$. 
  Therefore, with probability at least
  $\frac{1}{4}$,
  \begin{equation*}
    \begin{split}
      \Activity_\Astate(u,t+1) 
      &\stackrel{(*)}{\geq}
      \left(\frac{98}{100}\Activity(u,t)-\Activity_\Astate(u,t)\right) +
      \frac{1}{2}\cdot\Activity_\Astate(u,t)
      \\
      &\stackrel{\Activity(u,t)\geq 1}{\geq}
      \frac{98}{100}\Activity(u,t) -
      \frac{1}{2}\cdot\Activity_\Astate(u,t) \geq
      \frac{3}{5}\Activity_\Astate(u,t).
    \end{split}
  \end{equation*}
  The last inequality follows because we assumed that
  $\Activity_\Astate(u,t)<\frac{3}{5}\Activity(u,t)$. In $(*)$ we used the
  fact that the activity of all nodes in $N^1(u)$ can only grow from round $t$ to $t+1$, except for those that neighbor heralds or leaders and those that switch to states $\Mstate$ or $\Estate$. But the former only make up a small percentage of $u$'s activity mass and the latter do not exist by our assumptions. We therefore in the following assume that
  $t'\in\set{t,t+1}$ such that $\Activity_\Astate(u,t')\geq
  \frac{3}{5}\cdot \Activity(u,t)$ and $u$ is in state $\Astate$ in
  round $t'$. 

  To show that in round $t'$, event $H_u$ occurs, we distinguish the
  two cases given in the lemma statement. We start with the simpler
  case (a), where in round $t$, $1\leq \Activity(u)<5\alpha$ and
  $\activity(u)=\frac{1}{2}$. The latter implies that $u \notin N(\Lstate \cup \Hstate)$, due to the requirements of the lemma that such nodes have low activity. Because no node in $N^1(u)$ switches to states
  $\Mstate$ or $\Estate$ in round $t$ and neighbors of leaders and heralds have low activity, $u$'s activity mass cannot change much. Neighbors of leaders and heralds cause a drop of at most $\Activity(u,t)/100$ and all others increase their activity by at most $\increaseactivity$. More precisely, 
  $0.99\Activity(u,t)\leq \Activity(u,t')\leq
  (1+\epsilon_\activity)\Activity(u,t)$. We know that
  $\Activity_\Astate(u,t')\geq \frac{3}{5}\Activity(u,t) \geq
  \frac{3}{5}$. Consequently, since $u$ is in state $\Astate$, it has
  activity level $\activity(u,t')=\frac{1}{2}$, and the total activity mass
  $\Activity_\Astate(u,t')-\activity(u,t')$ of all \emph{neighbors} is between $\frac{1}{10}$ and
  $\increaseactivity 5\alpha=\bigO(1)$. Therefore, \wcp, $u$ and
  exactly one of its neighbors $v$ operate on channel
  $\lambda=\Achannel_1$. (Recall that a node $w$ in state $\Astate$
  chooses channel $\Achannel_1$ with probability $\frac{\activity(w)}{2}$.)
  Because we assume that $u$ is 
  $\frac{1}{5\alpha}$-fat 
  at time $t$,
  $\Activity(v,t)$ is also bounded and therefore, \wcp, no other
  neighbor of $v$ picks channel $\Achannel_1$. Hence, the only thing that is
  missing to show that event $H_u$ occurs with constant probability is
  to prove that no other node in $N^2(u)$ hears a message on channel
  $\Achannel_1$ in round $t'$. 

  But this follows from Lemma \ref{lemma:app:wcpnoherald}.(\ref{lemma:app:wcpnoherald:5}) by choosing $S=N_\Astate^2(u)\cap (N(u)\cup N(v))$.
We have shown that in case (a), the
  event $H_u$ occurs with constant probability. Let us therefore
  switch to case (b), where $N(u,t)$ is $\eta$-fat and $\Activity(u,t)
  \geq 5\alpha$. For the following argumentation, we define
  \[
  \hat{N}_\Astate^1(u,t'):=\set{v\in N_\Astate^1(u,t') \setminus N(\Lstate \cup \Hstate):
    \Activity^u_\Astate(v,t')\geq
    \frac{\Activity_\Astate(u,t')}{2\alpha}} 
  \]
  and
  \[
  \hat{\Activity}_\Astate(u,t'):=\sum_{\mathclap{v\in
    \hat{N}_\Astate^1(u,t')}}\activity(v,t').
  \]

  To analyze the probability of the event $H_u$, consider two
  neighboring nodes $v,w\in N^1_\Astate(u,t') \setminus N(\Lstate \cup \Hstate)$. We define $L_{v,w}$ to
  be the event that in round $t'$ both $v$ and $w$ decide to operate
  on channel $\lambda:=\lceil\log_2 \Activity(u,t)\rceil$ and no other
  node in $N(v)\cup N(w)$ chooses the same channel $\lambda$. Further
  $H_{v,w}$ is the event that $L_{v,w}$ occurs and in addition, no
  node in $N^2(u)\setminus (N(v)\cup N(w))$ receives a message on
  channel $\lambda$ in round $t'$. Lemma
  \ref{lemma:app:wcpnoherald}.(\ref{lemma:app:wcpnoherald:5}) implies again that
  $\Pr(H_{v,w}|L_{v,w})=1-\bigO(\pi_\ell)$.  Further, note that
  $H_u = \bigcup_{v,w\in N_\Astate^1(u) \setminus N(\Lstate \cup \Hstate), \set{v,w} \in E} H_{v,w}$,
  and we have $H_{v,w}=H_{w,v}$ and $H_{v,w}\cap H_{v',w'}=\emptyset$
  for $\set{v,w}\neq\set{v',w'}$. It therefore holds that
  \begin{equation}\label{eq:Hevent}
    \Pr(H_u) = \frac{1-\bigO(\pi_\ell)}{2}\ \ \cdot
    \sum_{
      \mathclap{
        \substack{
          \set{v,w}\in E,\\
          (v,w)\in (N_\Astate^1(u,t'))^2
        }
      }
    }  \ \ \ \Pr(L_{v,w}).
  \end{equation}
  The probability for a node $v\in \Astate$ to choose channel
  $\lambda$ is $\activity(v)\cdot
  2^{-\lceil\log\Activity(u,t)\rceil}\in
  \big[\frac{\activity(v)}{2\Activity(u,t)},\frac{\activity(v)}{\Activity(u,t)}\big]$. We can therefore bound the
  probability that $L_{v,w}$ occurs in round $t'$ as
  \begin{eqnarray*}
    \Pr(L_{v,w}) & \geq &
    \frac{1}{4}\cdot 
    \frac{\activity(v,t')\activity(w,t')}{\Activity(u,t)^2}\cdot
    \prod_{x\in N(v)\cup N(w)}\left(1-\frac{\activity(x,t')}{\Activity(u,t)}\right)\\
    & \geq &
    \frac{\activity(v,t')\activity(w,t')}{4\Activity(u,t)^2}\cdot
    4^{-\sum_{x\in
      N(v)\cup N(w)}\frac{\activity(x,t')}{\Activity(u,t)}}\\
    & \geq &
    \frac{\activity(v,t')\activity(w,t')}{4\Activity(u,t)^2}\cdot
    4^{-\increaseactivity\frac{\Activity(v,t)+\Activity(w,t)}{\Activity(u,t)}}\\
    & \geq &
    \frac{\activity(v,t')\activity(w,t')}{\Activity(u,t)^2}\cdot
    4^{-1-\increaseactivity\frac{2}{\eta}\frac{\Activity(u,t)}{\Activity(u,t)}}
    \  =\ 
    \frac{1}{4\cdot 16^{\frac{\increaseactivity}{\eta}}}
    \cdot
    \frac{\activity(v,t')\activity(w,t')}{\Activity(u,t)^2}.
  \end{eqnarray*}
  The last inequality follows because in round $t$, node $u$ is
  $\eta$-fat, and therefore 
  \[
  \Activity(u,t) \geq \max\set{\Activity(v,t),\Activity(w,t)}.
  \]
  In the following, we restrict our attention to the
  events $L_{v,w}$ for $v\in \hat{N}_\Astate^1(u,t')$ as these are the
  only ones for which we obtain a significant lower bound on the
  probability that they occur. For $v\in \hat{N}^1_\Astate(u,t')$, let
  $K_v:=\bigcup_{w\in N(v)\cap N^1_\Astate(u,t')\setminus N(\Lstate \cup \Hstate)}L_{v,w}$ be the event
  that $L_{v,w}$ occurs for some neighbor $w$ of $v$. For a node $v\in
  \hat{N}_\Astate^1(u,t')$, we then have 
  \begin{eqnarray} \nonumber
    \Pr(K_v) & = &  \sum_{\mathrlap{\!\!\!\!\!\! w\in N(v)\cap N^1_\Astate(u,t')\setminus N(\Lstate \cup \Hstate)}}
    \ \ \Pr(L_{v,w})
    \\ \nonumber 
    & \geq &
    \frac{1}{4\cdot 16^{\frac{\increaseactivity}{\eta}}}
    \cdot\frac{\activity(v,t')}{\Activity(u,t)^2}
    \sum_{\mathrlap{\!\!\!\!\!\! w\in N(v)\cap N^1_\Astate(u,t')\setminus N(\Lstate \cup \Hstate)}} \activity(w,t')\\
    \nonumber 
    & \geq & 
    \frac{1}{4\cdot 16^{\frac{\increaseactivity}{\eta}}}\cdot
    \frac{\activity(v,t')\big(\Activity^u_\Astate(v,t')-\activity(v,t')-\frac{1}{100}\big)}{\Activity(u,t)^2}\\
    \label{eq:Kvstep}
    & \geq &
    \frac{1}{4\cdot 16^{\frac{\increaseactivity}{\eta}}}\cdot
    \frac{\activity(v,t')\big(\Activity_\Astate(u,t')-1.02\alpha\big)}
    {2\alpha\Activity(u,t)^2}\\
    \nonumber
    & \geq &
    \frac{1}{4\cdot 16^{\frac{\increaseactivity}{\eta}}}\cdot
    \frac{\activity(v,t')}{4\alpha\Activity(u,t)}.
  \end{eqnarray}
  Inequality \eqref{eq:Kvstep} follows because $\activity(v,t')\leq
  \frac{1}{2}$, $v\in\hat{N}_\Astate^1(v)$ and thus
  $\Activity^u_\Astate(v,t')\geq
  \frac{\Activity_\Astate(u,t')}{2\alpha}$. The last inequality follows
  from $\Activity_\Astate(u,t')\geq \frac{3}{5}\cdot\Activity(u,t)$ and 
  thus $\Activity_\Astate(u,t') - 1.02\alpha \geq \frac{3}{5}\cdot\Activity(u,t) - 1.02\alpha \geq \frac{1}{4} \Activity(u,t)$ if $\Activity(u,t)\geq 5\alpha$. Using \eqref{eq:Hevent}, we can now bound the probability of event $H_u$ in round $t'$ as
  \begin{equation}
    \label{eq:Hevent2}
    \begin{split}
      \Pr(H_u) 
      &\geq 
      \frac{1-\bigO(\pi_\ell)}{2}\!\!\!\!\!\!\sum_{v\in
        \hat{N}_\Astate^1(u,t')} \!\!\!\!\!\!\Pr(K_v) \geq
      \frac{1-\bigO(\pi_\ell)}{64\alpha \cdot
        16^{\frac{\increaseactivity}{\eta}}} \sum_{v\in
        \hat{N}_\Astate^1(u,t')}\!\!\!\!
      \frac{\activity(v,t')}{\Activity(u,t)} 
      \\
      &=
      \frac{1-\bigO(\pi_\ell)}{64\alpha\cdot
        16^{\frac{\increaseactivity}{\eta}}}\cdot
      \frac{\hat{\Activity}_\Astate(u,t')}{\Activity(u,t)}.
    \end{split}
  \end{equation}
  Applying Lemma \ref{lemma:heavyneighborhoods} to the graph induced
  by the nodes in $N_\Astate^1(u,t')$, the activity sum of nodes in
  $\hat{N}_\Astate^1(u,t')$ can be lower bounded as
  \[
  \hat{\Activity}_\Astate(u,t')\geq
  \frac{\Activity_\Astate(u,t')}{2\alpha}\geq
  \frac{3\Activity(u,t)}{10\alpha}.
  \]
  Together with \eqref{eq:Hevent2}, this proves that also in case (b),
  the event $H_u$ occurs with constant probability in a round $t'\in
  \set{t,t+1}$. Note also that in both cases (a) and (b), for
  $\pi_\ell$ and $\increaseactivity$ sufficiently small, the
  probability that $H_u$ occurs can be lower bounded by a constant $C_\eta$
  that is independent of the probabilities
  $\pi_\ell$ and $\increaseactivity$.

  To complete the proof, assume that in round $t'$, event $H_u$ occurs
  with probability $C_\eta$ and if it occurs, nodes $v$ and $w$ are the two
  nodes in $N^1(u)$ participating on channel $\lambda$ (channel
  $\Achannel_1$ in case (a)). Let $M$ be the event that no herald is
  created on a channel $\Achannel_i\neq \lambda$ in round
  $t'$. Clearly, the probability that $M$ occurs is lower bounded by
  the probability that no herald is created on any channel in round
  $t'$. By Lemma \ref{lemma:app:wcpnoherald}.(\ref{lemma:app:wcpnoherald:1}), we therefore have
  $\Pr(M)=1-\bigO(\pi_\ell)$. For the probability that events $H_u$ and
  $M$ both occur, we then get
  \[
  \Pr(M\cap H_u) = 1-\Pr(\overline{M}\cup \overline{H}_u) \geq
  1 - \Pr(\overline{M})-\Pr(\overline{H}_u) = C_\eta - \bigO(\pi_\ell).
  \]
  Recall that probability $C_\eta$ is a constant independent of
  $\pi_\ell$. Conditioned on the event that $M\cap H_u$ occurs, the
  probability that one of the two nodes $v,w$ listens on channel $\lambda$
  and the other one broadcasts on the channel is
  $2\pi_\ell(1-\pi_\ell)$. In that case one of the two nodes becomes a
  herald candidate and the other one its leader candidate. Also,
  $M\cap H_u$ implies that in round $t'$ no other herald candidates
  are created in $N^2(u)$. Let $t''$ be the round in
  $\set{t,t+1}\setminus{t'}$. If in addition in round $t''$ and in the
  remaining rounds $t+2,\dots,t+7$ no herald candidate is created in
  $N^2(u)$, nodes $v$ and $w$ make it through the handshake and become
  an isolated leader-herald pair as claimed by the lemma. By Lemma
  \ref{lemma:app:wcpnoherald}.(\ref{lemma:app:wcpnoherald:1}), this happens with probability
  $1-\bigO(\pi_\ell)$, which by choosing $\pi_\ell$ sufficiently small
  concludes the proof.
  \qedLNCS
\end{proof}
\fi


\hide{

  \ifthenelse{\boolean{swapheraldclaim}}{
    \subsection{\Cref{lemma:heraldbounds}}
  }{
    \subsection{\Cref{lemma:heraldwcpgoodherald}}
  }
  \ifthenelse{\boolean{swapheraldclaim}}{
    \begin{lemma}
      \label{lemma:app:heraldbounds}
      Let $r$ be a round in which node $u$ is in state \Astate and $N_\Astate(u) \neq \emptyset$.
      Let $B^u$ be the event that at the end of round $r$, $u$ moves to state
      $\Hstate'$ due to receiving a message from some node $v \in
      N_\Astate(u)$ on some channel $\Achannel_{\bar{\lambda}}$. Further, let $D^u\subseteq B^u$ be the event that
      $B^u$ holds and in addition no other node $v' \in N^3(v)\setminus\set{u}$
      receives any message on channel $\Achannel_{\bar{\lambda}}$ in round $r$. 
      It holds that
      \begin{eqnarray}
        \label{eq:app:claimequationsupper}
        \Pr(B^u) &=&
        \begin{cases}
          \bigO\brackets{\pi_\ell \frac{\activity(u)}{\AAC(u)}} & \AAC(u) > 2 \\
          \bigO\brackets{\pi_\ell \activity(u) \AAC(u)} & \AAC(u) \leq 2      
        \end{cases},
        \\
        \label{eq:app:claimequationslower}
        \Pr(D^u) &=&
        \begin{cases}
          \Omega\brackets{\pi_\ell \frac{\activity(u)}{\AAC(u)}} & \AAC(u) > 2 \\
          \Omega\brackets{\pi_\ell \activity(u) \AAC(u)} & \AAC(u) \leq 2      
        \end{cases}.
      \end{eqnarray}
    \end{lemma}
  }{
    \begin{lemma}
      \label{lemma:app:heraldwcpgoodherald}
      Let $B(r)^{u,v}$ be the event that in round $r$ node $u \in \Astate_r$ receives a message from one of its neighbors $v\in \Astate_r$, neither of them neighboring any leader, herald or herald candidate in the $5$th or $6$th round of its handshake protocol. Let $\hat H(r')^{u',v'}$ be the event that at the beginning of round $r'$ node $u' \in \Hstate_{r'}$ and $v' \in \Lstate_{r'}$ form a good \lhp and that $\Hstate' \cap N^3(u') = \emptyset$, i.e., there are no herald candidates in the $3$-neighborhood of $u'$. Then
      \begin{equation}
        \Pr(\hat H(r+8)^{u,v}|B(r)^{u,v}) = \Omega(1)
      \end{equation}
    \end{lemma}
  }

} 




\end{document}